\documentclass[10pt,prx,aps,floatfix,amsmath,superscriptaddress,twocolumn,longbibliography,nofootinbib]{revtex4-2}

\usepackage{amssymb,enumerate}
\usepackage{graphicx}
\usepackage{graphics}
\usepackage{amsmath}
\usepackage{amsthm,bbm}
\usepackage{thmtools,thm-restate}
\usepackage{color}
\usepackage{xcolor}
\usepackage{dsfont}
\usepackage{natbib}
\usepackage{multirow}
\usepackage[braket]{qcircuit}
\usepackage{wrapfig, blindtext}
\usepackage[most]{tcolorbox}
\tcbset{colback=yellow!10!white, colframe=red!50!black, 
  highlight math style= {enhanced,
    colframe=red,colback=red!10!white,boxsep=0pt}
}
\usepackage{amsmath}
\usepackage{amssymb}
\usepackage{amsfonts}
\usepackage{amsthm}
\usepackage{empheq}
\usepackage{graphicx} 
\usepackage{epsfig}
\usepackage{anyfontsize}
\usepackage{nameref}
\usepackage{xpatch}
\makeatletter
\xpatchcmd{\@ssect@ltx}{\@xsect}{\protected@edef\@currentlabelname{#8}\@xsect}{}{}
\xpatchcmd{\@sect@ltx}{\@xsect}{\protected@edef\@currentlabelname{#8}\@xsect}{}{}
\makeatother
\usepackage{youngtab}

\usepackage{tikz}
\usetikzlibrary{decorations.pathreplacing}
\usepackage{tikz-cd}
\usepackage{ytableau}
\usepackage{makecell}
\usepackage{enumerate}
\usepackage{bbm}
\usepackage{color}
\usepackage{dsfont}
\usepackage{lipsum}
\usepackage{lmodern}
\usepackage{tcolorbox}
\usepackage{epsfig,times,bm,mathrsfs}
\usepackage[normalem]{ulem}
\usepackage{subfigure}
\usepackage{dsfont}
\usepackage{braket}
\usepackage{upgreek}
\usepackage{natbib}
\usepackage{chngcntr}
\usepackage{algorithm}
\usepackage{algpseudocode}
\usepackage{mathtools}
\usepackage{microtype}
\ytableausetup{smalltableaux}
\usepackage{hyperref}
\usepackage[capitalise]{cleveref}

\newcommand{\ifmaintext}[1]{#1}

\renewcommand{\i}{\mathrm{i}}

\newcommand{\tr}{\operatorname{tr}}

\makeatletter 
\renewcommand\onecolumngrid{ 
  \do@columngrid{one}{\@ne}
  \def\set@footnotewidth{\onecolumngrid}
  \def\footnoterule{\kern-6pt\hrule width 1.5in\kern6pt}
}
\makeatother 
\newlength\dlf

\newtheorem{theorem}{Theorem}

\newtheorem{lemma}{Lemma}
\newtheorem*{lemma*}{Lemma}

\theoremstyle{remark}

\newcommand{\bes} {\begin{subequations}}
\newcommand{\ees} {\end{subequations}}
\newcommand{\bea} {\begin{eqnarray}}
\newcommand{\eea} {\end{eqnarray}}
\newcommand{\be} {\begin{equation}}
\newcommand{\ee} {\end{equation}}

\def\>{\rangle}
\def\<{\langle}
\def\Tr{\operatorname{Tr}}
\def\Pr{\textrm{Pr}}
 
\newcommand{\abs}[1]{\lvert #1 \rvert}

\newcommand{\ignore}[1]{}
\newcommand{\x}[1]{\textcolor{red}{\sout{#1}}}

\crefname{section}{Sec.}{Secs.}
\crefname{claim}{Claim}{Claims}

\newtcolorbox{subroutinebox}[2][]{%
    enhanced,
    breakable,
    colback=white,
    colframe=black,
    fonttitle=\bfseries,
    title={Subroutine #2},
    boxrule=0.6pt,
    arc=1mm,
    left=2mm,
    right=2mm,
    top=1mm,
    bottom=1mm,
    #1
}
\newcommand{\appendixentry}[4]{%
    \noindent
    $\bullet$ \textbf{Appendix #1: #2}
    \dotfill \pageref{#3}

    \smallskip
    \hspace{1.5em}
    \parbox{\dimexpr\textwidth-1.5em\relax}{#4}

    \vspace{0.8em}
}
\definecolor{apptocred}{RGB}{120,20,10}

\begin{document}

\title{Unitary Schur Sampling of Qudits via Random SWAP Tests: Hunt for Antisymmetry}

\author{Shrigyan Brahmachari}

\email{shrigyan.brahmachari@duke.edu}
\affiliation{Duke Quantum Center, Duke University}
\affiliation{Department of Electrical and Computer Engineering, Duke University}
\author{David Jakab}
\affiliation{Duke Quantum Center, Duke University}
\affiliation{Department of Physics, Duke University}
\author{Henry D. Pfister}
\affiliation{Duke Quantum Center, Duke University}
\affiliation{Department of Electrical and Computer Engineering, Duke University}
\author{Iman Marvian}
\email{iman.marvian@duke.edu}
\affiliation{Duke Quantum Center, Duke University}
\affiliation{Department of Electrical and Computer Engineering, Duke University}
\affiliation{Department of Physics, Duke University}
\begin{abstract}
Schur sampling is a fundamental primitive for extracting
permutation-invariant information from many-body quantum systems
and has broad applications in quantum information science. 
Circuit-based implementations typically rely on coherent
representation-theoretic operations such as Clebsch-Gordan
transforms, generalized phase estimation, or quantum Fourier
transforms over the symmetric group. We show that unitary Schur
sampling of arbitrary permutation-invariant mixed states on
$n$ $d$-dimensional qudits can instead be implemented using
only pairwise SWAP tests, namely, two-qudit projective
measurements onto the symmetric and antisymmetric subspaces.
Our protocol achieves error $\epsilon$  in diamond distance
using
 $O\!\left(n\min\{n,d\}^{3}\log^2(n/\epsilon)\right)$ 
random SWAP tests.
The central idea is to repeatedly identify and extract the
largest antisymmetric subsystem. This turns unitary Schur
sampling into a search for antisymmetry and gives the algorithm
a natural interpretation as a stochastic traversal of a Young
diagram. The protocol preserves the $\mathrm{SU}(d)$-irrep
state while preparing a canonical pure state in the multiplicity
subsystem.  
As an application, we consider quantum purity amplification (QPA), in which multiple noisy copies of a pure quantum state are combined to produce a state of higher purity. We show that the optimal purified state can be obtained by retaining a single designated output qudit and discarding the rest.

\end{abstract}

\maketitle

\section{Introduction}

Permutation symmetry provides a natural organizing principle for quantum
systems composed of many identical subsystems. When the ordering of the
subsystems is irrelevant, Schur-Weyl duality gives a canonical description
of the information contained in the state, decomposing the \(n\)-qudit
Hilbert space into sectors labelled by Young diagrams
\cite{harrowthesis}. This structure underlies a variety of
quantum-information tasks, including state estimation
\cite{TomographyHarrow,TomographyDonnel},
spectrum estimation
\cite{spectrumchristandl,SpectrumWerner,spectrumodonnell},
purification \cite{Cirac,QPA1,QPA2}, coherence distillation
\cite{Marvian2020Coherence}, entanglement concentration
\cite{entanglement1,entanglement2}, and quantum majority vote
\cite{majority}. Permutation-invariant information is also central to
quantum metrology
\cite{helstrom1969quantum,holevo2011probabilistic,giovannetti2011advances}.

Schur sampling provides a way to access this structure. Weak Schur
sampling returns the Young-diagram label, while \emph{unitary Schur
sampling} (USS) also retains the quantum state on the corresponding
unitary-group irrep subsystem \cite{cervero2024memory}.
Existing circuit implementations of Schur transforms and related sampling
tasks use representation-theoretic tools such as Clebsch-Gordan
transforms \cite{BCH,cervero2023weakschursamplinglog,cervero2024memory},
generalized phase estimation \cite{harrowthesis}, or quantum Fourier
transforms over the symmetric group \cite{Krovi_2019,HighDimSchur}.
This motivates a simpler question: how much of this machinery is
necessary if the goal is only to extract the permutation-invariant
information of the state?

For qubits, Brahmachari et al.~\cite{qubitSWAPUSS} showed that random
pairwise SWAP tests suffice to implement USS by identifying and
separating singlet pairs. SWAP tests also appear in streaming
purification protocols for qudits
\cite{Childs2025Streaming,Grier2025Streaming}. Extending the
singlet-extraction approach to general local dimension requires
identifying antisymmetric subsystems of more than two qudits while
preserving the quantum information needed for USS.

In this work, we show that unitary Schur sampling can be realized to
arbitrary accuracy using only pairwise SWAP measurements and classical
control, for arbitrary numbers of qudits and arbitrary local dimension.
Our approach is based on antisymmetry: we repeatedly identify and
separate a largest antisymmetric subsystem and continue with the
remaining qudits. The extracted subsystem sizes determine the column
lengths of the Young diagram. Since the SWAP-test projectors commute
with the collective action of \(\mathrm{SU}(d)\), they act trivially on
the irrep subsystem within each sector. This allows the protocol to
resolve the symmetry structure while preserving the irrep state.
The resulting SWAP-test complexity is linear in the number of qudits
up to logarithmic factors for fixed local dimension.
Table~\ref{tab:USScomparison} places this construction alongside
existing implementations.

A particularly natural application is quantum purity amplification
(QPA), which processes multiple copies of a mixed state to improve
the fidelity with a chosen eigenstate \cite{QPA1,QPA2}.
Our protocol prepares a canonical representative of the retained
irrep state directly in the original qudit registers. For the
single-output QPA protocol of Ref.~\cite{QPA1}, that targets the principal eigenstate this combines
Schur sampling and the required correction within each symmetry
sector into a single procedure.

\begin{table*}[t]
\centering
\caption{
Comparison of implementations capable of unitary Schur sampling.
Complexities are shown up to logarithmic factors. The complexity measures
follow those used in the respective references; our bound counts pairwise
SWAP tests. The notation \(\widetilde O\) suppresses logarithmic factors, including
accuracy dependence. The circuit bounds count elementary gates,
whereas the last two rows count pairwise SWAP tests.
}
\label{tab:USScomparison}

\begin{tabular}{lccc}
\hline
\textbf{Method}
& \textbf{Main primitive}
& \textbf{Complexity}
& \textbf{Remarks}
\\
\hline

Bacon-Chuang-Harrow~\cite{BCH}
&
Clebsch-Gordan transforms
&
$\widetilde{O}\!\left(\min\{n^5,nd^4\}\right)$
&
Full Schur transform
\\[1mm]

Revised Krovi~\cite{Krovi_2019,HighDimSchur}
&
$\mathbb S_n$ QFT and basis changes
&
$\widetilde{O}\!\left(n^4\right)$
&
Full Schur transform
\\[1mm]

Cervero-Mart{\fontencoding{T1}\selectfont\'{i}}n et al.~\cite{cervero2024memory}
&
Streaming Clebsch-Gordan transforms
&
$\widetilde{O}\!\left(nd^4\right)$
&
Direct USS
\\[1mm]

\textbf{This work}
&
 SWAP tests
&
$\widetilde{O}\!\left(\min\{n^4,nd^3\}\right)$
&
Direct USS
\\

\hline
\end{tabular}
\end{table*}

\subsection{Overview of Results}

For a permutation-invariant (PI) state, the information retained by
USS consists of the probability distribution over Young diagrams
\(\lambda\), together with the corresponding reduced states on the
\(\mathrm{SU}(d)\) irrep subsystems \(\mathcal Q_\lambda^d\).
Our implementation returns \(\lambda\) and retains the irrep state
inside the original \(n\)-qudit Hilbert space, with the multiplicity
subsystem \(\mathcal P_\lambda\) reset to a fixed pure state
\(\ket{\eta_\lambda}\). The ideal channel is lossless on PI inputs:
discarding the classical label and applying a permutation twirl to
the quantum output recovers the input state. For a general input,
we first apply a permutation twirl, so the information retained is
precisely its PI information. We denote this ideal channel by
\(\mathcal E_{\mathrm{Schur}}\), as defined in Eq.~\eqref{Schur}.

The central observation is that this canonical output can be
obtained by recursively extracting largest antisymmetric subsystems.
For an input in the sector labelled by \(\lambda\), the largest
antisymmetric subsystem that can be extracted has size
\(\lambda'_1\), the length of the first column. Separating this
subsystem and repeating the procedure on the remaining qudits
reveals the subsequent column lengths and prepares the required
multiplicity state. We develop two algorithms based on this
observation. Algorithm~1 searches over candidate subsets using
ideal antisymmetric measurements and serves as a proof of concept.
Algorithm~2 grows each antisymmetric subsystem one qudit at a
time, giving an efficient implementation.

The growth process admits a natural interpretation as a stochastic
traversal of the Young diagram. Within a column, a successful growth
step advances the traversal by one box, while a failed attempt
leaves the current size unchanged. The algorithm uses a prescribed
number of unsuccessful attempts to decide when to stop growing the
current subsystem and begin the next column. We bound the
probability of stopping before reaching the bottom of a column,
as well as the total number of growth attempts. The geometry of
the Young diagram thus provides a framework for analyzing both
correctness and convergence.

The remaining step is to implement the antisymmetric measurements
using pairwise SWAP tests. A SWAP test projects two qudits onto
their symmetric or antisymmetric subspace. Repeated tests can
identify a fully antisymmetric subsystem, but reproducing the
post-measurement state also requires care: a test can disturb
coherences in the orthogonal complement. Algorithm~2 supplies
the additional structure needed for our construction, since
each growth measurement acts on an already antisymmetric block
and one additional qudit. Under this promise, we give a
SWAP-test protocol that approximates the full antisymmetric
measurement, including its post-measurement state.

For any \(0<\epsilon<1\), our protocol approximates unitary Schur
sampling to error at most \(\epsilon\) in diamond distance using
\begin{equation}
O\!\left(
n\min\{n,d\}^{3}\log^2\!\frac{n}{\epsilon}
\right)
\end{equation}
pairwise SWAP tests. This bound holds on every execution.
For fixed local dimension, the number of tests is therefore
linear in \(n\), up to logarithmic factors.


While, unitary Schur sampling is a versatile primitive in a broad range of quantum-information tasks, we highlight several applications for which the specific form of  encoding of output states in  our protocol is particularly well suited. Most notably, in  QPA \cite{QPA1, QPA2}, existing implementations require not only Schur sampling but also a representation-theoretic correction within each Schur-Weyl sector, implemented using generalized quantum phase estimation or sequences of Clebsch-Gordan transforms and inverse representation-theoretic processing. In our case, on the other hand, one can achieve optimal QPA for single-copy output  simply by discarding all the output qudits of our protocol except one. In Sec.\ref{sec:Applications}, we further discuss QPA and other applications, such as optimal quantum cloning and quantum rank testing.

\section{Setup}

Consider a system of $n$ qudits with the total Hilbert space $(\mathbb{C}^d)^{\otimes n}$. We are interested in the properties of quantum states that are independent of the order of the qudits. More formally, for any state $\sigma$, we are interested in the information that can be obtained from the PI state $\mathcal{T}(\sigma)$, where $\mathcal{T}$ is the quantum channel defined by applying a permutation on $n$ qudits, chosen uniformly at random from the group of all such permutations $\mathbb{S}_n$ (See Appendix \ref{App:Concepts} for further details).  By averaging over all permutations, $\mathcal{T}(\cdot)$ projects any state $\sigma$ onto the space of permutation-invariant states. Hence, in the following, we often restrict our attention to the case of PI states.

As we discuss further below, thanks to \emph{Schur-Weyl duality}, such states can be conveniently characterized using the representation theory of the group  $\mathrm{SU}(d)$,  i.e., single-qudit unitaries with determinant 1, which is widely used in physics (e.g., to describe rotations).  Recall that under the action, $U \mapsto U^{\otimes n}$, of $\mathrm{SU}(d)$, the \(n\)-qudit Hilbert space decomposes into subspaces corresponding to inequivalent irreps of $\mathrm{SU}(d)$. These are commonly labeled by integer partitions of \(n\) with at most \(d\) parts. Such a partition, $\lambda \vdash n$, is a non-increasing sequence of non-negative integers,
\begin{align}
\lambda=(\lambda_1,\lambda_2,\ldots,\lambda_d),
\qquad
\lambda_1\ge \lambda_2\ge \cdots \ge \lambda_d \ge 0,
\end{align}
such that \(\sum_{i=1}^{d}\lambda_i=n\). Partitions are conveniently represented by \textit{Young diagrams}, where the \(i\)-th row contains \(\lambda_i\) boxes. We denote the length of the \(i\)-th column by \(\lambda'_i\), so that
\begin{align}
\min{\{n,d\}}\ge \lambda_1'\ge \lambda_2'\ge \cdots \ge 0.
\end{align}

Under Schur-Weyl duality,
each partition $\lambda$ labels simultaneously an irreducible
representation of $\mathrm{SU}(d)$ and an irreducible representation
of the symmetric group $\mathbb S_n$. The corresponding
$\mathrm{SU}(d)$ irrep is denoted by $\mathcal Q^d_\lambda$, while
$\mathcal P_\lambda$ denotes the corresponding multiplicity space,
which carries an irrep of $\mathbb S_n$. The number of rows
of a Young diagram cannot exceed $d$, reflecting the fact that fully
antisymmetrizing more than $d$ qudits yields zero.\\

As a simple illustration, consider $n=3$ qudits with local dimension
$d=3$. The allowed partitions of $3$ with at most three rows are
\[
(3,0,0),\qquad (2,1,0),\qquad (1,1,1).
\]
Equivalently, the Schur-Weyl decomposition contains three symmetry
sectors, represented by the Young diagrams
\[
\begin{array}{c@{\qquad}c@{\qquad}c}
\yng(3) & \yng(2,1) & \yng(1,1,1) \\
(3,0,0) & (2,1,0) & (1,1,1)
\end{array}
\]
The sector labelled by $(3,0,0)$ corresponds to the totally symmetric
subspace. It carries the trivial representation of $\mathbb S_3$,
meaning that every permutation acts as the identity on this sector.
Consequently, it is spanned by  pure states that are invariant
under permutations of the three subsystems.

At the opposite extreme, the sector labelled by $(1,1,1)$ corresponds
to the totally antisymmetric subspace. This sector carries the sign
representation of $\mathbb S_3$, under which even permutations act as
$+1$ and odd permutations act as $-1$.\\

Then, under the action of $\mathrm{SU}(d)$ transformations, the total Hilbert space of qudits decomposes as 
\be\label{dec}
(\mathbb{C}^d)^{\otimes n} \cong \bigoplus_{\lambda \vdash n} \big(\mathcal{Q}^d_{\lambda}\otimes \mathcal{P}_{\lambda}\big)\ .
\ee
Here $\mathcal Q^d_\lambda$ carries the irreducible representation of
$\mathrm{SU}(d)$ labelled by $\lambda$, while $\mathcal P_\lambda$
carries the corresponding irreducible representation of the symmetric
group $\mathbb S_n$. The subspace
\(
\mathcal H_\lambda :=
\mathcal Q^d_\lambda \otimes \mathcal P_\lambda
\)
is called the \emph{isotypic component} associated with the partition
$\lambda$. The multiplicity of the $\mathrm{SU}(d)$ irrep
$\mathcal Q^d_\lambda$ is
\(
m(n,\lambda)=\dim(\mathcal P_\lambda).
\)\\

According to Schur-Weyl duality, the group of unitaries generated by pairwise SWAPs of qudits, which is a representation of $\mathbb{S}_n$, is block-diagonal with respect to this decomposition, acts trivially on the subsystem $\mathcal{Q}_{\lambda}^d$ and acts irreducibly on $\mathcal{P}_{\lambda}$. Indeed, each $\mathcal{P}_{\lambda}$ with a different $\lambda$ corresponds to an inequivalent irrep of this group. In particular, the isotypic component labelled by the partition $\lambda=(n,0,\ldots,0)$ corresponds to the totally symmetric
subspace and carries the trivial representation of $\mathbb S_n$. This subspace includes all pure states that are invariant under permutations, such as GHZ and Dicke states. 

It follows that any PI $n$-qudit state $\sigma$, such as  $\sigma=\rho^{\otimes n}$, is block-diagonal with respect to this decomposition and  takes the form
\be\label{dec1}
\sigma=\bigoplus_{\lambda} p_{\lambda} \left(\omega_{\lambda}\otimes \frac{\mathbb{I}_{m(n,\lambda)}}{m(n,\lambda)}\right)\ ,
\ee
where $p_{\lambda} = \Tr(\sigma \Pi_{\lambda})$ is the probability that the state $\sigma$ is found in the subspace with partition/Young Diagram $\lambda$, i.e., the subspace $\mathcal{Q}^d_{\lambda}\otimes \mathcal{P}_{\lambda}$, $\Pi_{\lambda}$ denotes the projector onto this subspace, and $\omega_{\lambda}$ is a density operator on $\mathcal{Q}^d_{\lambda}$, the irrep of $\mathrm{SU}(d)$ with the partition/Young Diagram $\lambda$. The projective measurement with projectors $\{\Pi_\lambda\}_{\lambda\vdash n}$ is known as \emph{Weak Schur Sampling} \cite{weakchilds}. \\

For each permutation \(\pi\in \mathbb S_n\), let \(P_\pi\) denote the unitary
operator that permutes the \(n\) tensor factors according to \(\pi\).
The permutation-twirling channel is then defined by
\begin{align}
\mathcal T(\rho)
=
\frac{1}{n!}
\sum_{\pi\in\mathbb S_n}
P_\pi\rho P_\pi^\dagger.
\end{align}
 By Eq.~\eqref{dec1} and Schur's lemma, it
removes all information stored in the multiplicity subsystem while preserving
the state on each $\mathrm{SU}(d)$ irrep subsystem. The permutation twirl projects
arbitrary states onto the permutation-invariant subspace, allowing us to
restrict our attention to PI states. \\

A fundamental consequence of Schur-Weyl duality is that any operator
\(M\) satisfying
\begin{align}
   [M,U^{\otimes n}]=0
\qquad
\forall\, U\in \mathrm{SU}(d), 
\end{align}
acts trivially on the $\mathrm{SU}(d)$ irrep subsystem. More precisely, with respect to
the decomposition in Eq.~\eqref{dec},
\begin{align}
M
=
\bigoplus_{\lambda}
\left(
\mathbb I_{\mathcal Q_\lambda^d}
\otimes
\widetilde M_\lambda
\right), \label{Eqn:Schur-weylSU(d)inv}
\end{align}
where \(\widetilde M_\lambda\) is an operator acting on the multiplicity
space \(\mathcal P_\lambda\). Equivalently, any operator that respects
the collective \(\mathrm{SU}(d)\) symmetry can only modify the
multiplicity degrees of freedom and cannot affect the state stored in
the irrep subsystem \(\mathcal Q_\lambda^d\).\\

Moreover, Schur-Weyl duality implies that the algebra of all
$\mathrm{SU}(d)$-invariant operators on
$(\mathbb{C}^d)^{\otimes n}$ is precisely the algebra generated by the
permutation operators $\{P_\pi:\pi\in\mathbb S_n\}$. Consequently, every
$\mathrm{SU}(d)$-invariant operator can be expressed as a linear
combination of permutations.\\

These observations will play a central role in our analysis. In
particular, it implies that measurements constructed from
\(\mathrm{SU}(d)\)-invariant operators, such as SWAP tests, preserve
the information encoded in the irrep subsystem $
\mathcal{Q}^d_{\lambda}$ while acting nontrivially
only on the multiplicity subsystem $
\mathcal{P}_{\lambda}$. 

\subsection{A canonical state in each isotypic component}

For each partition \(\lambda\vdash n \), let \(\lambda'\) denote the conjugate partition, whose entries \(\lambda'_1,\ldots,\lambda'_r\) are the column lengths of $\lambda$. We define the state
\begin{align}
\ket{\phi_{\lambda}}=\bigotimes_{i=1}^{r}
  \ket{0 \wedge 1 \wedge \cdots \wedge (\lambda'_i-1)},
\label{Eq:representative}
\end{align}
where
\begin{align}
   & \ket{0 \wedge 1 \wedge \cdots \wedge (k-1)}\nonumber\\ 
   &= \frac{1}{\sqrt{k!}}
\sum_{\pi\in \mathbb S_k}
\mathrm{sgn}(\pi)
\ket{\pi(0),\pi(1)\cdots \pi(k-1)}
\end{align}
denotes the normalized totally antisymmetric state on \(k\) qudits. The state \(\ket{\phi_\lambda}\) lives in 
the isotypic component
$
\mathcal H_\lambda =
\mathcal Q^d_\lambda \otimes \mathcal P_\lambda$. Intuitively, each column of the Young diagram is encoded as a totally antisymmetric block, and the tensor product over all columns produces a state with the symmetry type specified by \(\lambda\).

As an example, for the 3 qudit case:
\begin{align}
&\ket{\phi_{(3,0,0)}}=\ket{0}^{\otimes 3}, \nonumber
\quad
\ket{\phi_{(2,1,0)}}=\ket{0\wedge1}\otimes\ket0,\\
&\ket{\phi_{(1,1,1)}}=\ket{0\wedge1\wedge2}.
\end{align}
An interesting property of these states is that under the Schur-Weyl decomposition, the state
$\ket{\phi_\lambda}$ is unentangled between the irrep and
multiplicity subsystems. That is,
\begin{align} \label{Eqn:repstatefactorizes}
\ket{\phi_\lambda}
=
\ket{\psi_\lambda}
\otimes
\ket{\eta_\lambda},
\end{align}
where $\ket{\psi_\lambda}\in\mathcal Q_\lambda^d$ and
$\ket{\eta_\lambda}\in\mathcal P_\lambda$.  We prove this fact in Appendix \ref{AppAsub:PurePermutationRegister}.

Applying the permutation twirl yields
\begin{align}
\mathcal T
\left(
|\phi_\lambda\rangle\langle\phi_\lambda|
\right) 
=
\ket{\psi_\lambda}\bra{\psi_\lambda}
\otimes
\frac{\mathbb I_{m(n,\lambda)}}{m(n,\lambda)}.
\label{Eq:twirledRepstate}
\end{align}
In other words, twirling removes all information contained in the multiplicity subsystem $\mathcal{P}_{\lambda}$, 
while preserving the state on the \(\mathrm{SU}(d)\) irrep subsystem $
\mathcal{Q}^d_{\lambda}$. This means  the original state $\ket{\phi_\lambda}$ can be recovered using SU(d)-invariant operation. In particular, by bringing the multplicity  $
\mathcal{P}_{\lambda}$ back to its initial state $|\eta\rangle$ we recover the original state.\\

A standard construction of the irreducible representations of the
symmetric group uses Young symmetrizers~\cite{harrowthesis,SaganSymmetricGroup}.
A Young tableau of shape \(\lambda\vdash n\) is a filling of the boxes
of the Young diagram \(\lambda\) with the integers \(1,\ldots,n\),
each appearing exactly once. We consider the canonical tableau \(T\)
obtained by numbering the boxes down each column and then proceeding
from left to right.

Let \(R_T \in \mathbb S_n\) and \(C_T \in \mathbb S_n\) denote the row
and column groups of \(T\), respectively; that is, \(R_T\) consists of
permutations that preserve each row of \(T\), while \(C_T\) consists
of permutations that preserve each column. The normalized Young
symmetrizer associated with \(T\) is
\begin{align}
Y_{\lambda,T}
:=
\frac{\dim\mathcal P_\lambda}{n!}
\left(
\sum_{c}\operatorname{sgn}(c)P_c
\right)
\left(
\sum_{r}P_r
\right),
\end{align}
where the sums are over permutations preserving the columns and rows,
respectively.   \(Y_{\lambda, T}\) is supported on the \(\lambda\) isotypic component of the Schur-Weyl decomposition,
\begin{align}
Y_{\lambda,T}
=
I_{\mathcal Q_\lambda}
\otimes
\tilde{Y}_{\lambda, T},
\end{align}
where \(\tilde{Y}_{\lambda, T}\) is rank-1 and idempotent (but not generally self-adjoint). We define \(\ket{y(T)}\bra{y(T)}\) as the distinguished pure state that projects to the image of \(\tilde{Y}_{\lambda, T}\) in \(\mathcal{P}_{\lambda}\).

The representative state \(\ket{\phi_\lambda}\) introduced in this work
is an eigenvector of \(Y_{\lambda,T}\). Consequently, the multiplicity state
\(\ket{\eta_\lambda}\) appearing in the Schur-Weyl decomposition
coincides with the distinguished pure state selected by the Young
symmetrizer. We prove this property in
Appendix~\ref{AppAsub:YoungSymmetrizers}.

\subsection{SWAP tests}

The elementary operation used throughout this work is the two-qudit
SWAP test. For two qudits, the SWAP operator is defined by
\[
\mathrm{SWAP}\ket{a}\ket{b}=\ket{b}\ket{a},
\qquad a,b\in\{0,\ldots,d-1\}.
\]
It has two eigenspaces: the symmetric subspace, with eigenvalue $+1$,
and the antisymmetric subspace, with eigenvalue $-1$. The corresponding
projectors are
\[
\Pi^+=\frac{\mathbb I+\mathrm{SWAP}}{2},
\qquad
\Pi^-=\frac{\mathbb I-\mathrm{SWAP}}{2}.
\]
Thus, a SWAP test is a two-outcome measurement
\(\{\Pi^+,\Pi^-\}.\)
We say that the SWAP test succeeds when the symmetric outcome is
obtained and fails when the antisymmetric outcome is obtained. This measurement is compatible with the Schur-Weyl decomposition
because the SWAP operator commutes with the collective action
$U^{\otimes n}$ of $\mathrm{SU}(d)$.
Consequently, a SWAP test is block diagonal w.r.t.~the isotypic sectors, \(\mathcal{Q}_{\lambda}\otimes \mathcal{P}_{\lambda}\), and acts non-trivially only the \(\mathcal{P}_{\lambda}\) components; thus, it does not affect the information stored in PI states.

For example, in the three-qudit case discussed above, states in the
$(3,0,0)$ sector are fully symmetric, so every pairwise SWAP test
succeeds. States in the $(1,1,1)$ sector are fully antisymmetric, so
every pairwise SWAP test fails. In the mixed-symmetry sector
$(2,1,0)$, a pairwise SWAP test can return either outcome. Thus, SWAP tests are sufficient to project a 3-qudit PI state into isotypic components, i.e, to perform weak Schur sampling on 3 qudits. 

\subsection{Antisymmetric Measurements}

The primary measurement used in our algorithms is the
\emph{antisymmetric measurement}. For a subset \(S\subseteq[n]\), of
\(k\) qudits, where $[n]=[1,\cdots,n]$, this is the two-outcome projective measurement
\begin{align}
\left\{
\Pi^-_S,\,
\mathbb I-\Pi^-_S
\right\},
\end{align}
where \(\Pi^-_S\) denotes the projector onto the fully antisymmetric
subspace $\wedge^k(\mathbb C^dd)$.
of the qudits indexed by \(S\). We refer to the outcome associated with
\(\Pi^-_S\) as the \emph{antisymmetric outcome}. We review several properties of this measurement that will be used throughout our analysis in Appendix \ref{AppAsub: Antisym}. 

Although we treat the antisymmetric measurement as a primitive in the
description and analysis of our algorithms, we show that it can be
approximately using only pairwise SWAP tests. We give this construction
and analyze its accuracy and complexity in Sec.~\ref{section:SWAPtoAntisymm}.

\subsection{Pieri's branching rule}
\label{subsection:Pieri}

Throughout this work, we repeatedly separate fully antisymmetric
subsystems from states supported on a fixed isotypic component.
The representation-theoretic effect of this operation is governed
by Pieri's branching rule. Let $\lambda\vdash n$ be a Young diagram,
and suppose that a $k$-qudit subsystem supported on
$\wedge^k(\mathbb C^d)$ is separated from a state supported on the
$\lambda$-isotypic component. The reduced state of the remaining
$n-k$ qudits can have support only on isotypic components labelled
by Young diagrams $\mu\vdash(n-k)$ satisfying Pieri's rule.
The admissible diagrams $\mu$ are precisely those obtained from
$\lambda$ by removing $k$ boxes, with no two removed boxes in
the same row, such that the remaining boxes form a Young diagram.
\footnote{Equivalently, $\lambda/\mu$ is a vertical strip of size $k$.
We illustrate this rule in Fig.~\ref{Fig:Pieris rule}.}

A particularly important special case occurs when
$k=\lambda'_1$, the length of the first column of $\lambda$.
Since $\lambda$ has exactly $k$ nonempty rows, exactly one box
must be removed from each row. Thus, there is a unique admissible
Young diagram: the shape obtained by deleting the first column
of $\lambda$. We prove this independently, without invoking representation theoretic results, using only the representative states $\ket{\phi_\lambda}$ in Lemma~\ref{Lemma:removingcolumns}.
This special case underlies the recursive
decomposition used throughout the paper, where each invocation
of the algorithm separates a largest antisymmetric subsystem
before applying the same procedure to the remaining qudits.

In Appendix~\ref{AppA:pieris}, we go into more detail about Pieri's rule in terms of the representation theory of \(\mathbb{S}_{n}\).

\begin{figure}[h] 
\centering
\begin{tikzpicture}[scale=0.45]

\begin{scope}[shift={(0,0)}]
\node at (2,-3.6) {(i)};

\foreach \x in {0,1,2}{
    \draw (\x,0) rectangle ++(1,1);
}
\fill[red!40] (3,0) rectangle ++(1,1);
\draw (3,0) rectangle ++(1,1);

\foreach \x in {0,1}{
    \draw (\x,-1) rectangle ++(1,1);
}
\fill[red!40] (2,-1) rectangle ++(1,1);
\draw (2,-1) rectangle ++(1,1);

\foreach \x in {0,1}{
    \draw (\x,-2) rectangle ++(1,1);
}
\end{scope}

\begin{scope}[shift={(7,0)}]
\node at (2,-3.6) {(ii)};

\foreach \x in {0,1,2}{
    \draw (\x,0) rectangle ++(1,1);
}
\fill[red!40] (3,0) rectangle ++(1,1);
\draw (3,0) rectangle ++(1,1);

\foreach \x in {0,1,2}{
    \draw (\x,-1) rectangle ++(1,1);
}

\draw (0,-2) rectangle ++(1,1);
\fill[red!40] (1,-2) rectangle ++(1,1);
\draw (1,-2) rectangle ++(1,1);
\end{scope}

\begin{scope}[shift={(14,0)}]
\node at (2,-3.6) {(iii)};

\foreach \x in {0,1,2,3}{
    \draw (\x,0) rectangle ++(1,1);
}

\foreach \x in {0,1}{
    \draw (\x,-1) rectangle ++(1,1);
}
\fill[red!40] (2,-1) rectangle ++(1,1);
\draw (2,-1) rectangle ++(1,1);

\draw (0,-2) rectangle ++(1,1);
\fill[red!40] (1,-2) rectangle ++(1,1);
\draw (1,-2) rectangle ++(1,1);
\end{scope}

\end{tikzpicture}
\caption{The three admissible ways to remove two boxes from $\lambda=(4,3,2)$.
The highlighted boxes are removed; no two removed boxes lie in the same row.}
\label{Fig:Pieris rule}
\end{figure}
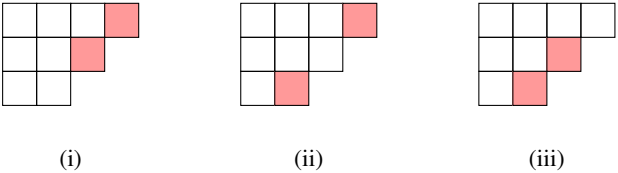

\section{unitary Schur sampling} \label{section:USS}
For many applications in quantum information and computation, it is useful to physically realize the transformation that, for a given PI input state $\sigma$, measures $\lambda$, i.e., returns the irrep label $\lambda$ with probability $p_{\lambda}$, along with the corresponding post-measurement state in $\mathcal{Q}^d_{\lambda}$, $\omega_{\lambda}$, discarding $\mathcal{P}_\lambda$. In Ref.~\cite{cervero2024memory}, \emph{unitary Schur sampling} is defined as weak Schur sampling followed by discarding the permutation register \(\mathcal P_\lambda\). Thus, conditioned on obtaining the label \(\lambda\), the output consists of the classical label \(\lambda\) together with the reduced post-measurement state on the unitary-group irrep subsystem \(\mathcal Q_\lambda^d\).

In practice, for many applications,  such as purification \cite{Cirac}, majority vote \cite{majority}, and coherence distillation \cite{Marvian2020Coherence},  rather than returning the post-measurement state $\sigma_{\lambda}$   in the full \(\lambda\)-isotypic component,  it is often more convenient to realize the retained irrep state directly within the original \(n\)-qudit Hilbert space.  A natural and convenient way to achieve this is to encode each \(\sigma_{\lambda}\) into an \(n\)-qudit state, \(\widetilde{\sigma}_{\lambda}\), supported on an irreducible subspace (w.r.t.~$\mathrm{SU}(d)$) within the \(\lambda\) isotypic component. That is, the goal is to realize the transformation,
\be\label{Eqn:USSabstract}
\mathcal{E}_{\text{Schur}}(\sigma)=\sum_{\lambda} p_{\lambda}\widetilde{\sigma}_{\lambda}\otimes |\lambda\rangle\langle \lambda|_{\text{reg}}\ ,
\ee
where \(\ket{\lambda}\) denotes the state of a classical register that keeps track of $\lambda$.
Our choice of the irreducible subspace to encode \(\sigma_{\lambda}\) is,
\begin{align*}
\mathrm{Span}\left\{
U^{\otimes n}
\ket{\phi_{\lambda}}
:\;
U \in \mathrm{SU}(d)
  \right\}= \mathcal{Q}_{\lambda}^{d}\otimes \mathrm{Span}\{\ket{\eta_{\lambda}}\}.
\end{align*}
In other words, our objective is to design a protocol that takes any PI state $\sigma$, measures \(\lambda\), preserves the state \(\omega_\lambda\) on the \(\mathrm{SU}(d)\) irrep subsystem \(\mathcal Q_\lambda^d\), and prepares its component on the multiplicity subsystem, \(\mathcal{P}_{\lambda}\), of \(\mathrm{SU}(d)\) irrep \(\lambda\) to \(\ket{\eta_{\lambda}}\).
\be\label{Schur}
\mathcal{E}_{\text{Schur}}(\sigma)=\sum_{\lambda} 
p_{\lambda}\omega_{\lambda}\otimes |\eta_{\lambda}\rangle\langle \eta_{\lambda}| \otimes |\lambda\rangle\langle \lambda|_{\text{reg}}.
\ee
That is, if
\be
\sigma_\lambda
=
\omega_\lambda\otimes
\frac{\mathbb I_{m(n,\lambda)}}{m(n,\lambda)},
\ee
then we encode the retained irrep state as
\be
\widetilde{\sigma}_\lambda
=
\omega_\lambda\otimes
\ket{\eta_\lambda}\bra{\eta_\lambda}.
\ee
This choice is also natural from the perspective of state purification, or quantum purity amplification. The multiplicity state \(\ket{\eta_\lambda}\) corresponds to the column-ordered Young tableau, which is precisely the representative used in the optimal single-output quantum purity amplification protocol of~\cite{QPA1}. Thus, our choice of multiplicity state is directly compatible with optimal single-output purification. We discuss this connection further in Subsection~\ref{subsec:QPAapplications}. In the qubit case, this choice has a particularly simple interpretation:
the qubits outside the spin-\(j\) irrep can be separated into singlet pairs,
which are fixed \(\mathrm{SU}(2)\)-invariant states and therefore carry no
information about the unknown input state.  All of the relevant PI information
is retained in the remaining spin-\(j\) subsystem~\cite{Cirac, qubitSWAPUSS}.

\subsection{ Reducing to representative states } \label{subsec:reducingtorepstates}

We now show that it is sufficient to analyze our protocol on a single pure representative state from each irrep of $\mathrm{SU}(d)$. The key observation is that every measurement used in the protocol is $\mathrm{SU}(d)$-invariant. Consequently, once the protocol succeeds on one representative state in an irrep, it necessarily succeeds on every state supported on the same irrep. Therefore, rather than analyzing arbitrary permutation-invariant states, it suffices to consider the representative states, \(\ket{\phi_{\lambda}}\), introduced in Eq.~\eqref{Eq:representative}.\\

Consider a quantum channel $\mathcal{E}(.)=\sum_r M_r(.)M_r^{\dagger}$ that uses $\mathrm{SU}(d)$-invariant Kraus operators $\{M_r\}_{r}$,
\begin{align}
   [M_r, U^{\otimes n}]=0,
\qquad
\forall\,U\in\mathrm{SU}(d),\ \forall r,
\end{align}
and suppose that it satisfies,
\begin{align}
    \mathcal{E}\left( \mathcal{T}\left(\ket{\phi_{\lambda}}\bra{\phi_{\lambda}}\right) \right) =\ket{\phi_{\lambda}}\bra{\phi_{\lambda}}\otimes \ket{\lambda}\bra{\lambda}_{\text{reg}},\quad \forall \lambda \vdash n.  \label{assumchannel}
\end{align}
Such a channel performs the $\mathcal{E}_{\text{Schur}}$ channel; moreover, since the classical register records the irrep label $\lambda$, it also performs weak Schur sampling. That is, for any PI state, $\omega_{\lambda}\otimes \mathbb{I}_{m(n,\lambda)}/m(n,\lambda)$, living in the irrep labelled by $\lambda$, it gives us the state,
\begin{align} 
\mathcal{E}\left( {\omega}_{\lambda}\otimes \frac{\mathbb{I}_{m(n,\lambda)}}{m(n,\lambda)} \right)=
{\omega}_{\lambda}\otimes |\eta_{\lambda}\rangle\langle \eta_{\lambda}|\otimes \ket{\lambda}\bra{\lambda}_{\text{reg}}, \label{Eqn:USS}
\end{align}
thus achieving unitary Schur sampling.\\ 

To prove this, first note that there exists a collection of $\{p_i, U_{\lambda,i} \}$ PI unitaries acting on the $\lambda$ isotypic component, i.e., $U_{\lambda,i}=\widetilde{U}_{\lambda,i}\otimes \mathbb I_{m(n,\lambda)}$ such that
\begin{align} \label{Eqn:Wtakestorepstate}   \sigma_{\lambda}=\mathcal{W}\left( \mathcal{T}(\phi_{\lambda}) \right),
\end{align}
where $   \mathcal{W}(\cdot) \coloneqq  \sum_i p_i U_{\lambda,i} (\cdot)U_{\lambda,i}^{\dagger}$ and $\phi_\lambda=|\phi_\lambda\rangle\langle\phi_\lambda|$. 
Indeed, writing
\(\omega_\lambda=\sum_i p_i\ket{\psi_i}\bra{\psi_i}\),
we may choose \(\widetilde U_{\lambda,i}\) such that
\(\widetilde U_{\lambda,i}\ket{\psi_\lambda}=\ket{\psi_i}\).
This would imply 
\begin{multline}\label{Eqn:WactioninSU(d)}
    \mathcal{W}\left(
    \psi_\lambda \otimes
    \frac{\mathbb{I}_{m(n,\lambda)}}{m(n,\lambda)}
    \right)= \\\sum_i p_i \widetilde{U}_{\lambda,i}\psi_\lambda \widetilde{U}^{\dagger}_{\lambda,i} \otimes
    \frac{\mathbb{I}_{m(n,\lambda)}}{m(n,\lambda)}\ ,
\end{multline}
where $\psi_\lambda=|\psi_\lambda\rangle\langle\psi_\lambda|$. Since each \(M_r\) is \(\mathrm{SU}(d)\)-invariant, Schur-Weyl duality implies that
\(M_r\) acts as the identity on \(\mathcal Q_\lambda^d\). Hence
\([U_{\lambda,i},M_r]=0\). Then, using Eq.~\eqref{Eqn:Wtakestorepstate},
\begin{align}
   &\mathcal{E}\left( \sigma_{\lambda} \right) = \mathcal{E}\left(\mathcal{W}\left( \mathcal{T}(\phi_{\lambda}) \right)\right) \notag\\ 
   &= \mathcal{W}\left(\mathcal{E}\left( \mathcal{T}(\phi_{\lambda}) \right)\right)=\mathcal{W}\left(\phi_{\lambda}\right)\otimes \ket{\lambda}\bra{\lambda}_{\text{reg}}
\end{align}
Since \(\phi_\lambda=\psi_\lambda\otimes|\eta_\lambda\rangle\langle\eta_\lambda|\), the channel \(\mathcal W\) acts only on the \(\mathcal Q_\lambda^d\) factor and leaves \(|\eta_\lambda\rangle\) unchanged. Finally using Eq.~\eqref{Eqn:WactioninSU(d)},
\begin{align}
   \mathcal{E}\left( \sigma_{\lambda} \right)={\omega}_{\lambda}\otimes |\eta_{\lambda}\rangle\langle \eta_{\lambda}|\otimes \ket{\lambda}\bra{\lambda}_{\text{reg}}.    
\end{align}
Therefore, it suffices to construct a protocol satisfying the condition in Eq. ~\eqref{assumchannel} for every partition $\lambda\vdash n$. In the next section we show that antisymmetric measurements are sufficient for this.

\color{black}
\section{unitary Schur sampling via antisymmetric measurements}
We wish to realize the channel defined in Eq.~\eqref{assumchannel}. The key observation is that the canonical representative state $\ket{\phi_\lambda}$ is constructed from antisymmetric states, one for each column of the Young diagram. Consequently, if we can identify and separate the first column, the problem reduces recursively to the Young diagram obtained by removing that column. Further, we will implement this channel using only antisymmetric measurements, i.e., show that a sequence of antisymmetric measurements can implement this channel, and thus unitary Schur sampling. The following lemma regarding representative states is the key idea behind this approach.
\begin{restatable}{lemma}{removingcolumns} \label{Lemma:removingcolumns}
Let \(\lambda\) be a Young diagram with column lengths
\(\lambda'_1\ge \dots\ge \lambda'_r\), and let
\(\ket{\phi_{\lambda}}\) be the corresponding representative state.
Let
\(S_1=\{0,\dots,\lambda'_1-1\}\) denote the first
\(\lambda'_1\) qudits, and define
\begin{align}
\ket{\Lambda_{S_1}}
:=
\ket{0 \wedge 1 \wedge \cdots \wedge (\lambda'_1-1)} .
\end{align}
Let \(\lambda \setminus \lambda'_1\) denote the Young diagram obtained by
removing the first column of \(\lambda\). Then
\begin{equation}
  \begin{aligned}
    &\Pi^{-}_{S_1}
      \,
      \mathcal{T}\!\left(
      \ket{\phi_{\lambda}}\bra{\phi_{\lambda}}
      \right)
      \,
      \Pi^{-}_{S_1}
      \propto \ifmaintext{\\ & \qquad \qquad \quad}
                         \ket{\Lambda_{S_1}}\bra{\Lambda_{S_1}}
                         \otimes
                         \mathcal{T}\!\left(
                         \ket{\phi_{\lambda \setminus \lambda'_1}}
                         \bra{\phi_{\lambda \setminus \lambda'_1}}
                         \right).    
  \end{aligned}
\end{equation}
Moreover, let \(M\) be any \(\mathrm{SU}(d)\)-invariant operator and let
\(S\subseteq[n]\). If $d \ge|S|> \lambda'_1$,
\begin{align}
\Tr_{\overline S}\!\left(
M|\phi_\lambda\rangle\langle\phi_\lambda|
\,M^\dagger
\right) \ \text{and} \ \Tr_{\overline S}\!\left(
M\,
\mathcal T(|\phi_\lambda\rangle\langle\phi_\lambda|)
\,M^\dagger
\right)
\end{align}
have no support on the  fully antisymmetric subspace of
\(|S|\) qudits.
\end{restatable}

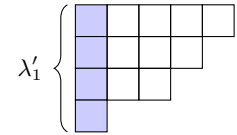
\begin{figure}[h]
\centering
\begin{tikzpicture}[scale=0.7]

\def\s{0.6}

\foreach \r/\len in {0/5,1/4,2/3,3/1}{
    \foreach \c in {0,...,\numexpr\len-1\relax}{
        \ifnum\c=0
            \fill[blue!20] (\c*\s,-\r*\s) rectangle ++(\s,-\s);
        \fi
        \draw (\c*\s,-\r*\s) rectangle ++(\s,-\s);
    }
}

\draw[
decorate,
decoration={brace,mirror,amplitude=5pt}
]
(-0.15,0) -- (-0.15,-4*\s)
node[midway,left=6pt] {$\lambda'_1$};

\end{tikzpicture}

\caption{
The first column of the Young diagram $\lambda$ has length $\lambda'_1$. The unshaded boxes represent the Young diagram $\lambda\setminus\lambda'_1$. After separating the antisymmetric subsystem corresponding to the first column, the remaining qudits are supported on the isotypic component labeled by $\lambda\setminus\lambda'_1$.
}
\label{fig:firstcolumn}
\end{figure}

The proof can be found in Appendix~\ref{AppAsub:lemma 1}. From the discussion in the previous section, the above argument should apply to general states living in $\lambda$ isotypic component. We make this precise in the following lemma:
\begin{restatable}{lemma}{largestantisymm} \label{Thm:largest antisymm}
Let \(\sigma_\lambda\) be a  state of $n$ qudits restricted to the
\(\lambda\)-isotypic component of
\((\mathbb C^d)^{\otimes n}\), and let
\(M\) be an
\(\mathrm{SU}(d)\)-invariant measurement operator, i.e.,
\[M^\dag M \le \mathbb{I}\quad\text{and}\quad  
[M,U^{\otimes n}]=0
\quad
\forall\,U\in\mathrm{SU}(d).
\]
Let 
\begin{align}
\tau_{S_1}=\Tr_{\overline{S_1}}\left( M \sigma_{\lambda} M^{\dagger} \right)
\end{align}
be the non-zero (unnormalized) post-measurement reduced state of subsystem $S_1$.
If $d\ge|S_1|>\lambda'_1$, then $\tau_{S_1}$ has no support on the toally anti-symmetric subspace of $|S_1|$ qudits. Furthermore, suppose that \(|S_1|=\lambda'_1\) and that \(M=\Pi^{-}_{S_{1}}\otimes \mathbb I_{\overline{S_1}}\). Then the reduced state, \(\mathrm{tr}_{S_{1}} ( M \sigma_{\lambda} )\), of the complementary subsystem is  supported entirely on the isotypic
component labelled \(\lambda\setminus\lambda'_1\).
\end{restatable}
The proof is found in Appendix~\ref{AppAsub:lemma 1}; crucially, Lemma~\ref{Thm:largest antisymm} implies that identifying the largest antisymmetric subsystem is equivalent to identifying the first column of \(\lambda\). We perform an antisymmetric measurement, $\{\mathbb{I}-\Pi^{-}_{S_1},\Pi^{-}_{S_1}\}$, on a set, \(S_1\), of $\lambda'_1$ qudits of \(\mathcal{T}\!(\ket{\phi_{\lambda}} \bra{\phi_{\lambda}})\). If projection into the antisymmetric subspace is achieved, then Lemma  \ref{Lemma:removingcolumns} tells us that we may separate the corresponding qudits and obtain,
\begin{align}    \ket{\Lambda_{S_{1}}}\bra{\Lambda_{S_{1}}}
    \otimes
    \mathcal{T}\!\left(
    \ket{\phi_{\lambda \setminus \lambda'_1}}
    \bra{\phi_{\lambda \setminus \lambda'_1}}
    \right).
\end{align}
By repeating this on the remaining columns $\lambda'_2, \lambda'_3, \cdots$, assuming the measurements keep resulting in outcome \(\Pi^{-}_{S_1}\), we recover $\ket{\phi_{\lambda}}$.

We illustrate this idea with the 3 qudit example. We have already established that SWAP tests implement weak Schur sampling for 3 qudits. Thus, to realize unitary Schur sampling, it remains to separate the maximal antisymmetric subsystem from \(\mathcal{T}\!(\ket{\phi_{\lambda}} \bra{\phi_{\lambda}})\). For three qudits, the only partition where this is not trivial is \((2,1)\), for which the maximal antisymmetric subsystem has
size two. In this case, any SWAP test yielding the antisymmetric outcome successfully identifies the desired antisymmetric subsystem. Moreover, a uniformly random SWAP test produces the antisymmetric outcome with probability \(1/2\) (see Lemma \ref{lemma:21k-1random SWAP}). We note that for three qudits, a SWAP-test based implementation of unitary Schur sampling was also given in the SWAPNET construction of~\cite{QPA1}.\\

 This strategy has two important gaps: 
\begin{itemize}
    \item It assumes postselection, i.e, the antisymmetric measurements might fail to project to the antisymmetric subspace. 
    \item It assumes that  $\lambda$ is known.
\end{itemize}

In the rest of the section, we provide two algorithms that circumvent these problems.
The first assumption can be removed by exploiting the permutation twirl. After an unsuccessful antisymmetric measurement, reapplying the twirl restores the state to $\mathcal{T}\!\left(|\phi_\lambda\rangle\langle\phi_\lambda|\right),$ so the protocol is effectively reset to its original configuration. Consequently, a failed measurement is harmless, and the antisymmetric measurement can simply be repeated until the desired outcome is obtained . This observation forms the basis of our first algorithm, which uses this idea via randomization. Although the resulting procedure is not gate-efficient, it provides a simple proof of concept. The second algorithm uses a more involved process of sequentially growing the antisymmetric subsystem to obtain a gate-efficient implementation and constitutes the main algorithmic contribution of this paper.

Before describing the algorithms, we note that the height of the first column always satisfies
\be
\lambda'_1\le \min\{n,d\}. 
\ee
Thus, all bounds below that are written in terms of \(d\) may be sharpened by replacing \(d\) with \(\min\{n,d\}\). For notational simplicity, we continue to write \(d\) throughout.

\subsection{Algorithm 1}\label{subsec:Algo1}

We analyze Algorithm~1 on the representative state \(\mathcal{T}(|\phi_\lambda\rangle\langle\phi_\lambda|)\). By the representative-state reduction established in subsection~\ref{subsec:reducingtorepstates}, the same analysis applies to any permutation-invariant state supported on the \(\lambda\)-isotypic component.
We begin with three simple observations. First, suppose we perform an
antisymmetric measurement $\Pi^{-}_{S}, \mathbb{I} - \Pi^{-}_{S} $ on an arbitrary subset of $k$ qudits. According to Lemma~\ref{Thm:largest antisymm}, if $k > \lambda'_1$ the probability of projecting to the antisymmetric subspace is 0. Second, after any antisymmetric measurement, performing a permutation twirl restores the state to $\mathcal{T}(\ket{\phi_{\lambda}}\bra{\phi_{\lambda}})$, which we explain later in subsection~ \ref{subsec: reversibility}. Third, if $k \le \lambda'_1$, 
\begin{align}
    \Tr\left( \Pi^{-}_{S}  \mathcal{T}(\ket{\phi_{\lambda}}\bra{\phi_{\lambda}}) \right) \ge \frac{\binom{\lambda'_1}{k}}{\binom{n}{k}}.
\end{align}
To see this, view \(\mathcal T(|\phi_\lambda\rangle\langle\phi_\lambda|)\) as the state obtained by applying a uniformly random permutation \(\pi\in\mathbb S_n\) to \(\ket{\phi_\lambda}\). A uniformly random permutation maps the first \(\lambda'_1\) qudits to a uniformly random subset of \(\lambda'_1\) qudits; therefore, any fixed subset \(S\) of \(k\le \lambda'_1\) qudits is contained in this permuted block with probability $\binom{\lambda'_1}{k}/\binom{n}{k}$. Conditioned on this event, the qudits in \(S\) are fully antisymmetric,
 so the antisymmetric outcome occurs with probability one.\\

These observations suggest searching for the maximal antisymmetric subsystem by starting from \(k=\min{\{n,d \}}=\tilde{d}\) and decreasing \(k\) until an antisymmetric outcome is found. We initialize $k=\tilde{d}$, place all the qudits in $B$, and initialize an empty list $L$ in our classical memory. 
We pick a random subset $S \subseteq B$ of size $k$ and perform the antisymmetric measurement $\{ \Pi^-_{S}, \mathbb{I} - \Pi^-_{S} \}$. If the antisymmetric outcome is obtained, the corresponding qudits are moved from \(B\) to \(A\), and \(k\) is appended to \(L\). In Appendix \ref{App:Algo 1}, we show that if $k$ is the size of the largest antisymmetric subsystem supported on the remaining qudits in $B$, the probability of not having at least one antisymmetric outcome within $\binom{n}{k}\ln(n\epsilon^{-1})$ consecutive measurements is at most $\epsilon/n$. If this many consecutive trials fail, we decrement \(k\) by one. Thus, with probability at least \(1-\epsilon/n\), the procedure identifies the largest antisymmetric subsystem supported on \(B\). We show in Appendix~\ref{App:Algo 1} that the algorithm performs at most \(2n^{\tilde{d}+1}\ln(n\epsilon^{-1})\) antisymmetric measurements and outputs \(\ket{\phi_\lambda}\) in system \(A\) with probability at least \(1-\epsilon\). The entries of \(L\) are the column lengths of \(\lambda\). We provide an explicit pseudocode in Algorithm~\ref{alg:1}.\\

The preceding upper bound follows from a simple worst-case analysis, so
one might wonder whether its poor dependence on \(d\) is merely an
artifact of the estimate. This is not the case for the subset-search
strategy of Algorithm~1. Consider an \(L\)-shaped Young diagram with
first-column length \(k\). By the representative-state reduction, we
may take
\begin{align}
\ket{\phi_\lambda}
=
\ket{0\wedge1\wedge\cdots\wedge(k-1)}
\otimes\ket0^{\otimes(n-k)}\ .
\end{align}
Any \(k\)-qudit subset containing fewer than \(k-1\) qudits from the
first antisymmetric block contains at least two copies of \(\ket0\),
and therefore has zero overlap with the fully antisymmetric subspace.
Thus at most \(1+k(n-k)\) of the \(\binom nk\) possible subsets can
succeed, giving 
\be
p_{\mathrm{succ}}
\le
\frac{1+k(n-k)}{\binom nk}\ .
\ee
Taking \(k=\tilde{d}\), and using
$\binom{n}{\tilde{d}}\ge (n-\tilde{d}+1)^{\tilde{d}}/\tilde{d}!$, we find that for fixed \(d\), the success probability is upper-bounded by $\tilde{d}!\,\tilde{d}/(n-\tilde{d}+1)^{\tilde{d}-1}$. Consequently, obtaining a constant probability of success requires
\(\Omega(n^{\tilde{d}-1})\) trials in the worst case.\\

It is useful to note that Algorithm~1 reduces, for \(d=2\), to the
random-SWAP protocol studied in our previous work~\cite{qubitSWAPUSS}. In the
qubit case the only nontrivial antisymmetric subsystem has size two,
so the subset search consists simply of choosing random pairs and
performing SWAP tests. As shown in Ref.~\cite{qubitSWAPUSS}, where we perform a sharper analysis, this requires only
\(2n\ln(n\epsilon^{-1})\) SWAP tests to achieve error \(\epsilon\).
Thus the subset-search strategy is efficient for qubits; its poor
dependence on the local dimension appears only when one attempts to
extend the same strategy to larger antisymmetric subsystems.\\

Finally, an explicit permutation twirl before every measurement is unnecessary. Applying a uniformly random permutation and then measuring a fixed subset of \(k\) qudits is statistically equivalent to choosing a uniformly random subset of \(k\) qudits and measuring it directly. Thus, throughout Algorithm~1, we implement the required randomization simply by choosing a fresh uniformly random \(k\)-qudit subset at every round.
\\

 \begin{algorithm}[H]
\caption{Subset-search Algorithm}
\label{alg:1}
\begin{algorithmic}[1]
\Require PI State \(\sigma_\lambda\),
error parameter \(\epsilon\)
\Ensure State \(\mathcal{E}_{\text{Schur}}(\sigma_\lambda)\) and Young diagram \(\lambda\)

\State Initialize \(B\) to contain all \(n\) qudits, \(A\gets\emptyset\),
and \(L\gets\emptyset\).
\State Set \(k\gets d\).

\While{$B\neq\emptyset$}

\State Perform antisymmetric measurements on uniformly random
\(k\)-qudit subsets \(S\subseteq B\), for at most
\(\binom nk\ln(n\epsilon^{-1})\) trials.

\If{an antisymmetric outcome is obtained}
    \State Move the corresponding subset \(S\) from \(B\) to \(A\).
    \State Append \(k\) to \(L\).
\Else
    \State Set \(k\gets k-1\).
\EndIf

\EndWhile

\State \Return \(\mathcal{E}_{\mathrm{Schur}}(\sigma_\lambda)\) stored in \(A\),
together with \(L\), whose entries specify the Young diagram \(\lambda\).
\end{algorithmic}
\end{algorithm}

\subsection{Algorithm 2}
By Lemma~\ref{Thm:largest antisymm}, the size of the largest
antisymmetric subsystem that can be prepared from a state in the
\(\lambda \)-isotypic component is equal to the length of the first column of
\(\lambda \). Furthermore, after separating this subsystem, the remaining
qudits are supported on the isotypic component corresponding to the Young
diagram obtained by removing the first column. Therefore, recursively
identifying and separating the largest antisymmetric subsystem recovers the
entire Young diagram \( \lambda \), and hence performs weak Schur sampling.
Since the resulting state in the $\mathbb{S}_n$ irrep is $\ket{\eta_\lambda}$, this also
realizes unitary Schur sampling.\\

Algorithm~1 achieves this objective by explicitly searching for antisymmetric
subsystems of a prescribed size. As shown in the previous section, this
requires at most $2n^{d+1}\ln{(n\epsilon^{-1})}$ antisymmetric measurements. While conceptually
simple, the resulting complexity scales exponentially with the local dimension \( d \). We therefore seek a more efficient approach to separate the largest antisymmetric subsystem.\\

We begin Algorithm~2 by applying a uniformly random permutation to the \(n\) qudits, i.e., by applying the permutation-twirling channel \(\mathcal T\). This leaves permutation-invariant input states unchanged and maps arbitrary input states to permutation-invariant states. The key ingredient of Algorithm~2 is a subroutine that constructs the largest
antisymmetric subsystem incrementally. 
Throughout the algorithm, the qudits are partitioned into two disjoint systems, $B$ and $C$. The system $C$ contains the antisymmetric subsystems that have already been separated , while $B$ contains the remaining qudits. We designate a subsystem $A \subseteq B$, which serves as the current candidate antisymmetric subsystem. Throughout the subroutine, the qudits in \(A\) are supported on the fully antisymmetric subspace, and \(A\) is grown by incorporating qudits from \(B\setminus A\). Initially, $B$ consists of all $n$ qudits, $C$ is empty, and $A$ contains a single qudit selected from $B$. We perform the following set of steps. (which we will refer to as a round of measurement) repeatedly:

\begin{itemize}
\item Select a qudit $r$ uniformly at random from $B\setminus A$ .
\item Perform an antisymmetric measurement on the qudits in \(A\) together with qudit $r$.

\item If the measurement projects onto the antisymmetric subspace, add the
$r$ to \(A\).

\item Otherwise, $r$ remains in $B\setminus A$. 
\end{itemize}
Suppose that, before the measurement, \(A\) contains \(k\) qudits supported on the fully antisymmetric subspace. Irrespective of the outcome of the measurement, the original \(k\) qudits in \(A\) remain supported on their fully antisymmetric subspace. Indeed, the projectors \(\Pi^-_{A}\otimes \mathbb I_r \) and \(\Pi^-_{A\cup{r}}\) commute, which is a specialisation of Lemma~\ref{Lemma:nestedcommutator} in Appendix~\ref{AppAsubsec:nestedcommutation}. Therefore \(\Pi^-_{A}\otimes I_r\) also commutes with the complementary projector \(\mathbb I-\Pi^-_{A\cup{r}}\).  Hence, after an unsuccessful growth attempt, we simply retain the original \(k\) qudits as \(A\), while the selected qudit \(r\) remains in \(B\setminus A\). In this way, the qudits in \(A\) are fully antisymmetric after every round. We provide an explicit pseudocode for this subroutine in Algorithm~\ref{alg:growth}.\\

Since the Young diagram \(\lambda\) is not known in advance, the subroutine does not know when \(A\) has reached the size of the largest antisymmetric subsystem supported on \(B\). We therefore introduce a stopping criterion to certify, with high probability, that \(A\) cannot be grown further. The growth subroutine terminates when either of the following conditions is satisfied:
\begin{enumerate}
\item If \(|A|=k\), the subroutine has performed \(R(k)\) consecutive rounds without increasing the size of \(A\).

\item The size of \(A\) has reached the known upper bound \(k_{\max}\), where
\(k_{\max}\) is the size of the antisymmetric subsystem extracted in the
previous invocation of the growth subroutine. For the first invocation, we
set \(k_{\max}=d\).

\end{enumerate}

Here \(R(k)\) is a certification threshold chosen so that, whenever a larger antisymmetric subsystem exists, the probability of observing \(R(k)\) consecutive rounds without growth is sufficiently small. If the growth subroutine terminates after \(R(k)\) consecutive unsuccessful rounds even though the current antisymmetric subsystem \(A\) could still be grown, this premature termination is not detected by the algorithm. The convergence analysis below bounds the probability of this event.
The second stopping condition follows from the fact that the column lengths of a Young diagram are nonincreasing: conditioned on the correctness of the previous extraction, the next antisymmetric subsystem cannot be larger than the preceding one. For the first invocation, the trivial upper bound \(d\) replaces the size of the previous subsystem. We determine the threshold \(R(k)\) and analyze the probability of premature termination in the following sections. When either stopping condition is met, the qudits in \(A\) are shifted to $C$.

  Simultaneously, we record the size of each extracted subsystem in a list \(L\). By Lemma~\ref{Thm:largest antisymm}, these sizes correspond to the column lengths of the Young diagram associated with the current isotypic component. We state this subroutine explicitly in algorithmic form in Algorithm ~\ref{alg:growth}. 
\\

 The full algorithm proceeds by repeatedly invoking this subroutine.  Each invocation of the growth subroutine returns, with high probability, the largest antisymmetric subsystem supported on the current state of \(B\). The algorithm terminates when $B$ is empty. The successive subsystem sizes recorded in \(L\) are the column lengths of \(\lambda\), while the extracted state prepares the multiplicity subsystem in \(\ket{\eta_\lambda}\).
 From Lemma~\ref{Thm:largest antisymm}, the resulting state is the output of unitary Schur sampling, as described in Eq.~\eqref{Schur}. We  illustrate the algorithm in Fig.~\ref{fig:algorithm2}.\\

\begin{algorithm}[H]
\caption{Growth Subroutine}
\label{alg:growth}
\begin{algorithmic}[1]
\Require Qudit sets \(A\subseteq B\), extracted subsystems \(C\), list \(L\),
and upper bound \(k_{\max}\)
\Ensure Maximal antisymmetric subsystem \(A\), with high probability

\State Set \(f\gets 0\).

\While{$A\neq B$ and $|A|<k_{\max}$ and $f<R(|A|)$}

    \State Choose \(r\in B\setminus A\) uniformly at random and perform
    \[
    \{\Pi^-_{A\cup\{r\}},\,\mathbb I-\Pi^-_{A\cup\{r\}}\}.
    \]

    \If{an antisymmetric outcome is obtained}
        \State Add \(r\) to \(A\).
        \State Set \(f\gets 0\).
    \Else
        \State Set \(f\gets f+1\).
    \EndIf

\EndWhile

\State Move \(A\) from \(B\) to \(C\).
\State Append \(|A|\) to \(L\).

\end{algorithmic}
\end{algorithm}

\begin{figure*}[t]
\centering
\includegraphics[width=10cm]{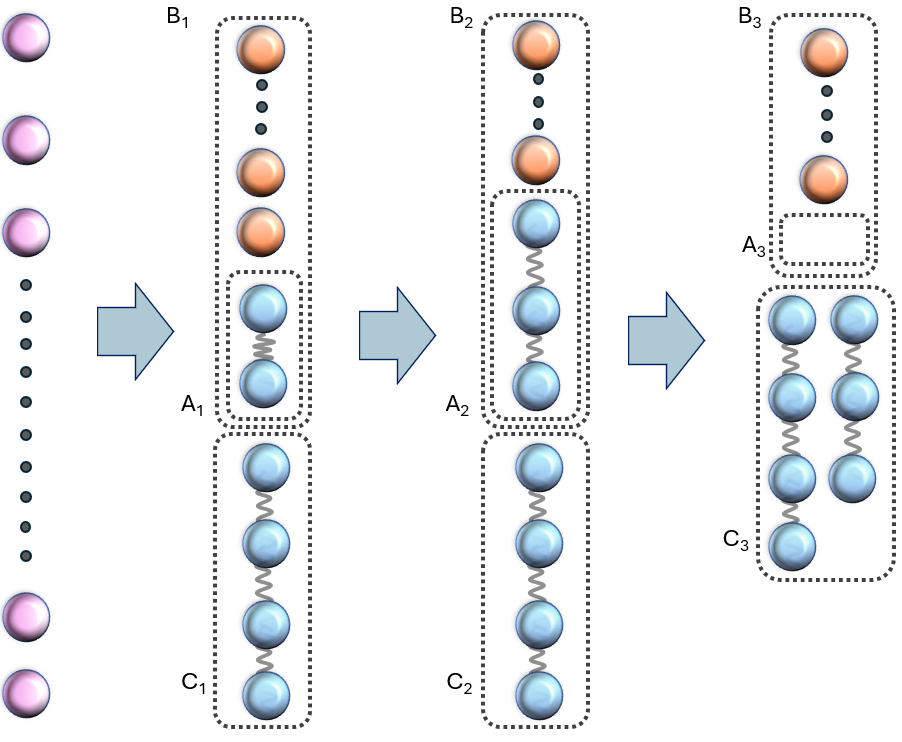}
\caption{\textbf{Illustration of Algorithm~2.}
Throughout the protocol, the qudits are partitioned into three registers
\(A_t\), \(B_t\), and \(C_t\), where \(t\) denotes the number of measurement
rounds performed so far. The register \(A_t\) contains the current candidate
antisymmetric subsystem, \(B_t\) contains the remaining qudits that have not
yet been assigned to an antisymmetric subsystem, and \(C_t\) stores the
antisymmetric subsystems extracted during previous invocations of the growth
subroutine. At each round, a qudit selected uniformly at random from
\(B_t\setminus A_t\) is tested together with the qudits in \(A_t\) using an
antisymmetric measurement. Upon a successful projection onto the
antisymmetric subspace, the selected qudit is incorporated into \(A_t\),
thereby growing the current antisymmetric subsystem. Once the stopping
criterion is satisfied, the subsystem \(A_t\) is transferred to \(C_t\), and
the procedure is repeated on the remaining qudits in \(B_t\). The sizes of
the extracted antisymmetric subsystems determine the column lengths of the
Young diagram associated with the input state.}
\label{fig:algorithm2}
\end{figure*}

Algorithm~2 performs unitary Schur sampling with high probability and admits the following bound on the number of antisymmetric measurements.

\begin{theorem}\label{Thm:Algo2performance}
For every \(0<\epsilon<1\), with probability at least
\(1-\epsilon\), Algorithm~2 correctly performs unitary Schur
sampling using at most
\[
T_{\max}
=
\left\lceil6nd^2\ln\!\left(\frac{2en}{\epsilon}\right) \right\rceil
\]
antisymmetric measurements.
More precisely, let \(T\) denote the total number of
antisymmetric measurements performed by Algorithm~2, and let
\(\mathcal F\) denote the event that the algorithm fails to
perform unitary Schur sampling correctly. Then
\[
\Pr\!\left(\{T>T_{\max}\}\cup\mathcal F\right)
\le \epsilon.
\]
\end{theorem}
 We discuss the proof for this theorem in section~\ref{sec:AnalysisofProtocol}.

\subsection{Reversibility of the Algorithms} \label{subsec: reversibility}

Both algorithms are implemented using sequences of antisymmetric
projective measurements, whose projectors commute with the
collective \(\mathrm{SU}(d)\) action. More generally, any projective
measurement with this symmetry preserves the expectation value
of \(U^{\otimes n}\), for every single-qudit unitary \(U\), when
averaged over the measurement outcomes:
\begin{align}\label{Eqn:innerproduct}
    \tr\!\left(U^{\otimes n}\sigma\right)
    =
    \tr\!\left(U^{\otimes n}\mathcal{E}(\sigma)\right),
\end{align}
where \(\mathcal{E}(\sigma)\) denotes the post-measurement state
averaged over all outcomes. This conservation property also
holds for a sequence of such measurements \cite{qubitSWAPUSS}.

Further, for PI states restricted to a single isotypic component, this result holds even when the measurement outcome is not discarded. This immediately implies that such projective measurements fully preserve PI information, because, according to Schur-Weyl duality, any PI state can be written as a linear combination of $U^{\otimes n}$ for $U \in \mathrm{SU}(d)$. 
Therefore, if we simply perform a permutation twirl after any projective measurement that respects $\mathrm{SU}(d)$ symmetry, the original PI state is restored (provided it lived in a single isotypic component). 

In the following, we present this argument and its implications more formally. We note that a similar property is satisfied by the Schur Sampling protocol, 
$\mathcal{E}_{\text{Schur}}$. Its Kraus operators $M_{\lambda} \otimes |\lambda\rangle$ commute with $U^{\otimes n}$, a property sometimes referred to as the strong symmetry of the channel, which in turn implies that it preserves inner product in Eq.~\eqref{Eqn:innerproduct}  (See Appendix \ref{AppAsubsec:Reversibility} for further discussion).\\

Consider an arbitrary $n$-qudit PI state restricted to a single isotypic component $\lambda$, i.e.,
\be\label{input}
\sigma_{\lambda}=\omega_\lambda\otimes\frac {\mathbb{I}_{m(n,\lambda)}}{m(n,\lambda)}=\mathcal{T}\big(\omega_\lambda\otimes \ket{\eta_{\lambda}}\bra{\eta_{\lambda}}\big)\ .
\ee
Here, the tensor product is with respect to the component $\mathcal{Q}^d_{\lambda} \otimes \mathcal{P}_{\lambda}$ in the decomposition of Eq.(\ref{dec}).

Thanks to the aforementioned $\mathrm{SU}(d)$ symmetry of the measurement projectors, the support of the output 
$\mathcal{E}(\sigma_{\lambda})$
is also restricted to the  subspace  $\mathcal{Q}^d_{\lambda} \otimes \mathcal{P}_{\lambda}$. Furthermore, since the measurement projectors act trivially on $\mathcal{Q}^d_{\lambda}$, for each possible outcome of the antisymmetric measurements, the post-measurement state is still of the form $\omega_{\lambda} \otimes \beta$ with respect to this decomposition, where $\beta$ depends on the outcomes. In other words, the reduced state in each irrep of $\mathrm{SU}(d)$ remains fully protected.\\

For instance, in Algorithm~2, The state, $\sigma_{BC}(T)$, of all qudits distributed across systems $B$ and $C$ (recall that $A\subseteq B$) after \(T\) rounds of measurement, is also restricted to the sector $\lambda$ of the $n$ qudits. In other words, the original input state $\sigma_{\lambda}$ can be fully recovered from the state we have at any stage of the algorithm by uniting the qudits in systems $B$ and $C$, and applying a random permutation,
\begin{align}\label{mixed}
  \mathcal{T}\left(\sigma_{BC}(T)\right)=\sigma_{\lambda}\ .
 \end{align}
 This can be seen by noting that, under random permutations, the state of the subsystem $\mathcal{P}_{\lambda}$ is mapped back to the maximally mixed state.  

For a general PI state $\tau$ that is not restricted to a single isotypic component  $\lambda$, a weaker notion of reversibility holds. Specifically, if we discard the classical register \(L\) that records the measurement outcomes, then apply a random permutation to the $n$ qudits, we recover the original PI state $\tau$, i.e.,
\be\label{rev4}
\mathcal{T}\circ\Tr_{\text{reg}}\circ\mathcal{E}(\tau)=\tau\ ,
\ee
which follows from the $\mathrm{SU}(d)$ symmetry of antisymmetric measurements, as shown in Appendix~\ref{AppAsubsec:Reversibility} (See also \cite{qubitSWAPUSS}). \color{black}
We emphasize that the above arguments rely solely on $\mathrm{SU}(d)$ symmetry and Schur-Weyl duality, and therefore admit broad generalizations. In particular, the reversibility holds for any sequence of antisymmetric measurements or SWAP tests on qudits.

\section{Analyzing the Protocol}\label{sec:AnalysisofProtocol}

In this section, we establish the correctness of Algorithm~2 and prove the performance guarantee stated in Theorem~\ref{Thm:Algo2performance}, bounding the number of antisymmetric measurements required for unitary Schur sampling. The key idea is that the subroutine finds the largest antisymmetric subsystem amongst the remaining qudits in $B$ with high probability. In the following subsection \ref{subsec:probgrowth}, we show that whenever the subsystem \(A\) has fewer qudits than the
largest antisymmetric subsystem, the probability of extending \(A\) by one
qudit is strictly positive; this is captured by Eq. \eqref{Eqn:probgrowth}. Equally important, once the subsystem reaches this size, the probability of
further growth vanishes. Thus the subsystem can continue to grow until it
reaches the size of the first column of \(\lambda \), as shown in Lemma~\ref{Thm:largest antisymm}.\\

To make this statement quantitative, suppose that the current
antisymmetric subsystem \(A\) has size \(k\) and can still be extended.
As shown in Subsection~\ref{subsec:convanalysis} on convergence analysis, after
\be \label{Eqn:Algo2cap}
R(k)=n(k+1)\ln\!\left(\frac{2n}{\epsilon}\right)
\ee
growth attempts, \(A\) increases in size from \(k\) to \(k+1\) with
probability at least \(1-\epsilon/(2n)\). Consequently, whenever the current antisymmetric subsystem \(A\) can
still be enlarged, the probability that it fails to grow before the
stopping threshold is at most \(\epsilon/(2n)\). Each successful growth
step adds one new qudit to an antisymmetric subsystem, so there can be
at most \(n\) such growth steps over the entire algorithm. A union bound
therefore implies that all required growth steps succeed before their
respective stopping thresholds with probability at least
\(1-\epsilon/2\).
When the stopping criterion is met, the qudits in \(A\) are moved to
\(C\), and the same procedure is applied recursively to the remaining
qudits. On the event above, each such subsystem \(A\) is the largest
antisymmetric subsystem supported on the remaining qudits. Hence, the
algorithm performs unitary Schur sampling with probability at least
\(1-\epsilon/2\). The detailed argument is given in
Appendix~\ref{App:convergence2}.\\

These observations admit a natural interpretation in terms of the Young
diagram. The algorithm may be viewed as a Young-diagram traversal,
proceeding down columns from top to bottom and from left to right, and we discuss this in depth in subsection \ref{subsec:YDtraversal}. A round of measurement that has an antisymmetric outcome corresponds to moving one step downward in the
current column, while reaching the bottom of a column corresponds to the subsystem having reached it maximum size. This picture (as illustrated in Fig.~\ref{Fig:YDtraversal}  ) provides a simple language to analyze the resulting stochastic process of Algorithm~2. \\

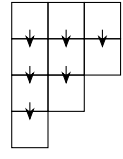
\begin{figure}[t]
\centering
\begin{tikzpicture}[scale=0.6,>=Stealth]

\def\s{0.8}

\foreach \y in {0,...,3}{
    \draw (0,-\y*\s) rectangle ++(\s,-\s);
}
\foreach \y in {0,...,2}{
    \draw (1*\s,-\y*\s) rectangle ++(\s,-\s);
}
\foreach \y in {0,...,1}{
    \draw (2*\s,-\y*\s) rectangle ++(\s,-\s);
}

\draw[->] (.5*\s,-.75*\s) -- (.5*\s,-1.25*\s);
\draw[->] (.5*\s,-1.75*\s) -- (.5*\s,-2.25*\s);
\draw[->] (.5*\s,-2.75*\s) -- (.5*\s,-3.25*\s);

\draw[->] (1.5*\s,-.75*\s) -- (1.5*\s,-1.25*\s);
\draw[->] (1.5*\s,-1.75*\s) -- (1.5*\s,-2.25*\s);

\draw[->] (2.5*\s,-.75*\s) -- (2.5*\s,-1.25*\s);

\end{tikzpicture}

\caption{
The Young-diagram traversal process associated with Algorithm~2. Each successful antisymmetric measurement extends the current antisymmetric subsystem by one qudit, corresponding to a transition to the box immediately below in the current column. Once the bottom of a column is reached, the algorithm certifies that the subsystem is maximal and continues from the top of the next column. Consequently, unitary Schur sampling may be viewed as a stochastic traversal of the Young diagram.
}
\label{Fig:YDtraversal}
\end{figure}

The above argument already establishes that the algorithm performs
unitary Schur sampling with probability at least
\(1-\frac{\epsilon}{2}\), using at most
\(
n^2d\ln\!\left(2n\epsilon^{-1}\right)
\)
antisymmetric measurements. In the following section, we obtain a
substantially sharper analysis by interpreting the algorithm as a
stochastic traversal of the Young diagram. This refined analysis captures the fact that, away from the bottom of a column, the traversal typically advances to the next box before exhausting the full cutoff \(R(k)\), and the additional structure imposed by the nonincreasing column heights of the Young diagram. In subsection \ref{subsec:convanalysis}, we show that the algorithm,  with probability at least $1-\epsilon$, performs unitary Schur sampling within
\(
6nd^2\ln\!\left(2en\epsilon^{-1}\right)
\)
 antisymmetric measurements.

\subsection{Probability of Growing the Antisymmetric subsystem} \label{subsec:probgrowth}
We now lower bound  the probability that one round of the subroutine in Algorithm~\ref{alg:growth} extends the current antisymmetric state in subsystem $A \subseteq B$. Let the system \(B\) be composed of \(n_{B}\) qudits and its state, $\sigma_B$,  supported on the \(\mu\)-isotypic component. Suppose we are at a stage where subsystem $A$ has $k$ qudits labelled $1,\cdots,k$ in an antisymmetric subspace, and we pick a random qudit $r$ from the $n_B-k$ qudits. We will now perform the antisymmetric measurement $\{ M_r, \mathbb{I}-M_r \}$, where $M_r =\Pi^-_{1,\cdots,k,r}$. Our objective is to find a lower bound for the probability that the $r$-th qudit joins the antisymmetric subsystem $A$, i.e.,\\
\begin{align}
     \mathbb{E}_r\tr{(M_r \sigma_B)}= \frac{1}{n_B-k}\sum_{r=k+1}^{n_B}\tr{(M_r \sigma_B)}.
\end{align}
We simplify this expression using the following identity regarding antisymmetric subspace projectors (see Appendix \ref{App:Concepts} for the proof): 
\begin{align}
    &M_r = \frac{1}{k+1}\Pi^-_{1,\cdots,k}\bigg(I - \sum_{i=1}^k P_{(i,r)}\bigg)\\
    &= \frac{1}{k+1}\bigg(I - \sum_{i=1}^k P_{(i,r)}\bigg)\Pi^-_{1,\cdots,k}.
\end{align}
Since the  $k$ qudits in $A \subset B$ are in the antisymmetric subspace, we have $\sigma_B = \Pi^-_{1,\cdots,k}\sigma_B\Pi^-_{1,\cdots,k}$. Then 
\begin{align}
\mathbb{E}_r \tr\!\left(M_r \sigma_B\right)
=
\frac{1}{k+1}
-
\frac{k}{k+1}
\tr\!\left(P_{(k,k+1)}\sigma_B\right).
\end{align}
By the antisymmetry of \(A\) and permutation invariance of \(B\setminus A\), all transpositions exchanging one qudit in \(A\) with one qudit in \(B\setminus A\) have the same expectation value.
Thus we need to find an upper bound for $\tr\!\left(P_{(k,k+1)}\sigma_B \right)$. We now provide a sketch of the proof and leave the details to
Appendix~\ref{App:probabgrowing singlets}. Recall that the qudit \(r\) is chosen uniformly at random
from the subsystem \(B \setminus A={k+1,\dots,n_B}\). Equivalently, we may view this
procedure as first applying a permutation twirl to \(B\setminus A\) and then selecting a
fixed qudit. Consequently, the reduced state on \(B\setminus A\) is permutation
invariant and therefore block-diagonal with respect to the Schur-Weyl
decomposition.

The reduced state on \(B\setminus A\) is generally not supported on a single isotypic
component. Rather, it admits a decomposition as a convex combination of states
supported on the isotypic components of \((\mathbb{C}^d)^{\otimes(n_B-k)}\),
\begin{align}
\sigma_{B\setminus A}
=
\sum_{\nu} p_{\nu}\sigma_{\nu},
\end{align}
where each \(\sigma_{\nu}\) is supported entirely on the
\(\nu\)-isotypic component. The Young diagrams \(\nu\) for which
\(p_{\nu}>0\) are not arbitrary. Since the \(k\) qudits in $A$ are constrained
to lie in the fully antisymmetric subspace and the full system $B$ is supported on
the \(\mu\)-isotypic component, the admissible diagrams \(\nu\) are
restricted by the Pieri's branching rule \ref{subsection:Pieri}. 
More explicitly, the admissible Young diagrams $\nu$ are precisely those that
can be obtained from $\mu$ by removing $k$ boxes, no two from the same
row, such that the remaining boxes still form a Young diagram.

Next, we exploit the symmetry of the state $\sigma_B$. The permutation symmetries of \(A\) and $B\setminus A $ divide the transpositions into three classes:

\begin{itemize}
\item If \(a,b\le k\), then the corresponding qudits lie entirely in the
antisymmetric subsystem, therefore,
\begin{align}
\tr(P_{(a,b)}\sigma_B)=-1.
\end{align}

\item If \(a,b>k\), then both qudits lie entirely within the subsystem
  \(B\setminus A\). Since \(\sigma_{B\setminus A}\) is permutation invariant and all transpositions form a single conjugacy class of \(\mathbb S_{n_B-k}\), their expectation values are equal, i.e.,
\begin{equation}
  \tr(P_{(a,b)}\sigma_B) = \sum_{\nu}p_{\nu}\chi_{\nu}((1,2)), \quad \forall a,b > k,
\end{equation}
where \(\chi_{\nu}\) denotes the irreducible character of \(\mathbb S_{n_B-k}\).
Although we do not determine this expression exactly, it is constrained by
the set of admissible Young diagrams (permitted by Pieri's rule) that may appear in the Schur-Weyl
decomposition of \(B\setminus A\). Consequently, we can bound these transposition
expectations by optimizing over the corresponding admissible irreducible
representations. This amounts to optimizing the character of the transposition conjugacy class over all admissible Young diagrams.

\item Finally, the total sum over all transpositions,
\begin{align}
\sum_{a<b}\tr(P_{(a,b)}\sigma_{B}) \propto \chi_{\mu}((1,2))
\end{align}
is fixed by the isotypic component \(\mu\) in which \(\sigma\) resides. In
particular, this quantity is determined by the character of the
transposition conjugacy class in the irrep \(\mu\).
\end{itemize}

The above observations allow us to upper bound \(\tr\!\left(P_{(k,k+1)}\sigma_B \right)\), which leads to the following lower bound for the probability of extending the antisymmetric subsystem. 
\begin{align}
\mathbb{E}_r \tr(M_r\sigma _B)
\ge
\frac{
\sum_{i=k+1}^{d}\mu_i
}{
(n_{B}-k)(k+1)
}. \label{Eqn:probgrowth}
\end{align}
 The proof is given in
Appendix~\ref{App:probabgrowing singlets}.\\

\subsection{Young-Diagram Traversal Process}\label{subsec:YDtraversal}
 The
lower bound above may therefore be interpreted directly in terms of the
geometry of the original Young diagram \(\lambda\). Recall that each antisymmetric subsystem identified by the protocol
corresponds to a column of the Young diagram. Consequently, if the antisymmetric subsystems corresponding to the first
\(l-1\) columns have already been separated and the current antisymmetric
subsystem contains \(k\) qudits, then the state of the algorithm may be
identified with the box \((k,l)\) of the Young diagram. In this picture,
extending the antisymmetric subsystem by one qudit corresponds to moving from
\((k,l)\) to \((k+1,l)\), while separating a completed subsystem corresponds
to moving to the top of the next column (See Fig. \ref{fig:young-traversal}).\\

For a box \(b=(k,l)\in \lambda\), we define two associated quantities. First, define 
\begin{align}
T_{\lambda}(l)
:=
\sum_{i=1}^{d}(\lambda_i-l+1)_+\ ,
\end{align}
where
\(
(x)_+ := \max\{x,0\}
\).
 The quantity $T_{\lambda}(l)$ counts the number of boxes remaining after removing the first \(l-1\)
columns of the Young diagram. Next, define
\begin{align}
T_{\lambda}(k,l)
:=
\sum_{i=k+1}^{d}(\lambda_i-l+1)_+\ .
\end{align}
Equivalently, \(T_{\lambda}(k,l)\) is the size of the subdiagram lying strictly below row \(k\) after removing the first \(l-1\) columns.\\

\begin{figure}[ht]
\centering

\includegraphics[width=0.95\linewidth]{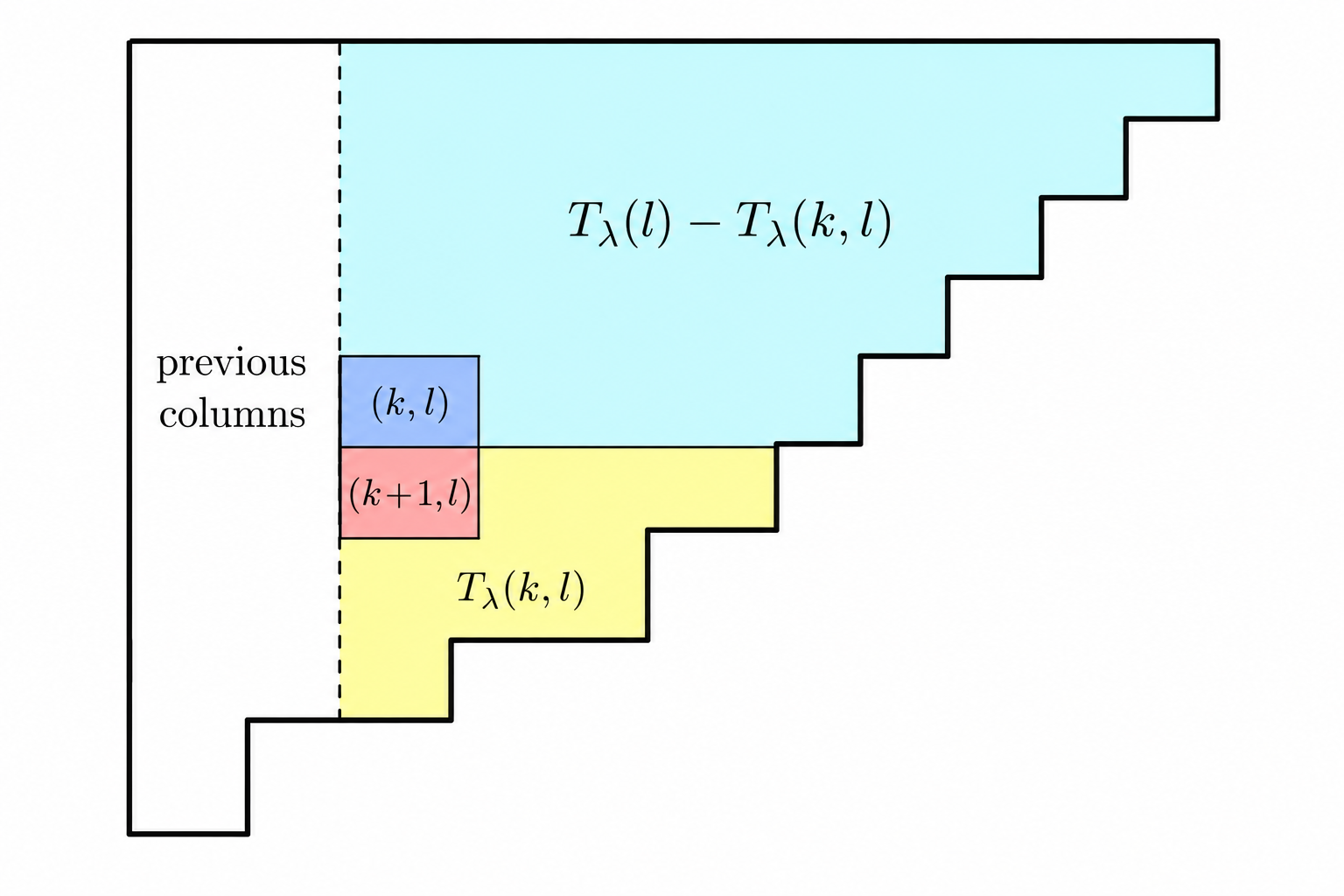}
\caption{ Illustration of the Young-diagram traversal process. The white region corresponds to columns that have already been separated by previous invocations of the growth subroutine. The blue square denotes the current state \((k,l)\), while the red square represents the candidate next state \((k+1,l)\). The yellow region (including $(k+1,l)$) contains the \(T_{\lambda}(k,l)\) boxes contributing to the probability of extending the current antisymmetric subsystem, whereas the cyan region (including $(k,l)$) contains the remaining \(T_{\lambda}(l)-T_{\lambda}(k,l)\) boxes of the current Young diagram. A successful antisymmetric measurement corresponds to the transition \((k,l)\rightarrow(k+1,l)\).}
\label{fig:young-traversal}
\end{figure}

Expressed in this language, the bound in Eq.~\eqref{Eqn:probgrowth} becomes 
\begin{align}
\Pr\!\left((k+1,l)\mid(k,l)\right)
\ge p_{(k,l)},
\label{Eq:probYDstep}
\end{align}
where
\begin{align}
p_{(k,l)} \coloneqq
\frac{
T_{\lambda}(k,l)
}{
\bigl(T_{\lambda}(l)-k\bigr)(k+1)
}\ .
\end{align}
Here \(\Pr\!\left((k+1,l)\mid(k,l)\right)\) is shorthand for the conditional
probability that the Young-diagram traversal moves from \((k,l)\) to
\((k+1,l)\), for any admissible measurement history that reaches the box
\((k,l)\). If \((k+1,l)\in\lambda\), then \(T_{\lambda}(k,l)>0\), so the probability of extending the antisymmetric subsystem is strictly positive. Moreover, since there is at least one box below \((k,l)\), we obtain the simpler bound:
\begin{align}\label{Eqn:growthsimplebound} 
p_{(k,l)}
\ge
\frac{1}{n(k+1)}.
\end{align}
This lower bound motivates the cap on the number of growth attempts at row \(k\) introduced in Eq.~\eqref{Eqn:Algo2cap}.

Conversely, once the bottom of a column is reached,
\(T_{\lambda}(k,l)=0\), and the lower bound vanishes. Thus, provided the algorithm is permitted to make sufficient attempts at each box, with high probability, the algorithm continues descending the current column till it reaches the bottom. This is consistent with Lemma~\ref{Thm:largest antisymm}, which implies that the current subsystem \(A\) is already a largest antisymmetric subsystem of the remaining state and cannot be extended further. \\

The preceding observations reduce the analysis of
Algorithm~2 to the study of a stochastic process on the Young diagram
\(\lambda\).  The box \((k,l)\) records only the current position of the traversal; the conditional quantum state may also depend on the preceding measurement history. However, conditioned on any admissible history leading to \((k,l)\), the next growth probability obeys the lower bound in Eq.~\eqref{Eq:probYDstep}. If the traversal reaches all the boxes in the Young diagram, then the algorithm would have performed unitary Schur sampling. The transition probabilities are controlled by the
local geometry of the diagram through Eq.~(\ref{Eq:probYDstep}). Consequently, the
convergence rate of the protocol may be analyzed entirely through the
properties of this Young-diagram traversal process.

\subsection{Convergence Analysis} \label{subsec:convanalysis}
The Young-diagram traversal defines a stochastic process in which the algorithm attempts to move from one box to the next. Since the Young diagram is not known in advance, repeated unsuccessful attempts do not determine with certainty whether the traversal has reached the bottom of the current column, or whether it has simply not yet moved to the next box. Consequently, the algorithm requires a stopping criterion that both limits the runtime and certifies, with high probability, that every box is reached. The threshold is chosen so that, whenever another box exists below \((k,l)\), the traversal reaches it within \(R(k)\) rounds (see Eq.~\ref{Eqn:Algo2cap}) with probability at least \(1-\frac{\epsilon}{2n}\), which follows from Eq.~\eqref{Eqn:growthsimplebound}. If the traversal has not advanced after \(R(k)\) rounds, the algorithm treats the current box as the bottom of its column and proceeds to the next column.\\

In Appendix~\ref{App:Number of rounds}, we prove that the algorithm successfully traverses the Young Diagram, i.e., performs unitary Schur sampling, in at most \(6nd^2
\ln\!\left(2en/\epsilon\right)\) rounds of antisymmetric measurements, with probability at least $1-\epsilon$.
Here, we present a proof sketch of the convergence analysis. 
 We associate a random variable \(X_{(k,l)}\) with every box \((k,l)\) of the Young diagram. The variable \(X_{(k,l)}\) denotes the number of rounds that would be required for the traversal to move from \((k,l)\) to \((k+1,l)\) if no stopping threshold were imposed.

We let \(Y_{(k,l)}\) denote the actual number of rounds performed by Algorithm~2 while at the box \((k,l)\). For a non-bottom box,
\begin{align}
Y_{(k,l)}
=
\min\{X_{(k,l)},R(k)\}.
\end{align}
For a bottom box, the probability of further growth is zero. Hence,
\(Y_{(k,l)}=R(k)\) if a certification stage is required, and
\(Y_{(k,l)}=0\) if the current column has already reached its known upper bound, i.e., the length of the previous column.\\

The total number of rounds of measurements is simply
\begin{align}
T
=
\sum_{(k,l)\in\lambda}
Y_{(k,l)}\ .
\end{align}
From Eq.~(\ref{Eq:probYDstep}),  the probability that a round of measurement transitions from $(k,l)$ to $(k+1,l)$ is  at least $p_{(k,l)}$.

Assuming the algorithm has correctly reached $(k,l)$, i.e., has been able to detect all existing antisymmetric subsystems, the probability that after $t$ attempts (antisymmetric measurements) it has not reached the next box $(k+1,l)$ is bounded by 
\begin{align}
\Pr\!\left(
X_{(k,l)}>t
\;\middle|\;
\text{for any admissible path}
\right)
\le
(1-p_{(k,l)})^t.
\end{align}
More formally, let \(\mathcal{X}\) denote the set of fillings of the Young diagram \(\lambda\) with positive integers. For any \(x\in\mathcal{X}\), we denote the entry corresponding to the box \((r,s)\) by \(x_{(r,s)}\). We order the boxes in column-major order, consistent with the traversal of Algorithm~2, and denote by \(x_{<(k,l)}\) the restriction of \(x\) to the boxes preceding \((k,l)\) in this ordering.
For every non-bottom box \((k,l)\), the lower bound \(p_{(k,l)}\) on the transition probability holds conditioned that algorithm has visited all the boxes that precede \((k,l)\). In particular,
define $ \mathbf X_{<(k,l)} := \{X_{(r,s)}:(r,s)<(k,l),\,r<h(s)\},$
and
\begin{align}
 &E_{(k,l)}:= \nonumber \\
 & \left\{ x_{<(k,l)}: x_{(r,s)}\le R(r) \text{ for all }(r,s)<(k,l),\,r<h(s) \right\}.
\end{align}
We have used $  (r,s) < (k,l) $ as a shorthand to indicate that the box $(r,s)$ comes before $(k,l)$ in  the column-major order. Here $h(s)$ is the number of boxes in the $s$th column. Then, from Eq.~\eqref{Eqn:probgrowth}, for every \(x_{<(k,l)}\in E_{(k,l)}
\),
\begin{align}
\label{Eqn:tail bound}
\Pr\!\left(
X_{(k,l)}>t
\;\middle|\;
\mathbf X_{<(k,l)}=x_{<(k,l)}
\right)
\le
(1-p_{(k,l)})^t.
\end{align}

\color{black}

Thus, Eq.~\eqref{Eqn:tail bound} holds conditionally on every exact admissible history, not merely after averaging over histories that reach \((k,l)\).\\

 This bound has a natural interpretation in terms of geometric waiting times: conditioned on the past, the waiting time \(X_{(k,l)}\) has tails no larger than those of a geometric random variable with success parameter \(p_{(k,l)}\). Accordingly, for each non-bottom box \((k,l)\), we define independent geometric random variables
\be
Z_{(k,l)}\sim\mathrm{Geom}(p_{(k,l)}),
\ee
so that
\be
\Pr\!\left(Z_{(k,l)}>t\right)=(1-p_{(k,l)})^t.
\ee
Thus, conditioned on any successful history leading to \((k,l)\), \(X_{(k,l)}\) is stochastically dominated by \(Z_{(k,l)}\).\\

The tail bounds in Eq.\eqref{Eqn:tail bound} allow us to compare the contribution of the non-bottom boxes to the total traversal time \(T\), to the sum of random variables \(Z_{(k,l)}\). We justify this comparison in Theorem~\ref{Thm:Domination} of Sec.~\ref{App:probfacts}. Consequently, the measurement-complexity analysis reduces to bounding the expectation and concentration of $\sum_{(k,l)} Z_{(k,l)}$. We then apply standard concentration bounds for sums of independent geometric random variables to obtain a high-probability bound on \(T\). This leads to Theorem~\ref{Thm:Algo2performance}, and we provide the explicit proof is given in Appendix~\ref{App:Number of rounds}.

\section{Antisymmetric measurements via SWAP tests}
\label{section:SWAPtoAntisymm}

We now describe how pairwise SWAP tests can be used to implement
antisymmetric measurements. A state lies in the fully antisymmetric
subspace if and only if a SWAP test on every pair yields the
antisymmetric outcome with probability one. This suggests a way
to approximately measure whether a state lies in that subspace:
test uniformly random pairs and declare the state antisymmetric
if sufficiently many consecutive tests yield the antisymmetric
outcome. If only the outcome statistics are of interest, this
procedure distinguishes the fully antisymmetric subspace from
its orthogonal complement with arbitrarily small error.

Reproducing the corresponding post-measurement state is more
subtle. For a general input state on three or more qudits, a
pairwise SWAP measurement can destroy coherences within the
orthogonal complement of the fully antisymmetric subspace,
even though the ideal antisymmetric measurement preserves
these coherences. Thus, for arbitrary inputs, a sequence of
SWAP tests does not directly simulate the post-measurement
channel of the antisymmetric measurement.

Fortunately, the states encountered in Algorithm~2 have
additional structure: the $k$ qudits in $A$ are already fully
antisymmetric. When testing whether an additional qudit can
be added to $A$, the joint input state is therefore supported
on the subspace of vectors
$|\psi\rangle\in(\mathbb{C}^d)^{\otimes(k+1)}$ satisfying
\be\label{ldj}
(\Pi_{1^k}\otimes\mathbb{I})|\psi\rangle
=
|\psi\rangle,
\ee
where the qudits in $A$ occupy the first $k$ positions.
Such states have support only on the isotypic components
labelled by
\be
\lambda_0=(1^{k+1}),
\qquad
\lambda_1=(2,1^{k-1}).
\ee
The ideal antisymmetric measurement preserves this promise.
Indeed, the fully antisymmetric subspace of $k+1$ qudits
is contained in the subspace in which the first $k$
qudits are fully antisymmetric, so
\begin{equation}
(\Pi_{1^k}\otimes\mathbb I)\Pi_{1^{k+1}}
=
\Pi_{1^{k+1}}
=
\Pi_{1^{k+1}}(\Pi_{1^k}\otimes\mathbb I).
\end{equation}
Thus, after either outcome of the ideal measurement,
the first $k$ qudits remain fully antisymmetric.

This observation suggests a recovery procedure for the
disturbance caused by pairwise SWAP tests.
Once a symmetric SWAP-test outcome is observed, the
$\lambda_0$ component is eliminated, leaving the state
entirely in the $\lambda_1$ sector.
Interestingly, within this sector, the promise that the first $k$
qudits are fully antisymmetric fixes the multiplicity
subsystem $\mathcal{P}_{\lambda_1}$ to a one-dimensional subspace.
Moreover, every SWAP-test projector and permutation
acts trivially on the $\mathrm{SU}(d)$ irrep subsystem.
Consequently, conditioned on a successful projection
of the first $k$ qudits back onto their fully
antisymmetric subspace, the desired post-measurement
state is restored exactly, including its correlations
with any external system.
The recovery procedure below implements this idea
approximately using further pairwise SWAP tests.\\

The case $k=1$ is implemented exactly by a single SWAP test,
so we assume $2\le k\le d$ below.
For these inputs, we introduce a protocol specified by
positive integers $L_\ast$, $L$, and $T$, with $T\ge L$.
Here, $L_\ast$ is the certification threshold, $L$ is the
recovery threshold, and $T$ is the maximum number of SWAP
tests performed during recovery.

\begin{enumerate}
\item Apply a uniformly random permutation of all $k+1$ qudits.

\item Perform SWAP tests on independently and uniformly chosen
pairs among qudits $1,\ldots,k+1$, stopping at the first symmetric
outcome or after $L_\ast$ consecutive antisymmetric outcomes.
In the latter case, return the current quantum state together
with the outcome $\lambda_0$, and terminate.

\item If a symmetric outcome is obtained, record the outcome
$\lambda_1$. Move the pair that produced this outcome to positions
$k$ and $k+1$, and begin recovery.

\item \textbf{Recovery:}
Initialize a counter to zero and perform SWAP tests on
independently and uniformly chosen pairs among the first $k$
qudits. After each antisymmetric outcome, increment the counter.
After each symmetric outcome, move the tested pair to positions
$k$ and $k+1$ and reset the counter to zero.

Stop recovery as soon as the counter reaches $L$ or a total
of $T$ recovery SWAP tests have been performed.
Apply a uniformly random permutation of the first $k$ qudits,
return the resulting quantum state of $k+1$ qudits together with the recorded
outcome $\lambda_1$, and terminate.
\end{enumerate}

The protocol performs at most $L_\ast+T$ SWAP tests.
For fixed $k$, the theorem below shows that it reproduces
both the outcome statistics and the post-measurement states
of the ideal antisymmetric measurement on the specified
inputs as $L_\ast\to\infty$, $L\to\infty$, and $T-L\to\infty$.

To quantify the error for finite parameters, let
$\mathcal{F}_{\mathrm{ideal}}$ denote the channel associated
with the ideal projective measurement
$\{\Pi_{\lambda_0},\mathbb{I}-\Pi_{\lambda_0}\}$, including
a register that records its outcome.
For inputs satisfying Eq.~\eqref{ldj}, this channel acts as
\begin{align}
\mathcal{F}_{\mathrm{ideal}}(\rho)
={}&
\Pi_{\lambda_0}\rho\Pi_{\lambda_0}
\otimes|\lambda_0\rangle\langle\lambda_0|_{\mathrm{reg}}
\nonumber\\
&+
\Pi_{\lambda_1}\rho\Pi_{\lambda_1}
\otimes|\lambda_1\rangle\langle\lambda_1|_{\mathrm{reg}},
\end{align}
where $\mathrm{reg}$ is a qubit register with orthonormal
basis states $|\lambda_0\rangle$ and $|\lambda_1\rangle$.

We quantify the approximation error by the diamond distance
between the ideal channel and the channel implemented by
the protocol, restricted to inputs satisfying
Eq.~\eqref{ldj}. To express this restriction, define
\be
\mathcal{K}(\rho)
=
(\Pi_{1^k}\otimes\mathbb{I})\rho
(\Pi_{1^k}\otimes\mathbb{I}).
\ee
This completely positive, trace-nonincreasing map projects
onto the subspace in which the first $k$ qudits are fully
antisymmetric.

\begin{theorem}\label{THM:SWAP}
Let \(2\le k\le d\) and \(0<\delta<1\).
The parameters of the above protocol can be chosen so that it
implements a channel \(\mathcal{F}\) satisfying
\begin{equation}
\frac{1}{2}
\left\|
\mathcal{F}_{\mathrm{ideal}}\circ\mathcal{K}
-
\mathcal{F}\circ\mathcal{K}
\right\|_\diamond
\le \delta,
\end{equation}
using at most \(T+L_\ast\) SWAP tests, with
\begin{equation}
T+L_\ast \le 4k\ln\frac{3}{\delta}.
\end{equation}
These guarantees hold for the choice
\begin{align}
L_\ast
&=\left\lceil k\ln\frac{3}{\delta}\right\rceil,
\qquad
L=\left\lceil (k-1)\ln\frac{3}{\delta}\right\rceil,
\\
T
&=L+\left\lceil
\frac{2k(k-1)}{k+1}\ln\frac{3}{\delta}
\right\rceil.
\end{align}
\end{theorem}

We next analyze the protocol and prove this theorem.

\subsection{How the protocol works}

Averaging over the initial uniformly random permutation removes
coherences between the $\lambda_0$ and $\lambda_1$ sectors, as
does the ideal measurement. Moreover, all permutations and
SWAP-test projectors commute with the collective
$\mathrm{SU}(d)$ action and therefore preserve each isotypic
sector. It thus suffices to analyze the two input sectors
separately.

On the $\lambda_0$ sector, the protocol agrees exactly with the
ideal measurement. We therefore restrict our attention to
inputs supported on the image of the projector
\begin{equation}
Q
=
(\Pi_{1^k}\otimes\mathbb{I})\Pi_{\lambda_1}
=
\Pi_{\lambda_1}(\Pi_{1^k}\otimes\mathbb{I}).
\end{equation}
This projector selects the subspace of the $\lambda_1$ isotypic
component satisfying the additional constraint in
Eq.~\eqref{ldj}.

The structure of $Q$ significantly simplifies the analysis.
Within the Schur--Weyl sector
$\mathcal{Q}_{\lambda_1}\otimes\mathcal{P}_{\lambda_1}$, we have
\begin{equation}\label{hdjskl}
Q
=
\mathbb{I}_{\mathcal{Q}_{\lambda_1}}
\otimes|\alpha\rangle\langle\alpha|
\end{equation}
for a normalized vector
$|\alpha\rangle\in\mathcal{P}_{\lambda_1}$.
Indeed, the branching rule for restriction from
$\mathbb{S}_{k+1}$ to $\mathbb{S}_k$ is multiplicity free,
and the one-dimensional sign representation $(1^k)$ occurs
exactly once in the restriction of
$\mathcal{P}_{\lambda_1}$.
Consequently, the projector onto this representation has the
form in Eq.~\eqref{hdjskl}; see the appendix for details.

Thus, for inputs satisfying Eq.~\eqref{ldj}, conditioned on the
ideal measurement $\mathcal{F}_{\mathrm{ideal}}$ returning
outcome $\lambda_1$, the multiplicity subsystem is in the pure
state $|\alpha\rangle$, independently of the input.

To evaluate the diamond-norm error, we allow the input to be
entangled with an arbitrary reference system $R$.
For any normalized input supported on the range of
$Q\otimes\mathbb{I}_R$, Eq.~\eqref{hdjskl} implies the
factorization
\begin{equation}
\rho
=
\rho_{\mathcal{Q}_{\lambda_1}R}
\otimes|\alpha\rangle\langle\alpha|_{\mathcal{P}_{\lambda_1}},
\end{equation}
where we have reordered the tensor factors for convenience.

Every permutation and SWAP-test projector is
$\mathrm{SU}(d)$-invariant and hence, within this sector,
has the form
$\mathbb{I}_{\mathcal{Q}_{\lambda_1}}\otimes A$.
The protocol also acts trivially on $R$.
Therefore, although the initial random permutation and
subsequent measurements generally change the multiplicity
state, the joint state on
$\mathcal{Q}_{\lambda_1}\otimes\mathcal{H}_R$
remains unchanged and uncorrelated with the multiplicity
subsystem and the output register.

Consequently, the output factors as
\begin{equation}
\rho_{\mathcal{Q}_{\lambda_1}R}
\otimes\tau_{\mathcal{P}_{\lambda_1}\mathrm{reg}},
\end{equation}
while the ideal output has the same first factor and
\begin{equation}
\tau_{\mathrm{ideal}}
=
|\alpha\rangle\langle\alpha|
\otimes|\lambda_1\rangle\langle\lambda_1|_{\mathrm{reg}}.
\end{equation}
Since tensoring with a density operator preserves the trace
norm, the error analysis reduces to comparing these reduced
output states on
$\mathcal{P}_{\lambda_1}\otimes\mathbb{C}^2_{\mathrm{reg}}$.

We first analyze an auxiliary version of the protocol in which
the recovery stage is allowed to run until it terminates
naturally. The effect of imposing the cutoff $T$ will be
bounded separately. Writing $\tau_L$ for the reduced output
state of this auxiliary protocol, with its dependence on
$L_\ast$ suppressed, we have
\begin{align}
\tau_L
={}&
p_{\mathrm{succ}}(L_\ast)\,\sigma_L
\otimes|\lambda_1\rangle\langle\lambda_1|_{\mathrm{reg}}
\nonumber\\
&+
[1-p_{\mathrm{succ}}(L_\ast)]\,\sigma'_L
\otimes|\lambda_0\rangle\langle\lambda_0|_{\mathrm{reg}}.
\end{align}
Here, $\sigma_L$ and $\sigma'_L$ are the normalized conditional
output states of the multiplicity subsystem for the two
recorded outcomes. The quantity $p_{\mathrm{succ}}(L_\ast)$
is the probability that step~2 yields a symmetric outcome
before $L_\ast$ consecutive antisymmetric outcomes occur.
In this case, the protocol records $\lambda_1$ and proceeds
to recovery.

Lemma~\ref{lemma:21k-1random SWAP} shows that, for every
normalized state in the $\lambda_1$ sector, a SWAP test on a
uniformly random pair yields the symmetric outcome with
probability $1/k$. Since each test preserves the $\lambda_1$
sector and each pair is chosen independently and uniformly,
this probability remains $1/k$ conditional on any preceding
sequence of antisymmetric outcomes. Therefore, the probability
that step~2 terminates without observing a symmetric outcome
is $[(k-1)/k]^{L_\ast}$, giving
\begin{equation}\label{Eqn:p_succBound}
p_{\mathrm{succ}}(L_\ast)
=
1-\left(\frac{k-1}{k}\right)^{L_\ast}
\ge 1-e^{-L_\ast/k}.
\end{equation}
Consequently, for fixed $k$, $p_{\mathrm{succ}}(L_\ast)$
converges to one exponentially fast as $L_\ast\to\infty$.

At the end of recovery, after obtaining $L$ consecutive
antisymmetric outcomes, the auxiliary protocol applies a
uniformly random permutation of the first $k$ qudits.
Averaging over this permutation makes the output block diagonal
with respect to $\Pi_{1^k}\otimes\mathbb{I}$ and its orthogonal
complement. Within the $\lambda_1$ sector, this projector acts as
\[
Q
=
\mathbb{I}_{\mathcal{Q}_{\lambda_1}}
\otimes|\alpha\rangle\langle\alpha|.
\]
Consequently, $\sigma_L$ commutes with
$|\alpha\rangle\langle\alpha|$.

Moreover, the restriction of
$\mathcal{P}_{\lambda_1}$ to $\mathbb{S}_k$ consists of the
one-dimensional sign representation, spanned by
$|\alpha\rangle$, and an inequivalent irreducible representation
of dimension $k-1$. By Schur's lemma, the twirl makes the state
proportional to the identity on each of these two subspaces.
Therefore,
\begin{equation}\label{Eqn:sigmaLdecomp}
\sigma_L
=
q(L)|\alpha\rangle\langle\alpha|
+(1-q(L))
\frac{
\mathbb{I}_{\mathcal{P}_{\lambda_1}}
-
|\alpha\rangle\langle\alpha|}
{k-1}
,
\end{equation}
where
\[
q(L)=\langle\alpha|\sigma_L|\alpha\rangle
\]
is the probability that an antisymmetric
measurement on the first $k$ qudits at the end of recovery
returns the antisymmetric outcome. We compute this probability
below.

For the auxiliary protocol, the error is completely determined
by the two probabilities $p_{\mathrm{succ}}(L_\ast)$ and $q(L)$.
Indeed, the two register states are orthogonal, and
Eq.~\eqref{Eqn:sigmaLdecomp} is block diagonal with respect to
$|\alpha\rangle\langle\alpha|$ and its orthogonal complement.
Therefore,
\begin{equation}
D\!\left(
\tau_L,
|\alpha\rangle\langle\alpha|
\otimes|\lambda_1\rangle\langle\lambda_1|_{\mathrm{reg}}
\right)
=
1-p_{\mathrm{succ}}(L_\ast)q(L),
\end{equation}
where $D(\rho,\sigma)=\frac12\|\rho-\sigma\|_1$.
The unchanged factor on
$\mathcal{Q}_{\lambda_1}\otimes\mathcal{H}_R$
does not affect this trace distance.

To calculate $q(L)$, we use an important reset property.
Whenever a SWAP test yields the symmetric outcome, moving the
tested pair to positions $k$ and $k+1$ puts the multiplicity
subsystem in the same pure state, independently of its state
before that test. To see this, define
\begin{equation}
\widetilde{Q}
=
(\mathbb{I}^{\otimes(k-1)}\otimes\Pi_{(2)})
\Pi_{\lambda_1}
=
\Pi_{\lambda_1}
(\mathbb{I}^{\otimes(k-1)}\otimes\Pi_{(2)}),
\end{equation}
where $\Pi_{(2)}$ projects qudits $k$ and $k+1$ onto their
symmetric subspace. Within the Schur--Weyl sector
$\mathcal{Q}_{\lambda_1}\otimes\mathcal{P}_{\lambda_1}$,
we have
\begin{equation}
\widetilde{Q}
=
\mathbb{I}_{\mathcal{Q}_{\lambda_1}}
\otimes|\beta\rangle\langle\beta|
\end{equation}
for a normalized vector\footnote{This can be seen directly from representative states. The space
 is spanned by the \(\mathcal P^{k+1}_{(2,1^{k-1})}\) component the states
\[
v_i:=P_{(1,i)}
\left(
\ket{0}_1\otimes\ket{\Lambda}_{2,\ldots,k+1}
\right),
\quad i=1,\ldots,k+1,
\]
where \(\ket{\Lambda}=\ket{0\wedge \cdots \wedge k-1}\). Under a symmetric
projection on qudits \(1\) and \(2\), all states with \(i>2\) are
annihilated, while the two remaining states \(v_1\) and \(v_2\) are
projected onto the same vector, proportional to \(v_1+v_2\). Hence the
symmetric projection has one-dimensional image on
\(\mathcal P^{k+1}_{(2,1^{k-1})}\).}
$|\beta\rangle\in\mathcal{P}_{\lambda_1}$.

Indeed, all pairwise symmetric projectors are conjugate under
permutations. By Lemma~\ref{lemma:21k-1random SWAP}, their
average on $\mathcal{P}_{\lambda_1}$ is
$\mathbb{I}_{\mathcal{P}_{\lambda_1}}/k$.
Since $\dim\mathcal{P}_{\lambda_1}=k$, each of these projectors
has trace, and hence rank, equal to one on the multiplicity
subsystem.

Thus, after every symmetric outcome and subsequent move
of the tested pair, the multiplicity subsystem is reset to
$|\beta\rangle$. This preserves the factorization established
above: for the admissible inputs, the multiplicity subsystem
remains uncorrelated with
$\mathcal{Q}_{\lambda_1}\otimes\mathcal{H}_R$
throughout the protocol.
The pure states $|\alpha\rangle\langle\alpha|$ and
$|\beta\rangle\langle\beta|$ are fixed independently of the
input and measurement history. Each recovery attempt starts
from $|\beta\rangle$ and ends either with another symmetric
outcome, which triggers a reset, or with $L$ consecutive
antisymmetric outcomes, which complete recovery.

Define the trace-nonincreasing completely positive map
\begin{equation}
\mathcal{E}(X)
=
\frac{1}{\binom{k}{2}}
\sum_{1\le i<j\le k}
\Pi_{ij}^{-}X\Pi_{ij}^{-},
\end{equation}
where $\Pi_{ij}^{-}$ projects qudits $i$ and $j$ onto their
antisymmetric subspace.
The unnormalized state after $L$ consecutive antisymmetric
outcomes is obtained by applying $\mathcal{E}^L$.
Consequently,
\begin{equation}
q(L)
=
\frac{\operatorname{Tr}\!\left(
Q\mathcal{E}^{L}(\widetilde{Q})
\right)}
{\operatorname{Tr}\!\left(
\mathcal{E}^{L}(\widetilde{Q})
\right)}.
\end{equation}
Using the unnormalized projector $\widetilde{Q}$ in this ratio
is legitimate because the unchanged
$\mathcal{Q}_{\lambda_1}$ factor contributes the same
multiplicative factor to the numerator and denominator.
The final twirl preserves this probability.

Let
\begin{equation}
a=|\langle\alpha|\beta\rangle|^2=\frac{k+1}{2k},
\qquad
r=\frac{k-2}{k-1},
\end{equation}
where the overlap is established in
Lemma~\ref{lem:recovery-probability}.
The calculation in Appendix~\ref{AppSWAP:q(l)} gives
\begin{align}
\operatorname{Tr}\!\left(
Q\mathcal{E}^{L}(\widetilde{Q})
\right)
&=
\operatorname{Tr}(Q\widetilde{Q})
=ma,
\\
\operatorname{Tr}\!\left(
\mathcal{E}^{L}(\widetilde{Q})
\right)
&=
m\left[a+(1-a)r^L\right],
\end{align}
where $m=\dim\mathcal{Q}_{\lambda_1}$.
Hence,
\begin{equation}\label{Eqn:q(L)}
q(L)
=
\frac{a}{a+(1-a)r^L}
=
\left[
1+\frac{k-1}{k+1}
\left(\frac{k-2}{k-1}\right)^L
\right]^{-1}.
\end{equation}
In particular,
\begin{equation}\label{Eqn:q(L) bound}
1-q(L)
\le
\frac{k-1}{k+1}
\left(\frac{k-2}{k-1}\right)^L
\le
e^{-L/(k-1)}.
\end{equation}
Each recovery attempt succeeds with probability,
\begin{equation}\label{eq:15}
  \frac{\operatorname{Tr}(\mathcal{E}^{L}(\widetilde{Q}))}{\operatorname{Tr}(\widetilde{Q})}=a+(1-a)r^L\ge a>0,
\end{equation}
 therefore, the auxiliary recovery procedure terminates with probability one.

We now impose the recovery budget $T\ge L$.
Let $N$ denote the number of SWAP tests needed for natural
completion of recovery, starting from $|\beta\rangle$, and set
\begin{equation}
p_{\mathrm{cut}}=\Pr(N>T).
\end{equation}
Here, completion on the $T$-th test is counted as natural
completion.

Every successful recovery attempt consists of exactly $L$
antisymmetric outcomes and produces the same conditional
output state $\sigma_L$, independently of all preceding
failed attempts. Conditioning on $N\le T$ only restricts when
this final attempt can begin; it does not change its
conditional output state.
Therefore, the actual recovery output has the form
\begin{equation}
\widetilde{\sigma}_{L,T}
=
(1-p_{\mathrm{cut}})\sigma_L
+
p_{\mathrm{cut}}\omega,
\end{equation}
where $\omega$ is the conditional output state when the cutoff
is reached. By convexity of the trace distance and
Eq.~\eqref{Eqn:sigmaLdecomp},
\begin{align}\label{Eqn:sigmaLTtracedist}
D\!\left(
\widetilde{\sigma}_{L,T},
|\alpha\rangle\langle\alpha|
\right)
&\le
(1-p_{\mathrm{cut}})[1-q(L)]
+p_{\mathrm{cut}}
\nonumber\\
&=
1-(1-p_{\mathrm{cut}})q(L).
\end{align}
The cutoff probability satisfies
\begin{equation}\label{Eqn:recovery-tail}
p_{\mathrm{cut}}
\le
e^{-c(T-L)},
\qquad
c=a(1-r)=\frac{k+1}{2k(k-1)}.
\end{equation}
A proof is given in Subsection~\ref{AppSWAP:runtime}.
The intuition is simple: Each recovery attempt has a fixed nonzero probability of starting in the fully antisymmetric component, in which case the procedure terminates after at most \(L\) additional SWAP tests. Hence, after sufficiently many attempts, termination occurs with exponentially high probability. Since every individual attempt can have at most $L$ SWAP tests, the runtime has an exponentially decaying tail.

Write $\mathcal{F}_{L_\ast,L,T}$ for the channel implemented
by the protocol with the recovery cutoff.
Combining Eq.~\eqref{Eqn:sigmaLTtracedist} with the exact
agreement on the $\lambda_0$ sector and the factorization
with the reference system gives
\begin{align}
&\frac12
\left\|
\mathcal{F}_{L_\ast,L,T}\circ\mathcal{K}
-
\mathcal{F}_{\mathrm{ideal}}\circ\mathcal{K}
\right\|_\diamond
\le
1-p_{\mathrm{succ}}(L_\ast)(1-p_{\mathrm{cut}})q(L)
\nonumber\\
&\le
[1-p_{\mathrm{succ}}(L_\ast)]
+[1-q(L)]
+p_{\mathrm{cut}}
\nonumber\\
&\le
e^{-L_\ast/k}
+
e^{-L/(k-1)}
+
e^{-c(T-L)}.
\end{align}
In the second inequality we applied a union bound and for the third, we applied the bounds established in Eqs.~\eqref{Eqn:recovery-tail}, \eqref{Eqn:q(L) bound} and \eqref{Eqn:p_succBound}.

For any target error $0<\delta<1$, choose $L_\ast$, $L$,
and $T$ as specified in Theorem~\ref{THM:SWAP}.
Each of the three error terms is then at most $\delta/3$,
and hence
\begin{equation}
\frac12
\left\|
\mathcal{F}_{L_\ast,L,T}\circ\mathcal{K}
-
\mathcal{F}_{\mathrm{ideal}}\circ\mathcal{K}
\right\|_\diamond
\le \delta.
\end{equation}
The total number of SWAP tests satisfies
\begin{align}
N_{\mathrm{SWAP}}
&\le L_\ast+T
\nonumber\\
&\le
\left(
2k-1+\frac{2k(k-1)}{k+1}
\right)\ln\frac{3}{\delta}+3
\nonumber\\
&\le
4k\ln\frac{3}{\delta},
\end{align}
where the last inequality uses $k\ge2$ and
$\ln(3/\delta)>1$.

\subsection{SWAP-test complexity of Algorithm~2}

We now replace the ideal antisymmetric measurements in
Algorithm~2 with their pairwise SWAP-test implementations
and bound the total number of SWAP tests required to
approximate unitary Schur sampling within diamond
distance $\epsilon$.

Let $0<\epsilon<1$ and $\tilde d=\min\{n,d\}$.
The cases $n=1$ or $d=1$ require no measurements, so
assume $n,d\ge2$.
We run Algorithm~2 with error parameter $\epsilon/2$
in all stopping thresholds and impose a cutoff of
\begin{equation}
M=
\left\lceil
6n\tilde d^{\,2}
\ln\!\left(\frac{4en}{\epsilon}\right)
\right\rceil
\end{equation}
antisymmetric-measurement calls.
If an execution would require more than $M$ calls,
it terminates with a prescribed output.

Let $\mathcal A$ denote this cutoff implementation
with ideal antisymmetric measurements.
By Theorem~\ref{Thm:Algo2performance}, the probability
of algorithmic failure or exceeding the cutoff is at
most $\epsilon/2$ in every isotypic sector.
On each sector after the twirl, successful executions
produce the ideal output while preserving the irrep
state and its correlations with any reference system.
Convexity therefore gives
\begin{equation}\label{Eqn:Algo2-ideal-error}
\frac12
\left\|
\mathcal A-\mathcal E_{\mathrm{Schur}}
\right\|_\diamond
\le\frac{\epsilon}{2},
\end{equation}
where $\mathcal E_{\mathrm{Schur}}$ is the ideal
unitary Schur sampling channel defined in
Eq.~\eqref{Schur}.

Let $\widetilde{\mathcal A}$ denote the corresponding
implementation in which each antisymmetric measurement
is replaced by its SWAP-test implementation, with the
same initial twirl, classical control rules, stopping
thresholds, cutoff, and output convention.
Choose the error of each replacement to be
\begin{equation}
\eta=\frac{\epsilon}{2M}.
\end{equation}
By Theorem~\ref{THM:SWAP}, this accuracy is achievable
on inputs satisfying the antisymmetry promise.
The case $k=1$ is implemented exactly by a single
SWAP test.

We compare the two implementations by replacing
measurement calls one at a time, starting from the end.
For $j=0,\ldots,M$, let $\mathcal A^{(j)}$ denote the
implementation whose first $M-j$ calls are ideal and
whose subsequent calls, if any, use the SWAP-test
implementations.
Thus, $\mathcal A^{(0)}=\mathcal A$ and
$\mathcal A^{(M)}=\widetilde{\mathcal A}$.

The implementations $\mathcal A^{(j)}$ and
$\mathcal A^{(j-1)}$ differ only at call $M-j+1$.
Before this call, both use ideal measurements, so the
required antisymmetry promise holds on every history
that reaches it.
Replacing this call therefore introduces diamond
distance at most $\eta$.
All subsequent operations follow the same
outcome-dependent rules in both implementations and
cannot increase this distance.
Histories that terminate before the differing call
contribute no error.
Hence,
\begin{equation}
\frac12
\left\|
\mathcal A^{(j)}-\mathcal A^{(j-1)}
\right\|_\diamond
\le\eta.
\end{equation}
Consequently,
\begin{align}
\frac12
\left\|
\widetilde{\mathcal A}-\mathcal A
\right\|_\diamond
&\le
\sum_{j=1}^{M}
\frac12
\left\|
\mathcal A^{(j)}-\mathcal A^{(j-1)}
\right\|_\diamond
\nonumber\\
&\le M\eta=\frac{\epsilon}{2}.
\end{align}
Combining this with
Eq.~\eqref{Eqn:Algo2-ideal-error} yields
\begin{equation}
\frac12
\left\|
\widetilde{\mathcal A}
-\mathcal E_{\mathrm{Schur}}
\right\|_\diamond
\le M\eta+\frac{\epsilon}{2}
=\epsilon.
\end{equation}

By Theorem~\ref{THM:SWAP}, each call with promised
antisymmetric subsystem size $2\le k\le\tilde d$
requires at most
$4k\ln(3/\eta)\le4\tilde d\ln(3/\eta)$
SWAP tests.
The exact implementation for $k=1$ also satisfies
this bound.
Since every execution makes at most $M$ calls,
\begin{align}
N_{\mathrm{SWAP}}
&\le
4\tilde d M\ln\!\left(\frac{3}{\eta}\right)
\nonumber\\
&=
4\tilde d M\ln\!\left(\frac{6M}{\epsilon}\right)
\nonumber\\
&=
O\!\left(
n\tilde d^{\,3}\ln^2\!\frac{n}{\epsilon}
\right),
\end{align}
where the last step uses the definition of $M$ and
$\tilde d\le n$.

\section{Applications}\label{sec:Applications}

\subsection{Quantum purity amplification}\label{subsec:QPAapplications}

Quantum purity amplification was first studied for qubits by Cirac
et al.~\cite{Cirac}. The goal is to take multiple noisy copies
\(\rho^{\otimes n}\) of an unknown mixed state and produce an output
state with increased fidelity with the principal eigenstate of \(\rho\).
An immediate application of our construction is an implementation of
the optimal single-output QPA protocol of Ref.~\cite{QPA1}.

The protocol of Ref.~\cite{QPA1} can be viewed as consisting of three
steps: Weak Schur sampling, a correction within each Schur-Weyl sector,
and finally retaining a single output qudit. Conditioned on obtaining
the Young diagram \(\lambda\), the correction prepares the
\(\mathbb S_n\) irrep subsystem in the Young-Yamanouchi basis state
associated with the column-ordered standard Young tableau
\(T^\diamond_\lambda\).

For the permutation-invariant input \(\rho^{\otimes n}\), conditioning on the Weak Schur sampling outcome \(\lambda\) gives a state of the form
\begin{align}
\rho_\lambda\otimes
\frac{\mathbb I_{\mathcal P_\lambda}}{m(n,\lambda)},
\end{align}
where \(\rho_\lambda\) is the state on the \(\mathrm{SU}(d)\) irrep subsystem \(\mathcal Q_\lambda^d\), and \(\mathcal P_\lambda\) is the corresponding \(\mathbb S_n\) irrep subsystem. Our USS channel acts on this sector as
\begin{align}
\rho_\lambda\otimes
\frac{\mathbb I_{\mathcal P_\lambda}}{m(n,\lambda)}
\longmapsto
\rho_\lambda\otimes
|\eta_\lambda\rangle\langle\eta_\lambda|,
\end{align}
while also returning the classical label \(\lambda\). As shown in Appendix~\ref{AppAsub:YoungSymmetrizers}, the multiplicity state \(|\eta_\lambda\rangle\) coincides, up to an irrelevant phase, with the Young-Yamanouchi basis vector indexed by the column-ordered tableau \(T^\diamond_\lambda\). Thus, our USS protocol performs the Schur-sampling and correction stages of the single-output QPA protocol simultaneously.

The remaining QPA step is particularly simple. With the physical ordering corresponding to \(T^\diamond_\lambda\), one retains only the qudit carrying the final label \(n\), namely the bottom box of the rightmost column, and discards all remaining qudits. Hence the optimal single-output QPA channel can be implemented by our USS channel followed only by the appropriate trace-out,
\begin{align}
\mathcal E_{\mathrm{QPA}}
=
\operatorname{Tr}_{\mathrm{discard}}
\circ
\mathcal E_{\text{Schur}},
\label{eq:QPAfromUSS}
\end{align}
where \(\operatorname{Tr}_{\mathrm{discard}}\) traces out all quantum outputs except the \(n\)-th qudit, as well as the classical Schur-label register.

Ref.~\cite{QPA1} implements the Schur-sampling and correction stages using generalized quantum phase estimation (GQPE), state preparation in the representation register, and inverse GQPE, with overall gate complexity \(\operatorname{poly}(n,\log d)\). Our protocol realizes these stages using only
\(\widetilde O(\min\{nd^3,n^4\})\) pairwise SWAP tests.
In particular, for fixed local dimension, the number of pairwise SWAP tests is linear in \(n\) up to logarithmic factors. It is worth noting that Ref.~\cite{QPA1} also introduced the SWAP-test construction SWAPNET, but this implementation is restricted to three effective qudits, whereas the construction here applies for arbitrary \(n\) and \(d\).

The same observation also extends naturally to the multi-output QPA protocol of Ref.~\cite{QPA2}. In the principal-eigenvector case, whenever \(\lambda_1-\lambda_2\ge m\), the overhang-removal rule removes the \(m\) rightmost boxes of the first row. For the column-ordered tableau \(T^\diamond_\lambda\), these are precisely the boxes carrying the final \(m\) labels. In this regime, our USS procedure again prepares the required canonical state, and the final \(m\) registers are retained instead of only the final register. Thus, after USS, the multi-output protocol is obtained simply by tracing out the first \(n-m\) registers. For fixed \(m\) and a nonzero principal spectral gap, this is the asymptotically typical regime under Schur sampling.

\subsection{Optimal quantum cloning}

Another application of the random-SWAP viewpoint is the optimal universal quantum cloning channel. The cloning task is to transform $N$ identical copies of an unknown pure state into $M>N$ approximate copies with the highest possible fidelity. For qubits, the optimal $N\rightarrow M$ cloning channel acts on states supported on the totally symmetric subspace of $N$ qubits and outputs states in the totally symmetric subspace of $M$ qubits. It is given by
\begin{multline}
\mathcal E_{\mathrm{clone}}^{N\rightarrow M}(\rho_N)=\\
\frac{N+1}{M+1}
\Pi_{\mathrm{sym}}^{M}
\left(
\rho_N\otimes \mathbb{I}^{\otimes(M-N)}
\right)
\Pi_{\mathrm{sym}}^{M}.
\end{multline}
Kazi and Marvian~\cite{KaziMarvian2026} gives a SWAP-test implementation of this channel by introducing $M-N$ singlet pairs, using one qubit from each pair to enlarge the active system, and then applying random SWAP tests to project the resulting $M$ qubits onto the totally symmetric subspace. If a singlet is detected, the discarded halves are brought back, the singlets are recovered, and the procedure is restarted. This construction provides a complementary perspective to the ``hunt for antisymmetry'' developed here: starting from $N$ fully symmetric qubits together with $M-N$ singlets, the largest fully symmetric subsystem compatible with the conserved total angular momentum has size $M$, and their protocol finds precisely this maximal symmetric subsystem using random SWAP tests. Ref.~\cite{KaziMarvian2026} uses this map as a primitive to study optimal linear-rate conversion between unknown mixed qubit states.

\section{Discussion}
In this paper, we showed that unitary Schur sampling can be implemented by simply performing repeated SWAP tests. Further, we showed a more general strategy of extracting the largest antisymmetric subsystem using $\mathrm{SU}(d)$-invariant operations repeatedly will achieve unitary Schur sampling. This opens up the scope for developing several other implementations. Our algorithm has a gate complexity of $\tilde{O}(nd^3)$ making it competitive with existing strategies that use complex quantum circuits like the Clebsch-Gordon transform or the quantum fourier transform on the symmetric group. More importantly, our approach maybe more amenable to near-term devices, since it only requires the SWAP test. Our construction is also useful when only partial Young-diagram information is required. For example, quantum rank testing depends only on the first-column length \(\lambda'_1\)~\cite{ODonnellspectrum}, so the decision can be made after only the first maximal antisymmetric subsystem is identified, without performing the full unitary Schur sampling protocol. An interesting direction for future work is to determine whether the dependence on the local dimension can be improved.

\section*{Acknowledgments}

We acknowledge helpful discussions with other members of our group, especially Austin Hulse and Sujay Kazi.  We also acknowledge support from  NSF Phy-2046195,  and 
NSF QLCI grant OMA-2120757. 

\section*{AI Statement}
All scientific ideas, results, and mathematical proofs were developed and verified by the human authors. The initial draft was written entirely by the authors. ChatGPT was used for proofreading, editing, and improving the clarity and presentation of the manuscript. The authors take full responsibility for the manuscript's content, including the accuracy and validity of all results and proofs.

\newpage

\bibliographystyle{apsrev4-2}
\bibliography{Bib,localrefsv2}

@article{helstrom1969quantum,
  title={Quantum detection and estimation theory},
  author={Helstrom, Carl W},
  journal={Journal of Statistical Physics},
  volume={1},
  pages={231--252},
  year={1969},
  publisher={Springer}
}

@book{holevo2011probabilistic,
  title={Probabilistic and statistical aspects of quantum theory},
  author={Holevo, Alexander S},
  volume={1},
  year={2011},
  publisher={Springer Science \& Business Media}
}

@article{giovannetti2011advances,
  title={Advances in quantum metrology},
  author={Giovannetti, Vittorio and Lloyd, Seth and Maccone, Lorenzo},
  journal={Nature photonics},
  volume={5},
  number={4},
  pages={222--229},
  year={2011},
  publisher={Nature Publishing Group UK London}
}

@article{entanglement1,
  title = {Universal distortion-free entanglement concentration},
  author = {Matsumoto, Keiji and Hayashi, Masahito},
  journal = {Phys. Rev. A},
  volume = {75},
  issue = {6},
  pages = {062338},
  numpages = {13},
  year = {2007},
  month = {Jun},
  publisher = {American Physical Society},
  doi = {10.1103/PhysRevA.75.062338},
  url = {https://link.aps.org/doi/10.1103/PhysRevA.75.062338}
}

@ARTICLE{entanglement2,
  author={Blume-Kohout, Robin and Croke, Sarah and Gottesman, Daniel},
  journal={IEEE Transactions on Information Theory}, 
  title={Streaming Universal Distortion-Free Entanglement Concentration}, 
  year={2014},
  volume={60},
  number={1},
  pages={334-350},
  doi={10.1109/TIT.2013.2292135}}

@inproceedings{TomographyDonnel,
  title     = {Efficient Quantum Tomography},
  author    = {O'Donnell, Ryan and Wright, John},
  booktitle = {Proceedings of the 48th Annual ACM Symposium on Theory of Computing (STOC)},
  pages     = {899--912},
  year      = {2016},
  publisher = {ACM},
  doi       = {10.1145/2897518.2897544},
  url       = {https://doi.org/10.1145/2897518.2897544}
}

@article{SpectrumWerner,
  title = {Estimating the spectrum of a density operator},
  author = {Keyl, M. and Werner, R. F.},
  journal = {Phys. Rev. A},
  volume = {64},
  issue = {5},
  pages = {052311},
  numpages = {5},
  year = {2001},
  month = {Oct},
  publisher = {American Physical Society},
  doi = {10.1103/PhysRevA.64.052311},
  url = {https://link.aps.org/doi/10.1103/PhysRevA.64.052311}
}

@article{spectrumchristandl,
  title   = {The Spectra of Quantum States and the Kronecker Coefficients of the Symmetric Group},
  author  = {Christandl, Matthias and Mitchison, Graeme},
  journal = {Communications in Mathematical Physics},
  volume  = {261},
  number  = {3},
  pages   = {789--797},
  year    = {2006},
  doi     = {10.1007/s00220-005-1435-1},
  url     = {https://doi.org/10.1007/s00220-005-1435-1}
}

@inproceedings{spectrumodonnell,
  title     = {Quantum Spectrum Testing},
  author    = {O'Donnell, Ryan and Wright, John},
  booktitle = {Proceedings of the 47th Annual ACM Symposium on Theory of Computing (STOC)},
  pages     = {529--538},
  year      = {2015},
  publisher = {ACM},
  doi       = {10.1145/2746539.2746607},
  url       = {https://doi.org/10.1145/2746539.2746607}
}

@misc{harrowthesis,
      title={Applications of coherent classical communication and the Schur transform to quantum information theory}, 
      author={Aram W. Harrow},
      year={2005},
      eprint={quant-ph/0512255},
      archivePrefix={arXiv},
      primaryClass={quant-ph},
      url={https://arxiv.org/abs/quant-ph/0512255}, 
}

@book{fulton96_young_tableaux,
  title = {Young Tableaux},
  author = {Fulton, William},
  year = 1996,
  month = dec,
  publisher = {Cambridge University Press},
  doi = {10.1017/cbo9780511626241},
  isbn = {978-0-511-62624-1}
}

@article{Marvian2020Coherence,
  author  = {Marvian, Iman},
  title   = {Coherence Distillation Machines Are Impossible in Quantum Thermodynamics},
  journal = {Nature Communications},
  volume  = {11},
  pages   = {25},
  year    = {2020},
  doi     = {10.1038/s41467-019-13846-3}
}

@misc{QPA1,
  author        = {Li, Zhaoyi and Fu, Honghao and Isogawa, Takuya and Silva, Caio and Chuang, Isaac L.},
  title         = {Optimal Quantum Purity Amplification},
  year          = {2025},
  eprint        = {2409.18167},
  archivePrefix = {arXiv},
  primaryClass  = {quant-ph}
}

@misc{qubitSWAPUSS,
  author        = {Brahmachari, Shrigyan and Hulse, Austin and Pfister, Henry D. and Marvian, Iman},
  title         = {Optimal Qubit Purification and Unitary {Schur} Sampling via Random {SWAP} Tests},
  year          = {2025},
  eprint        = {2508.05046},
  archivePrefix = {arXiv},
  primaryClass  = {quant-ph}
}

@article{Odonnellspectrum,
  author  = {O'Donnell, Ryan and Wright, John},
  title   = {Quantum Spectrum Testing},
  journal = {Communications in Mathematical Physics},
  volume  = {387},
  number  = {1},
  pages   = {1--75},
  year    = {2021},
  doi     = {10.1007/s00220-021-04180-1}
}

@article{BCH,
  author  = {Bacon, Dave and Chuang, Isaac L. and Harrow, Aram W.},
  title   = {Efficient Quantum Circuits for {Schur} and {Clebsch--Gordan} Transforms},
  journal = {Physical Review Letters},
  volume  = {97},
  number  = {17},
  pages   = {170502},
  year    = {2006},
  doi     = {10.1103/PhysRevLett.97.170502}
}

@misc{HighDimSchur,
  author        = {Burchardt, Adam and Fei, Jiani and Grinko, Dmitry and Larocca, Martin and Ozols, Maris and Timmerman, Sydney and Visnevskyi, Vladyslav},
  title         = {High-Dimensional Quantum {Schur} Transforms},
  year          = {2025},
  eprint        = {2509.22640},
  archivePrefix = {arXiv},
  primaryClass  = {quant-ph}
}

@article{ODonnellprimer,
  author  = {O'Donnell, Ryan and Wright, John},
  title   = {Learning and Testing Quantum States via Probabilistic Combinatorics and Representation Theory},
  journal = {Current Developments in Mathematics},
  volume  = {2021},
  number  = {1},
  pages   = {43--94},
  year    = {2021},
  doi     = {10.4310/CDM.2021.v2021.n1.a2}
}

@misc{QPA2,
  author        = {Li, Zhaoyi and Theil, Elias and Harrow, Aram W. and Chuang, Isaac L.},
  title         = {Quantum Purity Amplification for Arbitrary Eigenstates and Multiple Outputs},
  year          = {2026},
  eprint        = {2605.21570},
  archivePrefix = {arXiv},
  primaryClass  = {quant-ph}
}

@misc{KaziMarvian2026,
  author        = {Kazi, Sujay and Marvian, Iman},
  title         = {Optimal Linear-Rate Conversion of Unknown Mixed Qubit States via {SWAP} Tests},
  year          = {2026},
  eprint        = {2609.17311},
  archivePrefix = {arXiv},
  primaryClass  = {quant-ph}
}

@article{Krovi_2019,
   title={An efficient high dimensional quantum Schur transform},
   volume={3},
   ISSN={2521-327X},
   url={http://dx.doi.org/10.22331/q-2019-02-14-122},
   DOI={10.22331/q-2019-02-14-122},
   journal={Quantum},
   publisher={Verein zur Forderung des Open Access Publizierens in den Quantenwissenschaften},
   author={Krovi, Hari},
   year={2019},
   month=Feb, pages={122} }

@inproceedings{weakchilds,
  title       = {Weak Fourier--Schur Sampling, the Hidden Subgroup Problem, and the Quantum Collision Problem},
  author      = {Childs, Andrew M. and Harrow, Aram W. and Wocjan, Pawel},
  booktitle   = {Proceedings of the 24th International Symposium on Theoretical Aspects of Computer Science (STACS 2007)},
  series      = {Lecture Notes in Computer Science},
  volume      = {4393},
  pages       = {598--609},
  publisher   = {Springer},
  year        = {2007},
  doi         = {10.1007/978-3-540-70918-3_51},
  url         = {https://doi.org/10.1007/978-3-540-70918-3_51}
}

@article{Cirac,
  author  = {Cirac, J. I. and Ekert, A. K. and Macchiavello, C.},
  title   = {Optimal Purification of Single Qubits},
  journal = {Physical Review Letters},
  volume  = {82},
  number  = {21},
  pages   = {4344--4347},
  year    = {1999},
  doi     = {10.1103/PhysRevLett.82.4344}
}

@inproceedings{majority,
  author    = {Buhrman, Harry and Linden, Noah and Man\v{c}inska, Laura and Montanaro, Ashley and Ozols, Maris},
  title     = {Quantum Majority Vote},
  booktitle = {14th Innovations in Theoretical Computer Science Conference (ITCS 2023)},
  series    = {Leibniz International Proceedings in Informatics (LIPIcs)},
  volume    = {251},
  pages     = {29:1--29:23},
  year      = {2023},
  doi       = {10.4230/LIPIcs.ITCS.2023.29}
}

@misc{cervero2024memory,
  author        = {Cervero-Martin, Enrique and Man\v{c}inska, Laura and Theil, Elias},
  title         = {A Memory and Gate Efficient Algorithm for Unitary Mixed {Schur} Sampling},
  year          = {2024},
  eprint        = {2410.15793},
  archivePrefix = {arXiv},
  primaryClass  = {quant-ph}
}

@article{jansontail,
  author  = {Janson, Svante},
  title   = {Tail Bounds for Sums of Geometric and Exponential Variables},
  journal = {Statistics \& Probability Letters},
  volume  = {135},
  pages   = {1--6},
  year    = {2018},
  doi     = {10.1016/j.spl.2017.11.017}
}

@book{SaganSymmetricGroup,
  author    = {Sagan, Bruce E.},
  title     = {The Symmetric Group: Representations, Combinatorial Algorithms, and Symmetric Functions},
  edition   = {2},
  series    = {Graduate Texts in Mathematics},
  volume    = {203},
  publisher = {Springer},
  year      = {2001},
  doi       = {10.1007/978-1-4757-6804-6}
}

@article{TomographyHarrow,
  author  = {Haah, Jeongwan and Harrow, Aram W. and Ji, Zhengfeng and Wu, Xiaodi and Yu, Nengkun},
  title   = {Sample-Optimal Tomography of Quantum States},
  journal = {IEEE Transactions on Information Theory},
  volume  = {63},
  number  = {9},
  pages   = {5628--5641},
  year    = {2017},
  doi     = {10.1109/TIT.2017.2719044}
}

@misc{cervero2023weakschursamplinglog,
  author        = {Cervero, Enrique and Man\v{c}inska, Laura},
  title         = {Weak {Schur} Sampling with Logarithmic Quantum Memory},
  year          = {2023},
  eprint        = {2309.11947},
  archivePrefix = {arXiv},
  primaryClass  = {quant-ph}
}

@article{Childs2025Streaming,
   title={Streaming quantum state purification},
   volume={9},
   ISSN={2521-327X},
   url={http://dx.doi.org/10.22331/q-2025-01-21-1603},
   DOI={10.22331/q-2025-01-21-1603},
   journal={Quantum},
   publisher={Verein zur Forderung des Open Access Publizierens in den Quantenwissenschaften},
   author={Childs, Andrew M. and Fu, Honghao and Leung, Debbie and Li, Zhi and Ozols, Maris and Vyas, Vedang},
   year={2025},
   month=Jan, pages={1603} }

@misc{Grier2025Streaming,
      title={Streaming quantum state purification for general mixed states}, 
      author={Daniel Grier and Debbie Leung and Zhi Li and Hakop Pashayan and Luke Schaeffer},
      year={2025},
      eprint={2503.22644},
      archivePrefix={arXiv},
      primaryClass={quant-ph},
      url={https://arxiv.org/abs/2503.22644}, 
}

\clearpage
\onecolumngrid
\appendix
\renewcommand{\ifmaintext}[1]{}
\begin{center}
    {\large\bfseries APPENDICES}
\end{center}

\vspace{0.5em}

\appendixentry
    {A}
    {Concepts}
    {App:Concepts}
    {We collect representation-theoretic facts and properties of
    antisymmetric measurements used throughout the paper.}

\appendixentry
    {B}
    {Analysis of Algorithm 1}
    {App:Algo 1}
    {We analyze the subset-search algorithm and bound the number of
    antisymmetric measurements required for unitary Schur sampling.}

\appendixentry
    {C}
    {Growing Antisymmetric Subsystems}
    {App:probabgrowing singlets}
    {We derive the probability of extending the current antisymmetric
    subsystem by one qudit.}

\appendixentry
    {D}
    {Algorithm 2: Proof of Convergence}
    {App:convergence2}
    {We prove that the growth procedure successfully traverses the
    Young diagram with high probability.}

\appendixentry
    {E}
    {Number of Antisymmetric Measurements in Algorithm 2}
    {App:Number of rounds}
    {We bound the total number of antisymmetric measurements required
    by Algorithm~2.}

\appendixentry
    {F}
    {Antisymmetric Test Using SWAP Tests}
    {App:SWAP}
    {We show how the antisymmetric measurements used by the algorithm
    can be implemented using pairwise SWAP tests.}

\appendixentry
    {G}
    {Probability Facts}
    {App:probfacts}
    {We collect the stochastic-domination and geometric-concentration
    results used in the convergence analysis.}

\newpage
\section{Concepts} \label{App:Concepts}

\subsection{Antisymmetric measurements} \label{AppAsub: Antisym}

Throughout the paper we repeatedly use several elementary properties of
antisymmetric measurements. For completeness, we collect them here.

\begin{enumerate}
    \item  For a subset \(S=\{i_1,\ldots,i_k\}\) of qudits, the projector onto the
    fully antisymmetric subspace is
    \begin{align}
        \Pi^-_{S}
        =
        \frac{1}{k!}
        \sum_{\pi\in\mathbb S_k}
        \operatorname{sgn}(\pi)\,
        P_\pi,
    \end{align}
    where \(P_\pi\) denotes the permutation operator corresponding to
    \(\pi\).

    \item    The projector \(\Pi^-_S\) commutes with the collective action of
    \(\mathrm{SU}(d)\),
    \begin{align}
        [\Pi^-_S,U^{\otimes k}]=0,
        \qquad
        \forall\,U\in\mathrm{SU}(d).
    \end{align}
    Consequently, by Schur-Weyl duality,
    \begin{align}
        \Pi^-_S
        =
       \bigoplus_{\lambda \vdash n} \mathbb I_{\mathcal Q_\lambda}
        \otimes
        \pi_\lambda(S).
    \end{align}

    \item Every transposition acting on qudits contained in \(S\) acts as
    \(-1\) on the antisymmetric subspace. Equivalently,
    \begin{align}
        P_{(i,j)}\Pi^-_S
        =
        \Pi^-_SP_{(i,j)}
        =
        -\Pi^-_S,
        \qquad
        i,j\in S.
    \end{align}

    \item Let
    \(
    S=\{1,\ldots,k\}
    \)
    and let \(r\notin S\). The projection, \(M_{r}\), to the antisymmetric subspace of qudits \(S \cup \{r\}\), used throughout the growth subroutine satisfies,
    \begin{align}
      &M_r = \Pi^-_{S \cup \{r\}}
        =
        \frac{1}{k+1}
        \Pi^-_{S}
        \left(
        \mathbb I-\sum_{i=1}^{k}P_{(i,r)}
        \right)
        \\
        &=
        \frac{1}{k+1}
        \left(
        \mathbb I-\sum_{i=1}^{k}P_{(i,r)}
        \right)
        \Pi^-_{S}.
    \end{align}

    \begin{proof}
Recall that
\begin{align}
  \Pi^-_{S \cup \{r\}}
=
\frac{1}{(k+1)!}
\sum_{\pi\in \mathbb{S}_{k+1}}
\operatorname{sgn}(\pi)\,P_\pi .
\end{align}
Every permutation \(\pi\in \mathbb S_{k+1}\) can be written uniquely as
\begin{align}
\pi=\tau_i\sigma,
\end{align}
where \(\sigma\in \mathbb S_k\) permutes the first \(k\) qudits and
\(\tau_i=(i,r)\) for \(i=1,\ldots,k\), while
\(\tau_{k+1}=I\). Moreover,
\begin{align}
\operatorname{sgn}(\tau_i)=
\begin{cases}
-1,&i\le k,\\
1,&i=k+1.
\end{cases}
\end{align}
Hence

\begin{align}
\Pi^-_{1,\ldots,k,r}
&=
\frac1{k+1}
\left(
\mathbb I-\sum_{i=1}^{k}P_{(i,r)}
\right)
\frac1{k!}
\sum_{\sigma\in \mathbb S_k}
\operatorname{sgn}(\sigma)P_\sigma\\
&=
\frac1{k+1}
\left(
\mathbb I-\sum_{i=1}^{k}P_{(i,r)}
\right)
\Pi^-_{1,\ldots,k}.
\end{align}

Finally, since $\mathbb I -\sum_{i=1}^{k}P_{(i,r)}$ is invariant over permutations in $\mathbb{S}_k$, it commutes with \(\Pi^-_S\) which a linear combination of permutations in $\mathbb{S}_k$. 
\end{proof}

\end{enumerate}

\subsection{Reversibility of antisymmetric measurements and other \texorpdfstring{$\mathrm{SU}(d)$}{SU(d)}-invariant measurements}\label{AppAsubsec:Reversibility}

We now show that permutation twirling restores the original
permutation-invariant state  after any sequence of antisymmetric measurements (See also \cite{qubitSWAPUSS}). Consider a permutation-invariant state, $\sigma_{\lambda}=\omega_{\lambda} \otimes \mathbb{I}_{m(n,\lambda)}/m(n,\lambda)$,
 supported entirely on the \(\lambda\)-isotypic component, where \(\omega_{\lambda} \)
is a state on the irrep subsystem \(\mathcal{Q}^d_{\lambda} \).
By Schur-Weyl duality and Schur's lemma, any permutation-invariant state on sector \(\lambda\) takes such a form, with the multiplicity subsystem, \(\mathcal{P}_{\lambda}\), in maximally
mixed state. Next, observe that any antisymmetric measurement,
\(\Pi^-_{S} \),
on a subset, \(S\subseteq [n]\), commutes with $U^{\otimes n}, \forall U \in \mathrm{SU}(d)$, and therefore acts trivially on all $\mathrm{SU}(d)$ irrep subsystems \(\mathcal{Q}^d_{\lambda}\) , 
\begin{align}
\Pi^-_{S}
= \bigoplus_{\lambda \vdash n}\mathbb{I}_{d_{\lambda}}
\otimes
\pi_{\lambda}(S),
\end{align}
where \(\pi_{\lambda}(S)\) acts only on the multiplicity subsystem
\(\mathcal{P}_{\lambda}\). More generally, any $\mathrm{SU}(d)$-invariant measurement defined by a set of measurement operators $\{M_r\}_{r}$, $[M_r, U^{\otimes n}] = 0$ for all $U \in \mathrm{SU}(d)$, satisfying the completeness relation $\sum_r M_r^\dag M_r = \mathbb{I}$, takes the 
form, 
\begin{align}
    M_r=\bigoplus_{\lambda \vdash n}(\mathbb{I}_{d_\lambda}\otimes \widetilde{M}_{r,\lambda}).\ 
\end{align}
Therefore, after any sequence of antisymmetric measurements (or for that matter, $\mathrm{SU}(d)$-invariant measurements), the state remains
unchanged on the irrep subsystem \(\mathcal{Q}^d_{\lambda}\). Thus, such a measurement maps \(\sigma_{\lambda}\) into the
 form, $\omega_{\lambda} \otimes \eta_{\lambda},$ for some state \(\eta_{\lambda}\) on the multiplicity subsystem.

By Schur's Lemma, applying the permutation twirling channel once more restores the maximally
mixed state on the multiplicity subsystem:
\begin{align}
\mathcal{T}\!\left(
\omega_{\lambda}\otimes \eta_{\lambda}
\right)
=
\omega_{\lambda}
\otimes
\frac{\mathbb{I}_{m(n,\lambda)}}{m(n,\lambda)}.
\end{align}
Hence permutation twirling restores the original
permutation-invariant state after any sequence of antisymmetric ($\mathrm{SU}(d)$-invariant) measurements.\\
In practice, it is not necessary to explicitly implement the permutation
twirling channel. Since the state is permutation-invariant, the outcome
probabilities of the antisymmetric measurements are unchanged if, instead of
twirling, we simply perform the measurements on uniformly random subsets of
qudits. Consequently, throughout the protocol, we may equivalently view the
randomness as arising either from permutation twirling or from the random
choice of measured qudits.\\

A version of this property holds true for more general PI states not necessarily restricted to any one isotypic component. After performing this measurement on any PI state $\sigma$, we can recover the original state simply by forgetting the outcome, and perform a permutation twirl, i.e.
\be
\mathcal{T}\left(\sum_{r} M_r \sigma M_r^\dag \right)=\sigma\ .
\ee 
To see this, note that according to Schur-Weyl duality,  $\{U^{\otimes n}: U\in\mathrm{SU}(d)\}$ spans the space of PI operators. Then, for $U\in\mathrm{SU}(d)$, 
\be
\Tr(U^{\otimes n}\sum_{r} M_r \sigma M_r^\dag)=\sum_{r} \Tr(U^{\otimes n}M_r^\dag M_r \sigma)=  \Tr(U^{\otimes n} \sigma) \ .
\ee
A permutation twirl maps \(\left(\sum_{r} M_r \sigma M_r^\dag \right)\) into a PI state without affecting these coefficients, thus it recovers \(\sigma\).

\subsection{Removing the maximal antisymmetric subsystem}\label{AppAsub:lemma 1}

\begin{lemma}\label{lem:canonicalprojection}\label{Lemma:repstateaftermeasurements}
Let \(\lambda\) be a Young diagram with column lengths
\(\lambda'_1,\ldots,\lambda'_r\), and let
\begin{align}
\ket{\phi_\lambda}
=
\bigotimes_{i=1}^r
\ket{\Lambda_{S_i}}
\end{align}
be the corresponding representative state, where \(S_i\) denotes the subset
of qudits corresponding to the \(i\)-th column and
\begin{align}
\ket{\Lambda_{S_i}}
:=
\ket{0 \wedge 1 \wedge \cdots \wedge (\lambda'_i-1)} .
\end{align}
Define
\begin{align}
M:=\Pi^-_{S_r}\cdots\Pi^-_{S_2}\Pi^-_{S_1}.
\end{align}
Then
\begin{align}\label{eq:16}
M\,
\mathcal T\!\left(
|\phi_\lambda\rangle\langle\phi_\lambda|
\right)
\,M^\dagger
\propto
|\phi_\lambda\rangle\langle\phi_\lambda|.
\end{align}
 Therefore, the sequence of antisymmetric measurements corresponding to the
columns of \(\lambda\) projects the twirled representative state onto
\(\ket{\phi_\lambda}\) with nonzero probability. Moreover, no antisymmetric subsystem containing more than \(\lambda'_1\) qudits can be prepared from the \(\lambda\)-isotypic component using \(\mathrm{SU}(d)\)-invariant operations.

\end{lemma}
 
\color{black}

\begin{proof}
  First, note that \(M\, \mathcal T\!\left(\ket{\phi_\lambda}\bra{\phi_\lambda} \right) \,M^\dagger\ne 0\), since for the identity permutation in the twirl, \(M \ket{\phi_{\lambda}}=\ket{\phi_{\lambda}}\).
 The state \(\ket{\Lambda_{S_i}}\) is a superposition of computational basis states that contain each basis label \(0, \ldots \lambda'_{i}-1\) with multiplicity 1. By extention, the state \(\ket{\phi_{\lambda}}\) contains label \(i-1\) with multiplicity \(\lambda_{i}\), that is,
  \begin{equation}\label{eq:1}
    \ket{\phi_{\lambda}} \in \mathrm{span}\left\{P_{\pi}\prod_{i=1}^{\lambda_{1}'}\ket{i-1}^{\otimes \lambda_{i}}: \pi \in \mathbb{S}_{n}\right\}=:V_{1}.
  \end{equation}
  \(V_{1}\) is closed under permutations, therefore, the twirled state, \(\mathcal T\!\left(\ket{\phi_\lambda}\bra{\phi_\lambda} \right)\) is also supported on \(V_{1}\).

  Now examine \(\Pi^{-}_{S_{1}}\mathcal T\!\left(\ket{\phi_\lambda}\bra{\phi_\lambda} \right)\Pi^{-}_{S_{1}}\). Each element of the antisymmetric subspace, \(\Pi^{-}_{S_{1}}\), of the first \(\lambda_{1}'\) qudits is a superposition of computational basis states that contain exactly \(\lambda_{1}'\) different basis labels out of \(0,1,\ldots ,d-1\). On the other hand, the subspace \(V_{1}\) is spanned by basis states containing only the labels \(0,1,\ldots ,\lambda_{1}'-1\). On the first \(\lambda_{1}'\) qudits, any completely antisymmetric linear combination of basis states constructed out of these labels is proportional to \(\ket{\Lambda_{S_{1}}}\). Since we established that \(\Pi^{-}_{S_{1}}\mathcal T\!\left(\ket{\phi_\lambda}\bra{\phi_\lambda} \right)\Pi^{-}_{S_{1}}\neq 0\), the only possibility is that,
  \begin{equation}\label{eq:3}
    \Pi^{-}_{S_{1}}\mathcal T\!\left(\ket{\phi_\lambda}\bra{\phi_\lambda} \right)\Pi^{-}_{S_{1}}\propto \ket{\Lambda_{S_{1}}}\bra{\Lambda_{S_{1}}}\otimes \omega_{2},
  \end{equation}
  where,
  \begin{equation}\label{eq:4}
    \mathrm{supp}(\omega_{2})\subseteq \mathrm{span}\left\{P_{\pi}\prod_{i=1}^{\lambda_{2}'}\ket{i-1}^{\otimes \lambda_{i}-1}: \pi \in \mathbb{S}_{n-\lambda_{1}'}\right\}=:V_{2},
  \end{equation}
  as fixing the state on the first \(\lambda_{1}'\) qudits used up one instance of the labels \(0,1, \ldots \lambda_{1}'\). The subspace \(V_{2}\) of qudits \(S/S_{1}\) has similar structure to \(V_{1}\); the previous arguments apply also for the state \(\Pi^{-}_{S_{2}} \omega_{2} \Pi^{-}_{S_{2}}\) and so on. At the \(r\)-th antisymmetric projection we reach,
  \begin{equation}\label{eq:13}
    \Pi^{-}_{S_{r}} \omega_{r} \Pi^{-}_{S_{r}} \propto \ket{\Lambda_{S_{r}}}\bra{\Lambda_{S_{r}}},
  \end{equation}
  which proves Eq.~\eqref{eq:16}.
  
\color{black}
  Finally, every computational basis state in \(V\) contains only the \(\lambda'_1\) distinct basis labels $
0,1,\ldots,\lambda'_1-1.
$ A nonzero fully antisymmetric state on \(k\) qudits requires \(k\) distinct basis labels. Hence, for \(k>\lambda'_1\), \(V\) has zero overlap with any \(k\)-qudit antisymmetric subspace. Since \(V\) is closed under permutations, it is also preserved by \(\mathrm{SU}(d)\)-invariant operations.
 \end{proof}
\begin{lemma}\label{AppAsub:PurePermutationRegister}
  The representative state \(\ket{\phi_{\lambda}}\) is unentangled w.r.t. the tensor decomposition \(\mathcal{Q}^{d}_{\lambda}\otimes \mathcal{P}_{\lambda}\) of the \(\lambda\) isotypic sector,
  \begin{equation}\label{eq:5}
    \ket{\phi_{\lambda}}= \ket{\psi_{\lambda}}\otimes \ket{\eta_{\lambda}}.
  \end{equation}
\end{lemma}
\begin{proof}
  According to Lemma~\ref{lem:canonicalprojection}, 
  \begin{equation}\label{eq:6}
    M\mathcal{T}(\ket{\phi_{\lambda}}\bra{\phi_{\lambda}})M=M(\mathrm{tr}_{\mathcal{P}_{\lambda}}(\ket{\phi_{\lambda}}\bra{\phi_{\lambda}}) \otimes \mathbb{I}_{m(n,\lambda)})M \propto  \ket{\phi_{\lambda}}\bra{\phi_{\lambda}}.
  \end{equation}
  The operator \(M=\Pi^-_{S_r}\cdots\Pi^-_{S_2}\Pi^-_{S_1}\) is an orthogonal projection since the antisymmetric projections, \(\Pi^{-}_{S_{i}}\), act on disjointed sets of qudits. Moreover, since \(M\) is an element of the algebra generated by permutations, according to the Schur-Weyl duality it takes the block-diagonal form,
  \begin{equation}\label{eq:7}
    M = \bigoplus_{\lambda \vdash n} \mathbb{I}_{d_{\lambda}} \otimes M_{\lambda},
  \end{equation}
  where \(M_{\lambda}\) are orthogonal projections. Therefore we have,
  \begin{equation}\label{eq:8}
    \mathrm{tr}_{\mathcal{P}_{\lambda}}(\ket{\phi_{\lambda}}\bra{\phi_{\lambda}}) \otimes M_{\lambda} \propto \ket{\phi_{\lambda}}\bra{\phi_{\lambda}}.
  \end{equation}
  The only way the LHS can be rank-1 is if both \(\mathrm{tr}_{\mathcal{P}_{\lambda}}(\ket{\phi_{\lambda}}\bra{\phi_{\lambda}})\) and \(M_{\lambda}\) are rank-1.
\end{proof}

\removingcolumns*

\begin{proof}
The second part of the statement follows directly from Lemma~\ref{lem:canonicalprojection}.  For the first part,
define \(M_{\overline{1}}= \Pi^{-}_{S_{r}}\Pi^{-}_{S_{r-1}}\ldots \Pi^{-}_{S_{2}}\).  Following from the proof of Lemma~\ref{lem:canonicalprojection},
  \begin{equation}\label{eq:9}
  \begin{aligned}
    \ket{\phi_{\lambda}}\bra{\phi_{\lambda}}&= \ket{\Lambda_{S_{1}}}\bra{\Lambda_{S_{1}}} \otimes \ket{\phi_{\lambda/\lambda_{1}'}}\bra{\phi_{\lambda/\lambda_{1}'}}\\
        &\propto M \mathcal{T}(\ket{\phi_{\lambda}}\bra{\phi_{\lambda}})M^{\dagger}\\
        &= \Pi^{-}_{S_{1}} M_{\overline{1}}\mathcal{T}(\ket{\phi_{\lambda}}\bra{\phi_{\lambda}}) M_{\overline{1}}\Pi^{-}_{S_{1}}\\
        &= \ket{\Lambda_{S_{1}}}\bra{\Lambda_{S_{1}}} \otimes M_{\overline{1}} \omega M_{\overline{1}}^{\dagger}.
  \end{aligned}
  \end{equation}
Since $\ket{\phi_{\lambda/\lambda_{1}'}}$ is supported on the $\lambda/\lambda'_{1}$ isotypic component of qudits \([n]/S_{1}\), so is $M_{\overline{1}}\omega M_{\overline{1}}$.
  Furthermore, following the argument in the proof of Lemma~\ref{AppAsub:PurePermutationRegister},
  \begin{equation}\label{eq:10}
    M_{\overline{1}}= \bigoplus_{\mu \vdash n- \abs{\lambda_{1}'}} \mathbb{I}_{d_{\mu}} \otimes M_{\overline{1}, \mu},
  \end{equation}
  where \(M_{\overline{1}, \mu}\) is an orthogonal projection to a subspace of \(\mathcal{P}_{\mu}\). 
  It is clear from Eq.~\eqref{eq:3}, that \(\omega\) is permutation invariant, therefore $M_{\overline{1}}\omega M_{\overline{1}} = \tilde{\omega} \otimes M_{\overline{1},\lambda/\lambda_1}$ on Schur–Weyl component \(\mathcal{Q}^{d}_{\lambda/\lambda'_{1}}\otimes \mathcal{P}_{\lambda/\lambda'_{1}}\).
  
  Because \(M_{\overline{1}}\) acts trivially on \(\mathcal{Q}^{d}_{\lambda/\lambda_{1}^{'}}\), it suffices to show that the component $\tilde{\omega}$ of 
  \begin{equation}\label{eq:21}
    \mathrm{Tr}_{S_{1}} \left( \Pi_{S_1}^- M_{\overline{1}}\mathcal{T}(\ket{\phi_{\lambda}}\bra{\phi_{\lambda}}) M_{\overline{1}} \Pi_{S_1}^-\right) = M_{\overline{1}} \omega M_{\overline{1}}^{\dagger},
  \end{equation}
  on \(\mathcal{Q}^{d}_{\lambda/\lambda_{1}^{'}}\) is proportional to that of  \(\mathcal{T}\!\left(\ket{\phi_{\lambda \setminus \lambda'_1}} \bra{\phi_{\lambda \setminus \lambda'_1}} \right)\) (and \(\ket{\phi_{\lambda \setminus \lambda'_1}} \bra{\phi_{\lambda \setminus \lambda'_1}}\)). Applying Lemma~\ref{AppAsub:PurePermutationRegister} to Eq.~\eqref{eq:9},
  \begin{equation}\label{eq:11}
    \ket{\phi_{\lambda/\lambda_{1}'}}\bra{\phi_{\lambda/\lambda_{1}'}}=\ket{\psi_{\lambda/\lambda_{1}'}}\bra{\psi_{\lambda/\lambda_{1}'}}\otimes \ket{\eta_{\lambda/\lambda_{1}'}}\bra{\eta_{\lambda/\lambda_{1}'}} \propto   \tilde{\omega} \otimes M_{\overline{1}, \lambda/\lambda_{1}}.
  \end{equation}
  This equality can be fulfilled only if \(\tilde{\omega} \propto \ket{\psi_{\lambda/\lambda_{1}'}}\bra{\psi_{\lambda/\lambda_{1}'}}\) (and \(M_{\overline{1}, \lambda/\lambda_{1}} \propto \ket{\eta_{\lambda/\lambda_{1}'}}\bra{\eta_{\lambda/\lambda_{1}'}}\)) which proves our statement.

\end{proof}

\largestantisymm*
\begin{proof}
    We first prove this for PI states, and then lift the proof for all states restricted to the $\lambda$-isotypic component. Suppose \(\tau_{S_1}\) is nonzero and supported on the fully antisymmetric subspace,
  \begin{equation}\label{eq:17}
    \Pi^{-}_{S_{1}} M \sigma_{\lambda} M^{\dagger}\Pi^{-\dagger}_{S_{1}}\neq 0.
  \end{equation}
 Recall from Sec.~\ref{subsec:reducingtorepstates} that every PI state \(\sigma_\lambda\) supported entirely on the \(\lambda\)-isotypic component can be written as $\sigma_\lambda=\mathcal W\!\left(\mathcal T(\Phi_\lambda)\right)$. Since both \(M\) and \(\Pi^-_{S_1}\) are
\(\mathrm{SU}(d)\)-invariant,
\begin{align}
&\Pi^-_{S_1} M \sigma_\lambda M^\dagger \Pi^-_{S_1}
=
\mathcal W\!\left(
\Pi^-_{S_1} M
\mathcal T(\Phi_\lambda)
M^\dagger \Pi^-_{S_1}
\right).
\label{Eqn:maxantisymW}
\end{align}
For \(|S_1|>\lambda'_1\), the expression inside \(\mathcal W\) vanishes by
Lemma~\ref{Lemma:removingcolumns} contradicting the assumption that \(\tau_{S_1}\) is nonzero and supported on the fully antisymmetric subspace.

For the second statement, suppose that \(|S_1|=\lambda'_1\). Since \(\Pi^-_{S_1}\) is \(\mathrm{SU}(d)\)-invariant, it commutes with \(\mathcal W\). Hence, by Lemma~\ref{Lemma:removingcolumns},
\begin{align}
\Pi^{-}_{S_1}
\, \sigma_\lambda \,
\Pi^{-}_{S_1}
=
\mathcal W\!\left(
\Pi^{-}_{S_1}
\mathcal T(\Phi_\lambda)
\Pi^{-}_{S_1}
\right)
\propto
\mathcal W\!\left(
\ket{\Lambda_{S_1}}\bra{\Lambda_{S_1}}
\otimes
\mathcal T\!\left(
\ket{\phi_{\lambda\setminus\lambda'_1}}
\bra{\phi_{\lambda\setminus\lambda'_1}}
\right)
\right).
\end{align}
The state inside \(\mathcal W\) is supported entirely on the fully antisymmetric subspace of \(S_1\), while the complementary subsystem is supported on the \(\lambda\setminus\lambda'_1\) isotypic component. Since \(\mathcal W\) is a probabilistic mixture of PI unitaries, each of which commutes with permutations acting within \(S_1\) and within \(\overline{S_1}\), these two symmetry sectors are preserved. Hence, after applying \(\mathcal W\), the reduced state on \(\overline{S_1}\) remains supported entirely on the \(\lambda\setminus\lambda'_1\) isotypic component.\\

To extend the claim to an arbitrary state \(\sigma_{\lambda}\) supported on the \(\lambda\)-isotypic component, consider its permutation twirl \(\mathcal T(\sigma_{\lambda})\). Now let \(\sigma_\lambda\) be arbitrary. If
\begin{align}
\Pi^-_{S_1}M\sigma_\lambda M^\dagger\Pi^-_{S_1}\neq0,
\end{align}
then, since the identity permutation is one of the terms in the twirl and every term is positive semidefinite,
\begin{align}
\Pi^-_{S_1}M\mathcal T(\sigma_\lambda)M^\dagger\Pi^-_{S_1}\neq0.
\end{align}
But \(\mathcal T(\sigma_\lambda)\) is PI and remains supported on the \(\lambda\)-isotypic component, contradicting the PI case when \(|S_1|>\lambda'_1\).
By the permutation-invariant case, $\mathrm{Tr}_{S_{1}} (\Pi^-_{S_1}\mathcal T(\sigma_{\lambda})\Pi^-_{S_1})$
has support only on the residual isotypic component \(\lambda\setminus\lambda'_1\) . But
\begin{align}
  \mathrm{Tr}_{S_{1}}\left( \Pi^-_{S_1}\mathcal T(\sigma_{\lambda})\Pi^-_{S_1} \right)
=
\frac1{n!}\sum_{\pi\in \mathbb{S}_n}
  \mathrm{Tr}_{S_{1}}\left( \Pi^-_{S_1}P_\pi\sigma_{\lambda} P_\pi^\dagger\Pi^-_{S_1} \right),
\end{align}
and every summand is positive semidefinite. Since the identity permutation contributes $\Pi^-_{S_1}\sigma_{\lambda}\Pi^-_{S_1}/n!,$ the latter cannot have support in any other residual isotypic component, or else the twirled state would as well. Hence the remaining system is supported entirely on \(\lambda\setminus\lambda'_1\) .
\end{proof}

\subsection{Representative states are eigenvectors of Young symmetrizers}\label{AppAsub:YoungSymmetrizers}

In this section, we show that the representative state
$\ket{\phi_\lambda}$ is an eigenvector with nonzero eigenvalue of the
Young symmetrizer associated with the canonical tableau $T$. As a
consequence, the multiplicity state $\ket{\eta_\lambda}$ appearing in
the Schur-Weyl decomposition of $\ket{\phi_\lambda}$ coincides, up to
a phase, with the distinguished state $\ket{y(T)}$ selected by the
Young symmetrizer.\\

Let $T$ be the canonical tableau obtained by numbering the boxes of
$\lambda$ down each column and then proceeding from left to right. The
corresponding Young symmetrizer is
\begin{align}
Y_{\lambda,T}
:=
\frac{\dim \mathcal P_\lambda}{n!}
\left(
\sum_{c\in C_T}
\operatorname{sgn}(c)\,P_c
\right)
\left(
\sum_{r\in R_T}
P_r
\right),
\label{Eq:YoungSymmetrizer}
\end{align}
where $C_T$ and $R_T$ denote the column and row stabilizers of $T$,
respectively. Here \(C_T\) consists of permutations that act within the columns of \(T\), while \(R_T\) consists of permutations that act within its rows. Thus, the first factor antisymmetrizes within each column and the second symmetrizes within each row; the prefactor is chosen so that \(Y_{\lambda,T}\) is a projector.\\

Let $S_1,\ldots,S_{\lambda_1}$ denote the subsets of qudits associated
with the columns of $T$. Since the columns are disjoint, the column
antisymmetrizer factorizes into the product of the corresponding
antisymmetric projectors. Hence,
\begin{align}
Y_{\lambda,T}
\propto
\left(
\prod_{j=1}^{\lambda_1}\Pi^-_{S_j}
\right)
\left(
\sum_{r\in R_T}P(r)
\right).
\label{Eq:YoungSymmetrizerProjectors}
\end{align}

According to Lemma~\ref{lem:canonicalprojection},
\begin{align}
\left(
\prod_{j=1}^{\lambda_1}\Pi^-_{S_j}
\right)
\mathcal T
\left(
\ket{\phi_\lambda}\bra{\phi_\lambda}
\right)
\left(
\prod_{j=1}^{\lambda_1}\Pi^-_{S_j}
\right)
\propto
\ket{\phi_\lambda}\bra{\phi_\lambda}.
\label{Eq:ProjectedTwirl}
\end{align}
Expanding the permutation twirl, the left-hand side is proportional to,
\begin{align}
\sum_{\pi\in\mathbb S_n}
\left(
\prod_{j=1}^{\lambda_1}\Pi^-_{S_j}
\right)
P_\pi
\ket{\phi_\lambda}\bra{\phi_\lambda}
P_\pi^\dagger
\left(
\prod_{j=1}^{\lambda_1}\Pi^-_{S_j}
\right).
\end{align}
Every term in this sum is positive semidefinite. Since the sum has support only on \(\operatorname{span}\{\ket{\phi_\lambda}\}\) and each summand is positive semidefinite, no summand can have support outside this subspace; otherwise the sum would also have nonzero support there. It follows that
\begin{align}
\left(
\prod_{j=1}^{\lambda_1}\Pi^-_{S_j}
\right)
P_\pi\ket{\phi_\lambda}
\propto
\ket{\phi_\lambda},
\qquad
\forall\,\pi\in\mathbb S_n,
\label{Eq:ProjectedPermutation}
\end{align}
where the proportionality constant is allowed to vanish.

We now consider the action of the Young symmetrizer on
$\ket{\phi_\lambda}$. Using
Eq.~\eqref{Eq:YoungSymmetrizerProjectors},
\begin{align}
Y_{\lambda,T}\ket{\phi_\lambda}
&\propto
\left(
\prod_{j=1}^{\lambda_1}\Pi^-_{S_j}
\right)
\sum_{r\in R_T}
P(r)\ket{\phi_\lambda}
\nonumber\\
&=
\sum_{r\in R_T}
\left(
\prod_{j=1}^{\lambda_1}\Pi^-_{S_j}
\right)
P(r)\ket{\phi_\lambda}
\nonumber\\
&\propto
\ket{\phi_\lambda},
\label{Eq:YoungEigenvectorPossiblyZero}
\end{align}
where the final relation follows from
Eq.~\eqref{Eq:ProjectedPermutation}. Thus,
$\ket{\phi_\lambda}$ is an eigenvector of $Y_{\lambda,T}$, possibly
with zero eigenvalue. It remains to show that
$Y_{\lambda,T}\ket{\phi_\lambda}\neq0$.

Recall from Lemma~\ref{AppAsub:PurePermutationRegister} that the representative state factorizes in
the Schur-Weyl decomposition as
\begin{align}
\ket{\phi_\lambda}
=
\ket{\psi_\lambda}\otimes\ket{\eta_\lambda},
\label{Eq:RepresentativeFactorization}
\end{align}
where
$\ket{\psi_\lambda}\in\mathcal Q_\lambda^d$ and
$\ket{\eta_\lambda}\in\mathcal P_\lambda$.
Permutations act trivially on $\mathcal Q_\lambda^d$ and irreducibly
on $\mathcal P_\lambda$. Therefore,
\begin{align}
P_\pi\ket{\phi_\lambda}
=
\ket{\psi_\lambda}\otimes
p_\lambda(\pi)\ket{\eta_\lambda},
\end{align}
where $p_\lambda$ denotes the irreducible representation of
$\mathbb S_n$ carried by $\mathcal P_\lambda$.

Since $\mathcal P_\lambda$ is irreducible, the orbit of any nonzero
vector spans the entire representation. In particular,
\begin{align}
\operatorname{span}
\left\{
p_\lambda(\pi)\ket{\eta_\lambda}
:
\pi\in\mathbb S_n
\right\}
=
\mathcal P_\lambda.
\label{Eq:SpechtOrbit}
\end{align}
Equation~\eqref{Eq:ProjectedPermutation} therefore implies, by
linearity, that the Young symmetrizer maps
$\ket{\psi_\lambda}\otimes\mathcal P_\lambda$ into
\begin{align}
\ket{\psi_\lambda}
\otimes
\operatorname{span}\{\ket{\eta_\lambda}\}.
\label{Eq:YoungImageEta}
\end{align}

Moreover, since $Y_{\lambda,T}$ is a linear combination of
permutations,
\begin{align}
[Y_{\lambda,T},U^{\otimes n}]=0,
\qquad
\forall\,U\in\mathrm{SU}(d).
\end{align}
By Schur-Weyl duality, restricted to the $\lambda$-isotypic component
it therefore takes the form
\begin{align}
Y_{\lambda,T}
=
\mathbb I_{\mathcal Q_\lambda^d}
\otimes
\widetilde Y_{\lambda,T},
\end{align}
where $\widetilde Y_{\lambda,T}$ acts on $\mathcal P_\lambda$.
Combining this with Eq.~\eqref{Eq:YoungImageEta}, the image of
$\widetilde Y_{\lambda,T}$ is contained in
$\operatorname{span}\{\ket{\eta_\lambda}\}$. Hence,
\begin{align}
\widetilde Y_{\lambda,T}
=
\ket{\eta_\lambda}\bra{a_\lambda},
\end{align}
for some
$\ket{a_\lambda}\in\mathcal P_\lambda$,
not necessarily normalized. Therefore
\begin{align}
Y_{\lambda,T}
=
\mathbb I_{\mathcal Q_\lambda^d}
\otimes
\ket{\eta_\lambda}\bra{a_\lambda}.
\label{Eq:YoungRankOne}
\end{align}

It follows that
\begin{align}
Y_{\lambda,T}\ket{\phi_\lambda}
&=
\left(
\mathbb I_{\mathcal Q_\lambda^d}
\otimes
\ket{\eta_\lambda}\bra{a_\lambda}
\right)
\ket{\psi_\lambda}\otimes \ket{\eta_\lambda}
\nonumber\\
&=
\braket{a_\lambda|\eta_\lambda}
\ket{\phi_\lambda}.
\label{Eq:YoungEigenvalue}
\end{align}
Suppose, for contradiction, that
$Y_{\lambda,T}\ket{\phi_\lambda}=0$. Then
\begin{align}
\braket{a_\lambda|\eta_\lambda}=0.
\end{align}
Equation~\eqref{Eq:YoungRankOne} would then imply
\begin{align}
Y_{\lambda,T}^2
&=
\mathbb I_{\mathcal Q_\lambda^d}
\otimes
\ket{\eta_\lambda}
\braket{a_\lambda|\eta_\lambda}
\bra{a_\lambda}
\nonumber\\
&=0.
\end{align}
However, Young symmetrizers are idempotent \cite{harrowthesis}, $Y_{\lambda,T}^2
=
Y_{\lambda,T},$
and therefore
\begin{align}
Y_{\lambda,T}\ket{\phi_\lambda}
=
\braket{a_\lambda|\eta_\lambda}
\ket{\phi_\lambda}
\neq0.
\end{align}
Thus $\ket{\phi_\lambda}$ is an eigenvector of the Young symmetrizer
associated with the canonical tableau $T$, with nonzero eigenvalue.

Finally, the image of the Young symmetrizer $Y_{\lambda,T}$ defines
the distinguished one-dimensional direction associated with $T$ in
the multiplicity space $\mathcal P_\lambda$. Denoting the normalized
state spanning this direction by $\ket{y(T)}$, we have
\begin{align}
\operatorname{Im}\widetilde Y_{\lambda,T}
=
\operatorname{span}\{\ket{y(T)}\}.
\end{align}
On the other hand, Eq.~\eqref{Eq:YoungRankOne} gives
\begin{align}
\operatorname{Im}\widetilde Y_{\lambda,T}
=
\operatorname{span}\{\ket{\eta_\lambda}\}.
\end{align}
Hence,
\begin{align}
\ket{\eta_\lambda}
\propto
\ket{y(T)}.
\end{align}

\subsection{Commutation of Nested Isotypic Projectors} \label{AppAsubsec:nestedcommutation}

\begin{lemma} \label{Lemma:nestedcommutator}
Consider an \(N\)-qudit system with qudits labelled by \(\{1,\ldots,N\}\), and let \(M\subseteq\{1,\ldots,N\}\) denote a subset of the qudits and $M' = N\setminus M$. Let \(\Pi_\lambda\) denote the projector onto the \(\lambda\)-isotypic component of the \(N\)-qudit system, and let \(\Pi_\mu\) denote the projector onto the \(\mu\)-isotypic component of the qudits in \(M\). Then
\begin{align}
\left[
\Pi_\mu\otimes\mathbb I_{M'},
\Pi_\lambda
\right]
=0.
\end{align}
Thus $\Pi_\mu\otimes\mathbb I_{M'}\Pi_\lambda$ is a projector as well.
\end{lemma}

\begin{proof}
The projector \(\Pi_\lambda\) is permutation invariant and therefore commutes with every permutation \(P_\pi\), \(\pi\in\mathbb S_N\). On the other hand,
\(\Pi_\mu\otimes\mathbb I_{M'}\) is a linear combination of permutation operators \(P_\sigma\) with \(\sigma\in\mathbb S_M\). Since \(\mathbb S_M\) is a subgroup of \(\mathbb S_N\), each such permutation commutes with \(\Pi_\lambda\). By linearity,
\begin{align}
\left[
\Pi_\mu\otimes\mathbb I_{M'},
\Pi_\lambda
\right]
=0.
\end{align}
\end{proof}

\subsection{Pieri's rule}
\label{AppA:pieris}
Throughout this work, we repeatedly separate fully antisymmetric
subsystems from states supported on a fixed isotypic component.
That is, we end up with a state \(\rho \otimes \sigma\) where \(\rho\) is fully supported on the antisymmetric sector \((1^{k})\) of \(k\) qudits while \(\rho \otimes \sigma\) is fully supported on sector \(\lambda\vdash n\) of \(n\) qudits. We are interested in which isotypic sectors of \(n-k\) qudits these constraints restrict \(\sigma\) to. By concentrating only on the components \(\mathcal{P}_{\lambda}\) of the Schur-Weyl decomposition, this question may be rephrased in terms of the representation theory of \(\mathbb{S}_{n}\) as follows: The restriction of an irrep, \(\mathcal{P}^{\lambda}\), from \(\mathbb{S}_{n}\) to the subgroup \(\mathbb{S}_{k} \times \mathbb{S}_{n-k} \subset \mathbb{S}_{n}\) decomposes into irreps as,

\begin{equation}\label{eq:18}
  \text{Res}_{\mathbb{S}_k \times \mathbb{S}_{n-k}}^{\mathbb{S}_n}(\mathcal{P}^\lambda) \cong \bigoplus_{\mu \vdash k, \; \nu \vdash n-k} (\mathcal{P}^\mu \otimes \mathcal{P}^\nu)^{\oplus c_{\mu\nu}^\lambda}.
\end{equation}
We are interested in the partitions \(\nu \vdash n-k\) that are compatible with \(\lambda\) when \(\mathcal{P}^{\mu}\) is fixed as the sign representation, \(\mu = (1^{k})\). In other words, for what \(\nu\) do we have multiplicities \(c_{(1^{k}) \nu}^{\lambda}\neq0\)?
By the Frobenius reciprocity theorem, the character inner products satisfy,
\begin{equation}\label{eq:19}
      c_{\mu\nu}^\lambda = \left\langle \text{Res}_{\mathbb{S}_k \times \mathbb{S}_{n-k}}^{\mathbb{S}_n}(\mathcal{P}^\lambda), \; \mathcal{P}^\mu \otimes \mathcal{P}^\nu \right\rangle_{\mathbb{S}_k \times \mathbb{S}_{n-k}} = \left\langle \mathcal{P}^\lambda, \; \text{Ind}_{\mathbb{S}_k \times \mathbb{S}_{n-k}}^{\mathbb{S}_n}(\mathcal{P}^\mu \otimes \mathcal{P}^\nu) \right\rangle_{\mathbb{S}_n},
  \end{equation}
  where \(\text{Ind}_{S_k \times S_{n-k}}^{S_n}(\mathcal{P}^\mu \otimes \mathcal{P}^\nu)\) is the \(\mathbb{S}_{n}\) representation induced from the \(\mathbb{S}_{k} \times \mathbb{S}_{n-k}\) representation \(\mathcal{P}^\mu \otimes \mathcal{P}^\nu\). According to the RHS, \(c^{\lambda}_{\mu \nu}\) are the coefficients for the branching rules of the symmetric group. These are known as the Littlewood–Richardson coefficients and can be determined by combinatorial rule that simplifies significantly when one fixes \(\mu=(1^{k})\); in this case it is called as Pieri's rule~\cite{fulton96_young_tableaux}: \(c^{\lambda}_{(1^{k}) \nu}=1\) if one can obtain the Young diagram of \(\lambda\) from that of \(\nu\) by attaching \(k\) boxes, at most one to each row, otherwise, \(c^{\lambda}_{(1^{k}) \nu}=0\). For example,
\begin{equation}\label{eq:20}
  \mathrm{Res}^{\mathbb{S}_{9}}_{\mathbb{S}_{2}\times \mathbb{S}_{7}}(\mathcal{P}^{(4,3,2)})=  \mathcal{P}^{(1^{2})}\otimes \left( \mathcal{P}^{(3,2,2)} \oplus  \mathcal{P}^{(3,3,1)} \oplus \mathcal{P}^{(4,2,1)} \right) \oplus \ldots.
\end{equation}

Thus, if \(\rho\) is a 2-qudit state supported on sector \((1^{2})\), while \(\rho \otimes \sigma\) is supported on sector \((4,3,2)\) of 9-qudits, \(\sigma\) may only have support on sectors \((3,2,2), (3,3,1)\) and \((4,2,1)\).
\newpage

\section{Analysis of Algorithm 1} \label{App:Algo 1}
In this Appendix, we analyse Algorithm 1 in subsection \ref{subsec:Algo1} and calculate the number of antisymmetric measurements it would require to perform unitary Schur sampling with probability at least $1-\epsilon,\ \forall \epsilon >0$.
We consider the state $\mathcal{T}\!\left( \ket{\phi_{\lambda}}\bra{\phi_{\lambda}} \right)$. This can be interpreted as a uniformly random permutation $\pi \in \mathbb{S}_n$ being applied to the state $\ket{\phi_{\lambda}}$. We are interested in the probability of finding the largest antisymmetric subsystem, $\tr\left( \Pi^-_{\lambda'_1}\mathcal{T}\!\left( \ket{\phi_{\lambda}}\bra{\phi_{\lambda}} \right)\right)$, when we pick $\lambda'_1$ qudits randomly and perform the antisymmetric measurement. This can be lower bounded by the probability that the random $\lambda'_1$ qudits we pick happen to be $\pi(1),\cdots, \pi(\lambda'_1)$. Note that conditioned on this event, the measurement will project to antisymmetric subspace with probability 1. Thus we have
\begin{align}
    \tr\left( \Pi^-_{\lambda'_1}\mathcal{T}\!\left( \ket{\phi_{\lambda}}\bra{\phi_{\lambda}} \right)\right) \ge \Pr\left(\text{pick } \pi(1),\cdots, \pi(\lambda'_1)\right) \ge \binom{n}{\lambda'_1}^{-1}.
\end{align}
There are other subsets of qudits we may pick that would still project to the antisymmetric subspace, and hence the largest antisymmetric subsystem with some non-zero probability. However, since the purpose of this section is to demonstrate that our first strategy works, we will settle for a simple lower bound.\\

Next, consider a stage of the algorithm at which \(n_{\mathrm{remaining}}\) qudits remain in the state $\mathcal{T}\left(\ket{\phi_\mu}\bra{\phi_\mu}\right),$ where \(\mu\) is the residual Young diagram obtained from \(\lambda\) after removing the columns identified in the preceding stages of the algorithm. Note that \(\mu\) is not known a priori. Suppose that, from the preceding tests, we have learned that the largest antisymmetric subsystem of the remaining qudits has size at most \(k\), equivalently, that the longest column of \(\mu\) has length at most \(k\).
 Then, if we perform the antisymmetric measurement on $k$ randomly chosen qudits, the probability of projecting to the antisymmetric subspace is at least 
 \be\binom{n_{\text{remaining}}}{k}^{-1} \ge \binom{n}{k}^{-1}
 \ee
 if the longest column of \(\mu\) has length \(k\), while this probability is zero if the longest column is shorter than \(k\). Therefore, if the longest column has \(k\) boxes, then after \(R\) such tests,
\begin{align}
&\Pr(\text{no test succeeds})
\le
\left(1-\binom{n}{k}^{-1}\right)^R
\le
\exp\left(-R\binom{n}{k}^{-1}\right).
\end{align}
 Thus, if the longest column has \(k\) boxes, then \(\binom{n}{k}\ln(\epsilon^{-1})\) antisymmetric tests suffice to detect it with probability at least \(1-\epsilon\).\\

With these facts in mind, we proceed to analyze the algorithm .  We can think of finding the largest antisymmetric subsystem as removing the longest column from the remaining young diagram. The strategy is simple: start by setting $k=\tilde{d} =\min{\{n,d\}}$, and every time $\binom{n}{k}\ln{(n\epsilon^{-1})}$ consecutive rounds of antisymmetric measurements fail, reset $k:=k-1$. Thus, the number of tests necessary to remove each column of length $\lambda'_i$ is at most $\binom{n}{\lambda'_i}\ln{(n\epsilon^{-1})}$ with probability at least $1-\epsilon/n$. Then, the number of rounds of antisymmetric measurements that the algorithm performs is at most, 
\begin{align}
    &\sum_{l=1}^{\tilde{d}}
    (\#\!\left\{\text{columns of length }l\right\}+1)
    \binom{n}{l}
    \ln(n\epsilon^{-1}).
\end{align}
The additional \(1\) accounts for the final stage in which the algorithm
certifies that no further columns of length \(l\) remain.\\

{Let \(c_1,\dots,c_m\) denote the columns of the Young diagram in the order in
which the protocol attempts to remove them. Then the event that unitary Schur
sampling succeeds is precisely the event that every column is successfully
identified and removed. Therefore, 
\begin{align}
    \Pr(\text{USS achieved})
    &=
    \Pr(\text{all columns removed})
    \notag\\
    &=
    \prod_{i=1}^{m}
    \Pr\!\left(
    c_i \text{ is removed}
    \;\middle|\;
    c_1,\dots,c_{i-1}\text{ removed}
    \right).
\end{align}
Since each column is removed with probability at least
$1-\epsilon/n$,
we obtain,
\begin{align}
    \Pr(\text{USS achieved})
    &\ge \prod_{i=1}^m \left(1-\frac{\epsilon}{n}\right)\ge
    \left(1-\frac{\epsilon}{n}\right)^{n}.
\end{align}}
Since there can be no more than $n$ columns,
\begin{align}
    &\left(1-\frac{\epsilon}{n}\right)^{n} \ge
    1-\frac{\epsilon}{n}n = 1-\epsilon.
\end{align}

We find an upper bound for the number of rounds of measurements:
\begin{align}
    \text{Rounds of measurement}& \le \sum_{l=1}^{\tilde{d}} 
    \left(1+ \#\!\left\{\text{columns of length }l\right\}\right)
    \binom{n}{l}
    \ln(n\epsilon^{-1})\notag\\
        &\le
    n^{\tilde{d}}
    \ln(n\epsilon^{-1})
    \sum_{l=1}^{\tilde{d}} 
    \left(1+\#\!\left\{\text{columns of length }l\right\}\right)\notag\\
    &=
    n^{\tilde{d}}\ln(n\epsilon^{-1})
    \left(
    \tilde{d}+\#\!\left\{\text{total columns}\right\}
    \right)
    \notag\\
    &\le
    (n+\tilde{d})n^{\tilde{d}}\ln(n\epsilon^{-1})
    \notag\\
    &\le
    2n^{\tilde{d}+1}\ln(n\epsilon^{-1}).
\end{align} The first equality follows from the fact every column in the Young diagram has at most $\tilde{d}$ boxes.
Thus our first strategy would require at most $2n^{\tilde{d}+1}\ln{(n\epsilon^{-1})}$ rounds of antisymmetric measurements to perform unitary Schur sampling with probability at least $1-\epsilon$. \\

\newpage

\section{ Growing Antisymmetric subsystems}\label{App:probabgrowing singlets}
In this appendix we prove Eq.~\eqref{Eqn:probgrowth}, which lower bounds the probability that, for a state supported on a fixed \(\lambda\)-isotypic component and containing a \(k\)-qudit fully antisymmetric subsystem, adjoining a uniformly random qudit from the remainder enlarges this antisymmetric subsystem to \(k+1\) qudits. Throughout, let $\sigma$ be supported on the
$\lambda$-isotypic component,
\[
\Pi_\lambda \sigma \Pi_\lambda=\sigma,
\]
and suppose that the first $k$ qudits lie in the antisymmetric
subspace,
\[
\Pi^-_{1,\dots,k}\sigma\Pi^-_{1,\dots,k}=\sigma.
\]
Since the first \(k\) qudits are supported on the fully antisymmetric
subspace, the state is invariant under conjugation by permutations of
these qudits, i.e.,
\begin{align}
P_\pi\sigma P_\pi^\dagger=\sigma,
\qquad
\forall\,\pi\in S_k.
\end{align}
We consider the following task: 
\begin{itemize}
    \item Pick a qudit $r$ uniformly randomly from $k+1, \cdots, n$. Since the quantity
    \(
    \mathbb{E}_r\operatorname{tr}(M_r\sigma)
    \)
    is unchanged by permutation twirling of the qudits
    \(k+1,\ldots,n\), we may assume, without loss of generality, that
    \(\sigma\) is permutation-invariant on this subsystem. That is,
    \begin{align}
    (\mathbb I\otimes\mathcal T_{>k})(\sigma)=\sigma.
    \end{align}
    \item Perform the measurement $\{\Pi ^{-}_{1,\cdots,k,r}, I - \Pi ^{-}_{1,\cdots,k,r} \}$ . Henceforth we will use $M_r = \Pi ^{-}_{1,\cdots,k,r}$.
\end{itemize}
Our goal is to lower-bound the probability of outcome $M_r$ of this measurement. Explicitly, 

\begin{equation}
  \mathbb{E}_r \tr\big( M_r \sigma M_r^{\dagger}    \big) =  \tr\bigg( (\mathbb{E}_r M_r) \sigma \bigg).    
\end{equation}
Note that $M_r = \frac{1}{k+1}\Pi^-_{1,\cdots,k}\bigg(I - \sum_{i=1}^k P_{(i,r)}\bigg)$ (see Appendix \ref{App:Concepts}). 
Using this, 
\begin{align}
     \mathbb{E}_r\tr\bigg(  M_r \sigma \bigg)
     &=\frac{1}{k+1}\mathbb{E}_r \tr\bigg(  \Pi^-_{1,\cdots,k}\left(\mathbb{I} - \sum_{i=1}^k P_{(i,r)}\right) \sigma \bigg)
       \notag\\
  &=\frac{1}{k+1}\mathbb{E}_r \tr\bigg(  \left(\mathbb{I} - \sum_{i=1}^k P_{(i,r)}\right) \sigma \bigg)\notag\\
      &=\frac{1}{k+1}- \frac{1}{k+1}\mathbb{E}_r  \tr\bigg(  \left(\sum_{i=1}^k P_{(i,r)}\right) \sigma \bigg)\notag\\
  &=\frac{1}{k+1}- \frac{1}{k+1}\mathbb{E}_r  \tr\Bigg(  \bigg(\sum_{i=1}^k P_{(i,k)}P_{(i,r)} P_{(i,k)}\bigg)\sigma  \Bigg).
\end{align}
The last equality follows from the permutation invariance of antisymmetric subsystems i.e.,  
 \begin{equation}
     P_{(i,k)}\sigma P_{(i,k)} = P_{(i,k)} \Pi^-_{1,\cdots,k}\sigma \Pi^-_{1,\cdots,k} P_{(i,k)} = \Pi^-_{1,\cdots,k}\sigma \Pi^-_{1,\cdots,k} =\sigma. 
 \end{equation}
 Since conjugation by a permutation relabels transpositions, \(P_{(i,k)}P_{(i,r)}P_{(i,k)} = P_{(k,r)}\), for every $i\le k$,
    \begin{align}     
       \mathbb{E}_r\tr\bigg(  M_r \sigma \bigg) &=\frac{1}{k+1}- \frac{1}{k+1}\mathbb{E}_r  \tr\Bigg(  \bigg(\sum_{i=1}^kP_{(k,r)}\bigg)\sigma  \Bigg)\notag\\
       &=\frac{1}{k+1}- \frac{k}{k+1} \tr\bigg(  \mathbb{E}_r P_{(k,r)}\sigma  \bigg)\notag\\
       &=\frac{1}{k+1}- \frac{k}{k+1} \tr\bigg(   P_{(k,k+1)}\sigma  \bigg). \label{Eq:probabgrow2}
 \end{align}
  The last equality follows from the permutation-invariance of $k+1, \cdots, n$ qudits, i.e, $\tr\!\left(P_{(k,r)}\sigma\right) = \tr\!\left(P_{(k,k+1)}\sigma\right)$ for all  $r>k$.

 In the following subsection we show that, 
 \begin{align}
    &\tr\!\left(P_{(k,k+1)}\sigma\right) \le  
\frac{1}{k(n-k)}\left( \sum_{i=1}^k\lambda_i 
-k \right). \label{Eq:character}
\end{align}
  We use this to lower-bound the expression in \eqref{Eq:probabgrow2},
\begin{equation}
\mathbb{E}_r\tr\!\left( M_r \sigma \right)
\ge
\frac{1}{k+1}
-
\frac{k}{k+1}
\cdot
\frac{1}{k(n-k)}
\left(
\sum_{i=1}^k\lambda_i-k
\right)
=
\frac{n-\sum_{i=1}^k\lambda_i}{(n-k)(k+1)}
=
\frac{
\sum_{i=k+1}^{d}\lambda_i
}{
(n-k)(k+1)
},
\end{equation}
where in the last step we used that $n = |\lambda|
=
\sum_{i=1}^{d}\lambda_i.$
This proves Eq.~\eqref{Eqn:probgrowth}.

\subsection{Bounding the SWAP expected value}
Now we provide a proof for \eqref{Eq:character}.
Recall that \(\sigma\) has the following symmetries,
\begin{align}
     [\sigma, P_{\pi} \otimes P_{\pi'}]=0 \quad  \forall \pi \in \mathbb{S}[1,\ldots ,k],\  \pi' \in \mathbb{S}[k+1,\ldots, n],\qquad 
     \Pi^-_{1,\dots,k}\,\sigma\,\Pi^-_{1,\dots,k}=\sigma.
 \end{align} 
 Our strategy is to express the expectation value
\(
\tr(P_{(k,k+1)}\sigma)
\)
in terms of averages of transpositions. These averages are determined by
the characters of the transposition conjugacy class in the irreducible
representations appearing in the Schur-Weyl decomposition. We therefore
recall the normalized character formula \cite{ODonnellprimer} for a transposition
\(\pi\in \mathbb{S}_n\)
\begin{align}\label{eq:swapcharacter}
   \chi_{\lambda}(\pi)
=
\frac{1}{n(n-1)}
\left(
\sum_{i=1}^{d}\lambda_i^2
-
\sum_{i=1}^{d}(2i-1)\lambda_i
\right). 
\end{align}

The expectation values of transpositions
\(\tr(P_{(a,b)}\sigma)\) fall into only three categories:
 If \(a,b \le k\), then the corresponding qudits lie entirely in the
    antisymmetric subsystem, and therefore
    \begin{align}
    \tr(P_{(a,b)}\sigma)=-1.
    \end{align}

 If \(a,b > k\), permutation symmetry on the remaining subsystem
    implies that all such transpositions have the same expectation value:
    \begin{align}
    \tr(P_{(a,b)}\sigma)
    =
    \tr(P_{(r,s)}\sigma),
    \qquad
    \forall a,b,r,s > k.
    \end{align}

 Finally, the total sum over all transpositions 
    \begin{align}
    \sum_{a<b}\tr(P_{(a,b)}\sigma)
    \end{align}
    is fixed by the isotypic component in which \(\sigma\) resides. In
    particular, this quantity is determined by, the normalized character of the
    transposition conjugacy class (SWAP permutations) in the irreducible representation labeled by
    \(\lambda\). 
These observations allow us to express cross-subsystem transposition expectations in terms of the character of the transposition class in the irreducible representation labeled by \(\lambda\),
\begin{equation}
  \begin{aligned}
k(n-k)\tr\!\left(P_{(k,k+1)}\sigma\right)
&=
\sum_{s=1}^{k}
\sum_{r=k+1}^{n}
\tr\!\left(P_{(r,s)}\sigma\right)
\\
&=
\sum_{1 \le a < b \le n}
\tr\!\left(P_{(a,b)}\sigma\right)
-
\sum_{1 \le c < d \le k}
\tr\!\left(P_{(c,d)}\sigma\right)
-
\sum_{k < r < s \le n}
\tr\!\left(P_{(r,s)}\sigma\right)
\\
&=
\binom{n}{2}\chi_{\lambda}((1\,2))
-
\binom{k}{2}(-1)
-
\sum_{k < r < s \le n}
    \tr\!\left(P_{(r,s)}\sigma\right).
    \end{aligned}
\end{equation}

We now analyze the remaining term, \(\sum_{k < r < s \le n}
\tr\!\left(P_{(r,s)}\sigma\right)\).
Observe that we are only considering operators supported on the subsystem
consisting of qudits \(k+1,\dots,n\). Since \(\sigma\) is permutation-invariant
under permutations of these qudits, the reduced state on this subsystem admits
a decomposition as a convex combination of isotypic components:
\begin{align}
\sigma_{>k}
=
\sum_{\mu} p_{\mu}\,\sigma_{\mu},
\end{align}
where each \(\sigma_{\mu}\) is supported on the \(\mu\)-isotypic component of
\((\mathbb{C}^d)^{\otimes (n-k)}\), and \(\{p_\mu\}_{\mu}\) forms a probability
distribution.

Consequently, the expectation value of transpositions on this subsystem is
determined entirely by the normalized characters of the transposition conjugacy class in
the irreducible representations labeled by \(\mu\):
\begin{align}
\sum_{k < r < s \le n}
\tr\!\left(P_{(r,s)}\sigma\right)
=
\sum_{\mu}
p_{\mu}\,
\binom{n-k}{2}
\chi_{\mu}((1\,2)).
\end{align}

Moreover, not every Young diagram \(\mu\vdash (n-k)\) can appear in this
decomposition. Since the first \(k\) qudits are constrained to lie in the fully
antisymmetric subspace, the Young diagrams \(\mu\) appearing in the
Schur-Weyl decomposition of the remaining \(n-k\) qudits are restricted
by Pieri's branching rule (see subsection~\ref{subsection:Pieri}). We denote the
set of admissible Young diagrams by \(M(\lambda)\), the set of Young diagrams obtained from \(\lambda\) by sequentially removing \(k\)  corner boxes,
with at most one box removed from each row.
In terms of the row lengths of the Young diagrams, this is equivalent to,
\begin{align}
\lambda_i-\mu_i \in \{0,1\},
\qquad
\sum_i (\lambda_i-\mu_i)=k,
\end{align}
with the additional constraint that \(\mu\) remains a valid Young diagram, i.e.
\begin{align}
\mu_i \ge \mu_{i+1}
\qquad
\forall i.
\end{align}
Hence 
\begin{align}
k(n-k)\tr\!\left(P_{(k,k+1)}\sigma\right)
&=
\binom{n}{2}\chi_{\lambda}((1\,2))
-
\binom{k}{2}(-1)
-
\sum_{k < r < s \le n}
\tr\!\left(P_{(r,s)}\sigma\right)\notag\\
& \le \max_{\mu \in M(\lambda)} \binom{n}{2}\chi_{\lambda}((1\,2))
-
\binom{k}{2}(-1)
-  \binom{n-k}{2}\chi_{\mu}((1\,2)).
\end{align}
 We justify the above inequality: The state $\sigma_{>k}$ is permutation invariant, and hence a convex combination of states living exclusively in some set of isotypic components $\mu \vdash n-k$ in the set of Young Diagrams $M(\lambda)$. Therefore 
\begin{align}
   \max_{\mu \in M(\lambda)}  \chi_{\mu}((1\,2)) \ge \tr\!\left(P_{(a,b)}\sigma\right) \ge \min_{\mu \in M(\lambda)} \chi_{\mu}((1\,2)).\ \forall a,b >k.
\end{align}
\\

Therefore, once the admissible Young diagrams \(\mu\) are determined via the
Pieri's rule in subsection~\ref{subsection:Pieri}, the corresponding transposition expectations
are fixed entirely by the row lengths of \(\mu\). Using the expression for the transposition character in Eq.~\eqref{eq:swapcharacter} we obtain,
\begin{align}
    k(n-k)\tr\!\left(P_{(k,k+1)}\sigma\right)
    &\le \max_{\mu \in M(\lambda)} \binom{n}{2}\chi_{\lambda}((1\,2))
    -
    \binom{k}{2}(-1)
    -  \binom{n-k}{2}\chi_{\mu}((1\,2))\notag\\
    &=\max_{\mu \in M(\lambda)} \frac{1}{2} \left(
\sum_{i=1}^{d}(\lambda_i^2 -\mu_i^2)
-(2i-1)(\lambda_i -\mu_i) 
       \right) +\binom{k}{2}\notag\\
    &=\max_{\mu \in M(\lambda)} \frac{1}{2} \left(
\sum_{i\in B}(2\lambda_i -1)
-(2i-1) 
      \right) +\binom{k}{2}\notag\\
         \intertext{Let $B, \ |B|=k$ be the set of rows that we remove boxes
      from, i.e., \(i \in B\) iff $\lambda_i-\mu_i =1,\ \forall \i \in B$. }
      & \le \max_{B: |B|=k }  \left(
\sum_{i\in B}(\lambda_i -i) 
        \right) +\binom{k}{2}\notag
        \intertext{since $\lambda_i -i$ is decreasing in $i$, this is maximized when B is $\{1,\ldots,k\}$, so }
            &\le   \left(
\sum_{1}^k\lambda_i 
\right) - \binom{k+1}{2}  +\binom{k}{2} \notag\\
& =  \left( 
\sum_{1}^k\lambda_i 
\right) -k.
\end{align}
In summary we have,
\begin{align}
    &k(n-k)\tr\!\left(P_{(k,k+1)}\sigma\right) \le  \sum_{i=1}^k\lambda_i 
 -k. 
\end{align}

\subsection{Young Diagram traversal Interpretation}\label{AppCsubsec:YDtraversal}
Recall that for a state corresponding to a Young diagram \(\lambda\), the
algorithm constructs antisymmetric subsystems using antisymmetric measurements iteratively, growing each
subsystem one qudit at a time. Equivalently, this process may be viewed
directly on the Young diagram: the algorithm traverses the columns of
\(\lambda\) from left to right, and within each column from top to bottom.\\

Suppose the algorithm has already identified and separated \(l-1\)
maximal antisymmetric subsystems and is constructing the largest antisymmetric subsystem for the \(l\)th time, which currently has $k$ qudits. In the Young diagram picture, this
corresponds to having completely traversed the first \(l-1\) columns of
\(\lambda\) and is positioned
at the \(k\)-th row of the \(l\)-st column. Let \(\mu\) denote the Young diagram obtained from \(\lambda\) by removing the first \(l-1\) columns. Our lower bound shows that the
probability of extending the subsystem from size \(k\) to \(k+1\), alternatively, traversing one more step down the current column is at least,
\begin{align}
\frac{
\sum_{i=k+1}^{d}\mu_i
}{
(n_{\mu}-k)(k+1)
}, 
\end{align}
where $\mu_i$'s are the row lengths  and $n_{\mu}$ is the total number of boxes in $\mu$ respectively.
The quantity
\(
\sum_{i=k+1}^{d}\mu_i
\)
is precisely the number of boxes lying strictly below the \(k\)-th row of the
Young diagram, equivalently, the size of the subdiagram below the current row, which we called $T_{\lambda}(k,l)$ in subsection~ \ref{subsec:YDtraversal}. These boxes correspond to the remaining opportunities for extending the current
antisymmetric subsystem.

The denominator has a similarly natural interpretation: the factor \(n_{\mu}-k\)
counts the number of remaining candidate qudits, while the factor \(k+1\)
corresponds to the size of the antisymmetric subsystem being built at that stage. This corresponds to all the boxes to the right of the current box and the boxes in the same column below it in this young diagram picture, which we called $T_{\lambda}(l)$ in the paper. We will use this picture for the rest of the analysis.\\

To express this bound directly in terms of the original Young diagram
\(\lambda\), we recall the quantities introduced in the main text. 
For a box \(b=(k,l)\in \lambda\), we defined two associated quantities:

First, the below-row mass
\begin{align}
T_{\lambda}(l)
:=
\sum_{i=1}^{d}(\lambda_i-l+1)_+ ,
\end{align}
which counts the number of boxes remaining after removing the first \(l-1\) columns of $\lambda$.

Next, the truncated below-row mass
\begin{align}
T_{\lambda}(k,l)
:=
\sum_{i=k+1}^{d}(\lambda_i-l+1)_+ ,
\end{align}
where
\(
(x)_+ := \max\{x,0\}
\).
Equivalently, \(T_{\lambda}(k,l)\) counts the number of boxes lying strictly
below the box \((k,l)\) after removing all columns strictly to the left of the
\(l\)-th column. Then the probability of moving to \((k+1,l)\), conditioned on the event that the algorithm has reached the box
\((k,l)\) without skipping any of the boxes that precede it in the traversal is lower bounded by,
\begin{align}
\Pr\!\left((k+1,l)\mid \text{reached }(k,l)\right)
\ge
\frac{
T_{\lambda}(k,l)
}{
\bigl(T_{\lambda}(l)-k\bigr)(k+1)
}. \label{Eq:probYDstepAPP}
\end{align}

We immediately obtain a crude uniform lower bound for \eqref{Eq:probYDstep}. If the \(l\)-th
column contains at least \(k+1\) boxes, i.e, the antisymmetric subsystem can grow further, $T_{\lambda}(k,l) \ge 1$ and $k+1 \le d$. Further, $T_{\lambda}(l)-k \le n$. Thus 
\begin{align}
    \Pr\!\left((k+1,l)\mid\text{reached }(k,l)\right) \ge \frac{1}{n(k+1)}
 \label{Eq:elementarybound}
\end{align}
If the \(l\)-th column contains just \(k\) boxes, then the move
to \((k+1,l)\) is impossible and the probability is zero.\\
\newpage

\section{Algorithm~2- Proof of Convergence}\label{App:convergence2}
In this appendix, we prove the convergence guarantee stated in Subsection~\ref{subsec:convanalysis}. In particular, we justify the stopping threshold \(R(k)\) and show that Algorithm~2 successfully traverses the entire Young diagram, and hence performs unitary Schur sampling, with probability at least \(1-\epsilon/2\).

\subsection{Upper limit of number rounds per box} \label{subsection:cap}
In the last section we established that 
\begin{align}
\Pr\!\left((k+1,l)\mid(k,l)\right)
\ge
\frac{
T_{\lambda}(k,l)
}{
\bigl(T_{\lambda}(l)-k\bigr)(k+1)
} \ge    \frac{1}{n(k+1)}. 
\end{align}
 We argue that $n(k+1)\ln{\big( 1/\epsilon' \big)}$ rounds of testing suffices move to the next box, if one exists with probability at least $1-\epsilon'$. If there exists another box under $(k,l)$, the probability that the algorithm does not move to it after $t$ rounds is less than 
\begin{align}
    &\left(1-\frac{1}{n(k+1)} \right)^t \le \exp \left( -\frac{t}{n(k+1)} \right) 
\end{align}
We wish that this probability is upper bounded by $\epsilon'$. It suffices that  
\begin{align}
    \exp \left( -\frac{t}{n(k+1)} \right)  = \epsilon'. 
\end{align}
 Rearranging the above expression we find that 
 \begin{align}
     t= n(k+1)\ln{\bigg(\frac{1}{\epsilon'}  \bigg)}
 \end{align}
 is sufficient. We interpret this in the context of Algorithm~2. Suppose the algorithm has successful reached $(k,l)$, i.e, it has traversed all the boxes $(r,s)$ that comes before $(k,l)$ in the ordering picked by our algorithm. Recall that the algorithm performs at most $R(r)$ rounds of measurements on box $(r,s)$. In terms of the random variables defined in the context of Algorithm~2, this is equivalent to 

\begin{align}
      X_{(r,s)} \le R(r)
      \quad
      \forall\, (r,s) < (k,l)
      \text{ such that } r<h(s).
\end{align}

 We have used $  (r,s) < (k,l) $ as a shorthand to indicate that the box $(r,s)$ comes before $(k,l)$ in the ordering we have chosen.

Define the vector of preceding non-bottom waiting times by
\begin{align}
\mathbf{X}_{<(k,l)}
:=
\{X_{(r,s)}:(r,s)<(k,l),\ r<h(s)\},
\end{align}
and let
\begin{align}
E_{(k,l)}
:=
\left\{
\mathbf{x}:
x_{(r,s)}\le R(r)
\text{ for all }(r,s)<(k,l),r<h(s)
\right\}.
\end{align}
Then, for every \(\mathbf{x}\in E_{(k,l)}\),
\begin{align}\label{Eqn:RVprobbound}
\Pr\!\left(
X_{(k,l)}
>
n(k+1)\ln\!\left(\frac{1}{\epsilon'}\right)
\;\middle|\;
\mathbf{X}_{<(k,l)}=\mathbf{x}
\right)
\le \epsilon'.
\end{align}

\subsection{Proof of Convergence}
Since the algorithm performs at most $R(k)= n(k+1)\ln\left(\frac{2n}{\epsilon}\right)$ rounds of measurements while it is at the box $(k,l), \quad \forall (k,l) \in \lambda$, inorder for it to succeed in traversing the entire Young Diagram $\lambda$ , i.e, perform unitary Schur sampling, the following condition must be satisfied: 
\begin{align}
    X_{(k,l)} \le R(k)
\quad
\forall (k,l)\in\lambda \text{ such that } k<h(l).
\end{align}
 We now calculate the probability of success of the algorithm:
\begin{align}
    \Pr\!\left(
    \bigcap_{(k,l)\in\lambda, k<h(l)}
    \left\{
    X_{(k,l)}
    \le
    R(k)
    \right\}
  \right) &\notag\\
  =\prod_{(k,l) \in \lambda, k<h(l)}& \Pr \left( X_{(k,l)} \le  R(k) \mid X_{(r,s)} \le R(r) \quad \forall (r,s) < (k,l)\text{ such that } r<h(s) \right)\notag
\intertext{Using Eq.~\eqref{Eqn:RVprobbound}, we have, }
    &\ge \prod_{(k,l) \in \lambda, k<h(l)} \left( 1 - \frac{\epsilon}{2n} \right) \ge 1 - \sum_{(k,l) \in \lambda, k<h(l)} \frac{\epsilon}{2n}\notag\\
    &\ge 1-\frac{\epsilon}{2}. \label{Eqn:ProbSuccessApp}
\end{align}
Thus Algorithm~2 succeeds with probability at least $1-\frac{\epsilon}{2}.$

\newpage

\section{Number of antisymmetric measurements in Algorithm~2}\label{App:Number of rounds}

In this section, we prove that $6nd^2\ln\left(2en/\epsilon
\right)$ antisymmetric measurements is sufficient to perform unitary Schur sampling with probability at least $1-\epsilon$. We established that this is analogous to traversing the corresponding young diagram left to right, down each column. Let's consider some box $(r,s)$ with row and column $r$ and $s$ respectively. For each non-bottom box \((r,s)\), let \(X_{(r,s)}\) denote the
number of measurement rounds required for the traversal to first move
from \((r,s)\) to \((r+1,s)\). Algorithm~2 performs at most
\begin{align}
    R(r)=n(r+1)\ln\!\left(\frac{2n}{\epsilon}\right)
\end{align}
rounds at such a box, so the actual number of rounds spent there is
\(\min\{X_{(r,s)},R(r)\}\). Note the number of rounds till will reach the last box in the column $s$ is 
\begin{align}
    \sum_{r< h(s)} X_{(r,s)}.
\end{align} 
 Here, $h(s)$ denotes the number of boxes in the $s$th column, and recall that $h(s) \le d$.
On the event that Algorithm~2 successfully traverses the entire Young
diagram, \(X_{(r,s)}\le R(r)\) for every non-bottom box, and hence the
actual number of rounds spent at each such box is simply
\(X_{(r,s)}\). At the bottom box of a column, there are two possibilities.
If its row index equals the known upper bound on the height of the current
column, then the algorithm immediately terminates the column without
performing any further measurements. For \(s>1\), this upper bound is the
height \(h(s-1)\) of the previous column, while for the first column it is
\(d\). Otherwise, the algorithm performs \(R(h(s))\) additional rounds to
certify that the bottom of the column has been reached. Conditioned on
successful traversal, such a certification stage therefore occurs only
when the current column is strictly shorter than its current upper bound.
Since these upper bounds can decrease at most \(d\) times, there are at
most \(d\) such certification stages.\\

We wish to calculate the number of rounds necessary to make sure that the algorithm reaches every box and thus traverses the whole Young Diagram with probability at least $1-\epsilon$, while keeping in mind that our algorithm only spends at most $n(r+1)\ln{\frac{2n}{\epsilon}}$ rounds on each box. This implies that the algorithm must spend $n(h(s)+1)\ln{\frac{2n}{\epsilon}}$ rounds on bottom boxes of columns that are shorter than their previous columns, conditioned that it succeeds in traversing the whole Young Diagram. These measurements can be interpreted as the algorithm certifying that it has reached the end of a column.  Note that there can be at most $d$ such columns. We bound the probability that that algorithm both succeeds and requires at most $T+\sum_{r=1}^d R(r)$, by choosing \(T\) such that

\begin{align}
    &\Pr\!\left(
    \left\{
    \sum_{(r,s)\in\lambda, r< h(s)} X_{(r,s)} \le T
    \right\}
    \cap
    \bigcap_{(r,s)\in\lambda, r< h(s)}
    \left\{
    X_{(r,s)}
    \le
    R(r)
    \right\}
    \right) \ge 1-\epsilon. 
\end{align}
First, we rewrite this probability as
\begin{multline}
    \Pr\!\left(
    \left\{
    \sum_{(r,s)\in\lambda, r< h(s)} X_{(r,s)} \le T
    \right\}
    \cap
    \bigcap_{(r,s)\in\lambda, r< h(s)}
    \left\{
    X_{(r,s)}
    \le
    R(r)
    \right\}
    \right)= \\
    \Pr\!\left(
    \bigcap_{(r,s)\in\lambda, r< h(s)}
    \left\{
    X_{(r,s)}
    \le
    R(r)
    \right\}
    \right)
    -\Pr\!\left(
    \left\{
    \sum_{(r,s)\in\lambda, r< h(s)} X_{(r,s)} > T
    \right\}
    \cap
    \bigcap_{(r,s)\in\lambda, r< h(s)}
    \left\{
    X_{(r,s)}
    \le
    R(r)
    \right\}
    \right).
\end{multline}
    Next we use the bound on the probability of success of Algorithm~2 derived in Eq. \eqref{Eqn:ProbSuccessApp} in Appendix \ref{App:convergence2}:
\begin{align}
    &\ge 1-\frac{\epsilon}{2} -\Pr\!\left(
    \left\{
    \sum_{(r,s)\in\lambda , r < h(s)} X_{(r,s)} > T
    \right\}
    \cap
    \bigcap_{(r,s)\in\lambda, r< h(s)}
    \left\{
    X_{(r,s)}
    \le
    R(r)
    \right\}
    \right) \label{Eqn:Boundstepmid}
\end{align}
We now use the observation about the bottom most boxes in a column that is shorter than the column to its left. Since a certification stage occurs only when the upper bound on the
column height strictly decreases, there is at most one such stage at
each height. Therefore the total number of certification measurements is
at most
\begin{align}
\sum_{r=1}^d R(r)
&=
\sum_{r=1}^d
n(r+1)\ln\!\left(\frac{2n}{\epsilon}\right)
\notag\\
&\le
n(d^2+1)\ln\!\left(\frac{2n}{\epsilon}\right).
\label{Eqn:boundingcert}
\end{align}
\\

We define completely independent geometric random variables
\(Z_{(r,s)}\in\{1,2,\ldots\}\), with \(\Pr(Z_{(r,s)}>t)=(1-p_{(r,s)})^t\) for integer \(t\ge0\), for all $(r,s)\in\lambda$ with $r<h(s)$, independent of the family
$\{X_{(r,s)}\}$, with success probabilities
\begin{align}
    p_{(r,s)}=
\frac{
T_{\lambda}(r,s)
}{
\bigl(T_{\lambda}(s)-r\bigr)(r+1)
}.
\end{align}

We restate the probability bound derived in Eq. \eqref{Eq:probYDstep} from Subsection~\ref{subsec:YDtraversal} in the following way: 

\begin{align}\label{Eqn:RVwaitingtime}
    &\Pr\left( X_{(r,s)}> t\mid X_{(k,l)} =x_{(k,l)} \le R(k)\quad \forall \{ (k,l) < (r,s) \text{ such that }k<h(l) \}  \right) \le (1-p_{(r,s)})^t = \Pr\left( Z_{(r,s)}> t \right).
\end{align}
 Applying Theorem~\ref{Thm:Domination} in Appendix~\ref{App:probfacts}, to Eq.~\eqref{Eqn:Boundstepmid} and Eq.~\eqref{Eqn:RVwaitingtime}, we arrive at  the following upper bound: 
\begin{align} \label{Eqn:domination}
    &\Pr\!\left(
    \left\{
    \sum_{(r,s)\in\lambda , r < h(s)} X_{(r,s)} > T
    \right\}
    \cap
    \bigcap_{(r,s)\in\lambda, r< h(s)}
    \left\{
    X_{(r,s)}
     \le
    R(r)
    \right\}
    \right) \le \Pr\!\left(
    \sum_{(r,s)\in\lambda , r < h(s)} Z_{(r,s)} > T
    \right).
\end{align}

Next, we find $T(\epsilon')$ such that 
\begin{align}
    \Pr\left( \sum_{(r,s)\in \lambda, r <h(s)} Z_{(r,s)} > T(\epsilon') \right) \le \epsilon'. \label{Eqn:geometricRVbound}
\end{align}
Standard concentration inequalities for sums of geometric random variables (see Eq.~\eqref{Eqn:T_nj bound} in Appendix \ref{App:probfacts}) implies that
\begin{align}
    T(\epsilon') = 2\sum_{(r,s)\in \lambda , r < h(s)}\mathbb{E}[Z_{(r,s)}]
+
\left(
\frac{2}{\min_{r,s, r < h(s)} p_{(r,s)}}
\ln\frac{1}{\epsilon'}
\right)
\end{align}
is sufficient to meet the bound in Eq.~\eqref{Eqn:geometricRVbound}.\\

Thus, after $T(\epsilon/2) +n(d^2+1)\ln{(2n/\epsilon)}$ rounds we would have traversed the whole Young Diagram with probability at least $1-\epsilon.$
We now upper-bound this quantity in terms of $n$ and $d$:
\begin{equation}
 T\!\left(\frac{\epsilon}{2}\right)
+ n(d^2+1)\ln{\frac{2n}{\epsilon}}=
 2\sum_{(r,s)\in \lambda , r < h(s)}\mathbb{E}[Z_{(r,s)}]
+
\left(
\frac{2}{\min_{r,s, r < h(s)} p_{(r,s)}}
\ln\frac{2}{\epsilon}
\right)
+
n(d^2+1)\ln{\frac{2n}{\epsilon}}.\label{Eqn:roundsbound1}
\end{equation}

In the next subsection we show that
\begin{align}
\sum_{(r,s)\in \lambda , r < h(s)} \mathbb{E}[Z_{(r,s)}]  \le nd^2\ln{(en)}.\label{Eqn:expectationbound}
\end{align}
Note that the left-hand side gives an upper bound on the expected number of measurements required for the growth steps, excluding the deterministic measurements used to certify the bottoms of columns. Finally, from Eq.~ \eqref{Eq:elementarybound} we have
\begin{align}\label{Eqn:certprobbound}
    \min_{r,s, r < h(s)} p_{(r,s)} & \ge \frac{1}{nd}.\notag
\end{align}
Combining this with Eq.~\eqref{Eqn:roundsbound1}, \eqref{Eqn:expectationbound} we have:
\begin{align}
   & \le 2nd^2\ln{(en)}
+
2nd
\ln\left(\frac{2}{\epsilon}
\right)
+
n(d^2+1)\ln\bigg( \frac{2n}{\epsilon}\bigg)\notag\\
& \le 6nd^2
\ln\bigg( \frac{2en}{\epsilon}\bigg).
\end{align}

\subsection{Expected Number of Measurements}

In this section, we calculate a bound the expected number of rounds of antisymmetric measurements required to traverse through any Young Diagram $\lambda$ (besides the bottom boxes which have a deterministic number of steps) corresponding to $n$ qudit systems of dimension $d$. From \eqref{Eq:probYDstep}, For each non-bottom box \((k,l)\), we have
\begin{align}
 \mathbb{E}[Z_{(k,l)}] = \frac{
\bigl(T_{\lambda}(l)-k\bigr)(k+1)
}{
T_{\lambda}(k,l)
}.
\end{align}

To bound the above expression, we use a simple geometric observation about the Young diagram. Recall that
\(T_{\lambda}(k,l)\) counts the number of boxes lying strictly below the box
\((k,l)\) after the first \(l-1\) columns have been removed. Since we only
consider boxes that have at least one box
directly beneath it,
\begin{align}
T_{\lambda}(k,l)
\ge
\lambda_{k+1}-l+1.
\end{align}
Indeed, the quantity \(\lambda_{k+1}-l+1\) counts the number of boxes in the
\(k+1\)-th row from column \(l\). Thus the
subdiagram counted by \(T_{\lambda}(k,l)\) contains at least
\(\lambda_{k+1}-l+1\) boxes. Then ,
\begin{align}
    \sum_{(k,l)\in \lambda, k < h(l)} \mathbb{E}[Z_{(k,l)}] &\le \sum_{(k,l)\in \lambda, k < h(l)} \frac{
\bigl(T_{\lambda}(l)-k\bigr)(k+1)
}{
T_{\lambda}(k,l)
}\notag\\
& \le \sum_{(k,l)\in \lambda, k < h(l)} \frac{nd}{\lambda_{k+1} - l+1}
    \le \sum_{k=1}^{h(1)-1}
\sum_{i=1}^{\lambda_{k+1}} \frac{nd}{i}.\notag\\
\end{align}
Now using the standard harmonic-series bound,
\begin{align}
   &\sum_{i=1}^{m}\frac{1}{i}\le 1+\ln m, \notag\\
    &\le
    nd\sum_{k=1}^{h(1)-1}
    \left(1+\ln\lambda_{k+1}\right)=
    nd\sum_{k=2}^{h(1)}
    \left(1+\ln\lambda_k\right)\notag\\
    &\le
    nd\,h(1)\ln(en)\le
    nd^2\ln(en).
\end{align}

\newpage

\section{Antisymmetric test using SWAP tests }\label{App:SWAP}

\subsection{Calculating \texorpdfstring{$q(L)$}{q(L)}}\label{AppSWAP:q(l)}
 We now derive the expression for \(q(L)\) used in Sec.~\ref{section:SWAPtoAntisymm}. Recall from Sec.~\ref{section:SWAPtoAntisymm} that \(\mathcal E\) denotes the postselected random-SWAP map, \(a=|\langle\alpha|\beta\rangle|^2=(k+1)/(2k)\), and \(q(L)\) is the weight of the fully antisymmetric component after \(L\) consecutive antisymmetric outcomes. Let \(\rho\) denote the reduced state of the \(k\) qudits at the beginning of a recovery attempt, after the permutation twirl was applied. By Lemma~\ref{lem:recovery-probability},
\begin{align}
\rho
=
a\,\rho_0
+
(1-a)\rho_{1},
\qquad
a:=\frac{k+1}{2k},
\end{align}
where \(\rho_{0}\) is supported on the fully antisymmetric component
\((1^k)\), while \(\rho_{1}\) is supported on
\((2,1^{k-2})\).

The$\mathcal E (\cdot)$ map can be interpreted as one round of a uniformly random SWAP test post-selected on obtaining the antisymmetric outcome. The fully antisymmetric component is unchanged by \(\mathcal E\), whereas by Lemma~\ref{lemma:21k-1random SWAP}, on the \((2,1^{k-2})\) component the antisymmetric outcome occurs with probability $r:=(k-2)/(k-1)$. Since \(\rho_1\) is permutation invariant on the irreducible \((2,1^{k-2})\) sector and \(\mathcal E\) is permutation covariant, \(\mathcal E(\rho_1)\) is proportional to \(\rho_1\). The proportionality constant is the antisymmetric-outcome probability \(r\). Consequently,
\begin{align}
\mathcal E^L(\rho)
=
a\,\rho_0
+
(1-a)r^L\rho_{1}.
\end{align}
Hence the normalized state conditioned on these \(L\) outcomes is
\begin{align}
\sigma_L
=
\frac{
a\,\rho_0
+
(1-a)r^L\rho_{1}
}{
a+(1-a)r^L
}.
\end{align}
Since \(q(L)\) is the weight of the fully antisymmetric component in
\(\sigma_L\), we have
\begin{align}
q(L)
&=
\frac{a}{a+(1-a)r^L},
\end{align}
and this gives us Eq.~\eqref{Eqn:q(L)}. In particular,
\begin{equation} \label{Eqn:q(L) boundderiv}
1-q(L)
=
\frac{(1-a)r^L}{a+(1-a)r^L}
\le
\frac{k-1}{k+1}\left(\frac{k-2}{k-1}\right)^L
\le e^{-L/(k-1)}.
\end{equation}

\subsection{Run-time analysis}\label{AppSWAP:runtime}

We now prove the recovery-tail bound stated in Eq.~\eqref{Eqn:recovery-tail} of Sec.~\ref{section:SWAPtoAntisymm}. Recall that \(N\) denotes the number of recovery SWAP tests required for natural termination, \(p_{\mathrm{cut}}=\Pr(N>T)\), and
\[
a=\frac{k+1}{2k},\qquad
r=\frac{k-2}{k-1}.
\]
It remains to bound the probability that the recovery procedure reaches the total SWAP-test budget \(T\) before obtaining \(L\) consecutive antisymmetric outcomes. Let \(T_{\mathrm{nat}}\) denote the number of recovery SWAP tests that would be performed in the absence of this cutoff. Consider the event
\begin{align}
E_t:=\{T_{\mathrm{nat}}>L+t\},
\end{align}
and the probability of this event is $\delta$ discussed earlier. On \(E_t\), every recovery attempt begun during the first \(t\) SWAP tests must lie in the \((2,1^{k-2})\) component: if an attempt were in the fully antisymmetric component, it would necessarily terminate after at most the next \(L\) tests.By the reset property established in Sec.~\ref{section:SWAPtoAntisymm}, each new recovery attempt starts from the same state \(|\beta\rangle\), and hence has the same sector weights \(a\) and \(1-a\) on \((1^k)\) and \((2,1^{k-2})\), respectively.

Now consider a length-\(t\) sequence of SWAP-test outcomes containing \(j<t\) symmetric outcomes; such a sequence contains \(j+1\) recovery attempts. Conditioned on \(E_{t}\), the state of the first \(k\) qudits at the beginning of all of these recovery attempts must lie in the \((2,1^{k-2})\) component. Ignoring the additional restriction that no attempt may already have accumulated \(L\) consecutive antisymmetric outcomes, and summing over all such sequences, we obtain
\begin{align}
\Pr(E_t)
&\le
\sum_{j=0}^{t}
\binom{t}{j}
(1-a)^{j+1}
(1-r)^jr^{t-j}
\notag\\
&=
(1-a)(1-a(1-r))^t
\notag\\
&\le
e^{-a(1-r)t}.
\label{eq:recovery-tail}
\end{align}
Therefore, choosing
\begin{align}
t
=
\frac{1}{a(1-r)}\ln\!\left(\frac{1}{\delta}\right)
=
\frac{2k(k-1)}{k+1}
\ln\!\left(\frac{1}{\delta}\right)
\end{align}
ensures
\begin{align}
\Pr\!\left(
N
>
L+
\frac{2k(k-1)}{k+1}
\ln\!\left(\frac{1}{\delta}\right)
\right)
\le
\delta.
\end{align}
Thus, taking the recovery cutoff to be
\begin{align}\label{Eqn:recovery cutuff}
T
=
L+
\frac{2k(k-1)}{k+1}
\ln\!\left(\frac{1}{\delta}\right)
\end{align}
ensures that the protocol reaches the cutoff before terminating naturally with probability at most \(\delta\). This equivalent to Eq.~\eqref{Eqn:recovery-tail},
\begin{align}
p_{\mathrm{cut}}
=
\Pr(N>T)
\le
e^{-a(1-r)(T-L)}
=
e^{-c(T-L)},
\end{align}
where
\begin{align}
c
=
a(1-r)
=
\frac{k+1}{2k(k-1)}.
\end{align}
\subsection{Auxiliary results}
\begin{lemma}[Uniqueness of the antisymmetric multiplicity state]
\label{lem:hook-antisym-unique}
Let \(\lambda=(2,1^{k-1})\), and let
\(S\subset\{1,\ldots,k+1\}\) be any subset of \(k\) qudits.
Within the \(\lambda\)-isotypic component, imposing that the qudits in
\(S\) are fully antisymmetric uniquely determines the state in the
\(\mathbb S_{k+1}\) irrep subsystem.

Equivalently, any state in the \(\lambda\)-isotypic component that is
fully antisymmetric on \(S\) must be of the form
\begin{align}
    \omega_\lambda
    \otimes
    \ket{\eta_{\lambda,S}}
    \bra{\eta_{\lambda,S}},
\end{align}
for some state \(\omega_\lambda\) on \(\mathcal Q_\lambda^d\), where
\(\ket{\eta_{\lambda,S}}\in\mathcal P_\lambda\) is uniquely determined
by \(S\) up to a phase; for the canonical subset
\(S=\{1,\ldots,k\}\),
\(\ket{\eta_{\lambda,S}}=\ket{\eta_\lambda}\) up to a phase, where
\(\ket{\eta_\lambda}\) is the multiplicity state of the representative
state,
\begin{align}
    \ket{\phi_\lambda}
    =
    \ket{\psi_\lambda}
    \otimes
    \ket{\eta_\lambda},
\end{align}
as defined in Eq.~\eqref{Eqn:repstatefactorizes}.
\end{lemma}
\begin{proof}
Let \(S=\{1,\ldots,k\}\) denote the canonical subset and let
\(\lambda=(2,1^{k-1})\). Since the antisymmetric projector
\(\Pi^-_S\) is \(\mathrm{SU}(d)\)-invariant, its restriction to the
\(\lambda\)-isotypic component has the form
\begin{align}
    \Pi^-_S
    =
    \mathbb I_{\mathcal Q_\lambda^d}
    \otimes E_S,
\end{align}
where \(E_S\) is a projector acting on \(\mathcal P_\lambda\). 

Suppose that antisymmetry on the first \(k\) qudits does not uniquely
determine the state in the \(\mathbb S_{k+1}\) irrep subsystem. Then
\begin{align}
    \operatorname{rank}(E_S)>1.
\end{align}
Using Eq.~\eqref{Eq:twirledRepstate}, the factorization of the twirled representative state,
\begin{align}
    \mathcal T\!\left(
    \ket{\phi_\lambda}\bra{\phi_\lambda}
    \right)
    =
    \ket{\psi_\lambda}\bra{\psi_\lambda}
    \otimes
    \frac{\mathbb I_{\mathcal P_\lambda}}
    {\dim\mathcal P_\lambda},
\end{align}
we obtain
\begin{align}
    \Pi^-_S\,
    \mathcal T\!\left(
    \ket{\phi_\lambda}\bra{\phi_\lambda}
    \right)
    \Pi^-_S
    =
    \ket{\psi_\lambda}\bra{\psi_\lambda}
    \otimes
    \frac{E_S}{\dim\mathcal P_\lambda}.
\end{align}
Hence, if \(\operatorname{rank}(E_S)>1\), the post-measurement state
has rank greater than one.

On the other hand,  by
Lemma~\ref{Lemma:removingcolumns},
\begin{align}
    \Pi^-_S\,
    \mathcal T\!\left(
    \ket{\phi_\lambda}\bra{\phi_\lambda}
    \right)
    \Pi^-_S
    \propto
    \ket{\Lambda_S}\bra{\Lambda_S}
    \otimes
    \ket{\phi_{(1)}}\bra{\phi_{(1)}}.
\end{align}
The state on the right-hand side is pure, which contradicts
\(\operatorname{rank}(E_S)>1\). Hence
\begin{align}
    \operatorname{rank}(E_S)=1.
\end{align}

Since the representative state \(\ket{\phi_\lambda}\) is itself fully
antisymmetric on the first \(k\) qudits, the unique state in the image of
\(E_S\) must be the multiplicity state \(\ket{\eta_\lambda}\). Therefore,
any state in the \(\lambda\)-isotypic component that is fully antisymmetric
on the first \(k\) qudits has the form,
\begin{align}
    \omega_\lambda
    \otimes
    \ket{\eta_\lambda}\bra{\eta_\lambda},
\end{align}
for some state \(\omega_\lambda\) on \(\mathcal Q_\lambda^d\).

Finally, let \(S\subseteq\{1,\ldots,k+1\}\) be any subset of \(k\) qudits,
and let \(\pi\in\mathbb S_{k+1}\) map \(S\) to \(\{1,\ldots,k\}\).
Since permutations act on the \(\lambda\)-isotypic component as,
\begin{align}
    P_\pi
    =
    \mathbb I_{\mathcal Q_\lambda^d}
    \otimes
    p_\lambda(\pi),
\end{align}
the same argument implies that the multiplicity state is again uniquely
determined. Any state that is fully antisymmetric on \(S\)
has the form,
\begin{align}
   \omega_{\lambda} \otimes p_{\lambda}(\pi)\ket{\eta_{\lambda}} \bra{\eta_{\lambda}}p_{\lambda}^{\dagger}(\pi):=  \omega_\lambda
    \otimes
    \ket{\eta_{\lambda,S}} \bra{\eta_{\lambda,S}}
\end{align}
\end{proof}

\begin{lemma} \label{lemma:21k-1random SWAP}
Let
\(
|\Lambda\rangle
=
|0\wedge1\wedge\cdots\wedge(k-1)\rangle,
\)
and
\(
\sigma
=
\mathcal{T}
\!\left(
|0\rangle\!\langle0|
\otimes
|\Lambda\rangle\!\langle\Lambda|
\right),
\) where \(\mathcal{T}\) denotes the permutation twirl.
Then the SWAP test on registers \(0\) and \(1\) accepts with probability,
$$\Tr\!\left[
\frac{\mathbb{I}+P_{(0,1)}}{2}\sigma
\right]
=
\frac1k.
$$

Since \(P_{(0,1)}\) is unitary invariant, the same probability holds for any permutation invariant state, \(\omega \otimes \mathbb{I}_{\mathcal{P}_{(2,1^{k-1})}}\), supported on the \((2, 1^{k-1})\) isotypic sector. Equivalently, a SWAP test applied to a uniformly random pair of qudits of an arbitrary state in the \((2, 1^{k-1})\) sector returns the symmetric outcome with probability \(1/k\).

\end{lemma}

\begin{proof}

 Recall that the permutation twirl \(\mathcal{T}\) may be viewed as applying a uniformly random permutation to the \(k+1\) qudits. The state
\(
|0\rangle\otimes|\Lambda\rangle
\)
contains a unique copy of \(|0\rangle\) outside the antisymmetric block. After twirling, this distinguished qudit is equally likely to occupy any of the \(k+1\) registers.

The SWAP test acts only on registers \(0\) and \(1\). Thus the only permutations that can contribute are those for which the distinguished \(|0\rangle\) is mapped to register \(0\) or register \(1\). This occurs with probability
\(2/(k+1). \)

If the distinguished \(|0\rangle\) is mapped to any other register, then registers \(0\) and \(1\) both belong to the antisymmetric block \(|\Lambda\rangle\). In that case the reduced state on registers \(0\) and \(1\) is supported entirely on the antisymmetric subspace, therefore \((\mathbb{I}+P_{(0,1)})/2 \)
annihilates it and those permutations contribute zero to the expectation value. Thus,
\begin{align}
\Tr\!\left[
\frac{\mathbb{I}+P_{(0,1)}}{2}\,\sigma
\right]
&=
\frac{2}{k+1}
\Tr\!\left[
\frac{\mathbb{I}+P_{(0,1)}}{2}
\Bigl(
|0\rangle\!\langle0|
\otimes
|\Lambda\rangle\!\langle\Lambda|
\Bigr)
\right].
\end{align}

Since \(P_{(0,1)}\) acts only on subsystems \(0\) and \(1\), we may trace out the remaining \(k-1\) subsystems. The one-body reduced density matrix of the antisymmetric state \(|\Lambda\rangle\) is,

\[
\Tr_{2,\ldots,k}
\bigl(
|\Lambda\rangle\!\langle\Lambda|
\bigr)
=
\frac{\Pi_{k}}{k},
\]

where \(\Pi_{k}\) is an orthogonal projection to the span of
\(\{|0\rangle,\ldots,|k-1\rangle\}\).
Hence

\begin{align}
\Tr\!\left[
  \frac{\mathbb{I}+P_{(0,1)}}{2}\,\sigma
\right]
&=
\frac{2}{k+1}
\Tr\!\left[
  \frac{\mathbb{I}+P_{(0,1)}}{2}
\Bigl(
|0\rangle\!\langle0|
\otimes
\frac{\Pi_{k}}{k}
\Bigr)
\right]
\\
&=
\frac{2}{k+1}
\left(
\frac12+\frac{1}{2k}
\right)
\\
&=
\frac{1}{k}.
\end{align}

Thus the probability that the SWAP test accepts on the state \(\sigma\) is \(\frac1{k}\).
\end{proof}

\begin{lemma}
\label{Lemma:SWAPSuccess21k-1}
Let
\begin{align}
\mathcal V_{(2,1^{k-1})}
:=
\operatorname{span}
\left\{
P_{\pi}
\left(
\ket{0}\otimes\ket{\Lambda}
\right)
:
\pi\in\mathbb S_{k+1}
\right\},
\end{align}
which carries the \(\mathbb S_{k+1}\) irrep corresponding to the partition
\((2,1^{k-1})\). Within this irrep space, the subspace in which the first
two qudits are symmetric is one-dimensional. In particular, if
\(\ket{\omega}\in\mathcal V_{(2,1^{k-1})}\) satisfies
\begin{align}
\frac{\mathbb I+P_{(0,1)}}{2}\ket{\omega}
=
\ket{\omega},
\end{align}
then
\begin{align}
\ket{\omega}
\propto
\frac{\mathbb I+P_{(0,1)}}{2}
\ket{0}\otimes\ket{\Lambda}.
\end{align}

Note that \(\ket{0} \otimes \ket{\Lambda}\) is unentangled w.r.t. the component \(\mathcal{Q}_{(2, 1^{k-1})}^{d}\otimes \mathcal{P}_{(2, 1^{k-1})}\) of the Schur-Weyl decomposition, therefore, the unitary-invariant operator, \((\mathbb I+P_{(0,1)})/2\), acts as a projection to a 1d subspace on the \(\mathbb{S}_{k+1}\) irrep subsystem, \(\mathcal{P}_{(2, 1^{k-1})}\).
\end{lemma}

\begin{proof}
Let \(v_0=|0\rangle\otimes|\Lambda\rangle.\)
Consider the states $\{P_\pi v_0:\pi\in S_{k+1}\}$ The only degree of freedom is the location of the distinguished copy of \(|0\rangle\), so we may label them $\{v_0,v_1,\ldots,v_k\},$ where \(v_j\) denotes the state obtained by moving the $0$th register to register \(j\).
We now determine the action of \(P_{(0,1)}\) on this spanning set. Since
\(P_{(0,1)}\) exchanges registers \(0\) and \(1\),
\[
P_{(0,1)}v_0=v_1,
\qquad
P_{(0,1)}v_1=v_0.
\]
For \(j\ge 2\), both registers \(0\) and \(1\) belong to the antisymmetric block, and therefore $P_{(0,1)}v_j=-v_j.$ Thus \(+1\) eigenspace of \(P_{(0,1)}\) is therefore one-dimensional and is spanned by $v_0+v_1.$ Since
\[
v_0+v_1
=
(\mathbb{I}+P_{(0,1)})v_0,
\]
it follows that every vector in \(V\) satisfying $\frac{\mathbb{I}+P_{(0,1)}}{2}|\omega\rangle=|\omega\rangle$ is proportional to $(\mathbb{I}+P_{(0,1)})
\bigl(
|0\rangle\otimes|\Lambda\rangle
\bigr).$
Equivalently,
\[
|\omega\rangle
\propto
\frac{\mathbb{I}+P_{(0,1)}}{2}
\bigl(
|0\rangle\otimes|\Lambda\rangle
\bigr).
\]
\end{proof}

\begin{lemma}
\label{lem:recovery-probability}
Let $\omega$ be a normalized state such that
\[
|\omega\rangle
\propto
\frac{\mathbb{I}+P_{(0,1)}}{2}
\bigl(
|0\rangle\otimes|\Lambda\rangle
\bigr).
\]
Then
\[
\Tr\!\left(
\Pi^-_{1,\ldots,k}
|\omega\rangle\!\langle\omega|
\right)
=
\frac{k+1}{2k}.
\]
Equivalently, conditioned on a successful SWAP test, the probability that the
resulting \(k\)-qudit subsystem lies in the fully antisymmetric subspace is $(k+1)/2k.$

Then, by Subsection~\ref{subsec:reducingtorepstates} and
Lemma~\ref{Lemma:SWAPSuccess21k-1}, the same statement holds for an
arbitrary state in the \((2,1^{k-1})\) isotypic component conditioned on
qudits \(0\) and \(1\) being symmetric: the probability that an
antisymmetric measurement on qudits \(1,\ldots,k\) projects onto the
fully antisymmetric subspace is $(k+1)/2k.$
\end{lemma}

\begin{proof}
    Let $\ket{w} \propto \frac{\mathbb I + P_{(0,1)}}{2}\ket{0} \otimes \ket{\Lambda}$. Note that 
    \begin{align}
        \mathcal{T}(w) = \mathcal{T}(\ket{0}\bra{0}\otimes \Lambda). 
    \end{align}
    Now,
    \begin{align}
      \tr \left(\Pi^-_{0,\cdots, k-1} \mathcal{T}(w) \right) &= \tr\left(\Pi^-_{0,\cdots, k-1} \mathcal{T}(\ket{0}\bra{0}\otimes \Lambda)\right)\notag\\
                                            &= \frac{1}{k+1}\cdot 1 + \frac{k}{k+1}\tr\left( \Pi^-_{0,\cdots, k-1} \ket{0}\bra{0}\otimes \Lambda \right)\notag\\
        &=\frac{1}{k+1} \cdot 1 + \frac{k}{k+1}\cdot\frac{1}{k^2} = \frac{1}{k}.        \label{eq:14}
    \end{align}
    The factor \(1/k^2\) follows from the overlap,
\[
\left\|
\Pi^-_{0,\ldots,k-1}
\bigl(
|0\rangle\otimes|\Lambda\rangle
\bigr)
\right\|^2
=
\frac{1}{k^2}.
\]
Indeed, substituting the expansion,
\[
|\Lambda\rangle
=
\frac{1}{\sqrt{k!}}
\sum_{\pi\in S_k}
\operatorname{sgn}(\pi)
|\pi(0),\ldots,\pi(k-1)\rangle,
\]
only the terms for which the second copy of the label \(0\) is not among the
projected registers contribute to the antisymmetric projection. These surviving
terms have total overlap \(1/k\) with the normalized antisymmetric state, and
therefore the squared overlap is \(1/k^2\).\\

Furthermore,
\begin{align}\label{eq:12}
  &\tr\left(\Pi^-_{0,\cdots, k-1} \mathcal{T}(w) \right) = \frac{2}{k+1}\tr\left( \Pi^-_{1,\cdots, k} w \right).
\end{align}
To understand this identity, interpret \(\mathcal{T}(\omega)\) as applying a
uniformly random permutation to \(\omega\). The state \(\omega\) is supported on
  the symmetric subspace of registers \(1\) and \(2\). If in a random permutation of \(\omega\), the symmetric pair is part of  \(0,\ldots,k-1\), then the product with the antisymmetric projection, \(\Pi^-_{0,\cdots, k-1}\), yields zero. The only nonzero contribution occurs when exactly one of the two registers
in the symmetric pair is sent to the register \(k\); this happens with
probability $2/(k+1).$
In either of these two cases, by relabeling registers, the contribution is $\Tr\!\left(\Pi^-_{1,\ldots,k}\omega \right)$. Combining Eq.~\eqref{eq:12} and Eq.~\eqref{eq:14} we have, 
\begin{align}
  \tr\left( \Pi^-_{1,\cdots, k} w \right) = \frac{k+1}{2k}.
\end{align}
\end{proof}

\newpage

\section{Probability Facts}\label{App:probfacts}
\subsection{Comparison with geometric waiting times}
\begin{theorem}\label{Thm:Domination}
Let \(X_1,\dots,X_n\) be nonnegative integer-valued random variables.
Suppose that, for every \(i\in\{1,\dots,n\}\), every realization
\(x_1,\dots,x_{i-1}\) satisfying \(x_j\le R_j\) for all \(j<i\),
and every \(r\ge0\),
\begin{align}\label{Eqn:conditionalDomination}
\Pr\left(
X_i>r
\mid
X_1=x_1,\dots,X_{i-1}=x_{i-1}
\right)
\le
\Pr\left(Z_i>r\right).
\end{align}
Let \(Z_1,\dots,Z_n\) be completely independent non-negative integer-valued random variables,
independent of \(X_1,\dots,X_n\). Then, for every \(T\ge0\),
\begin{align}\label{Eqn:probcondition}
\Pr\left(
\sum_{i=1}^n X_i>T,\,
X_1\le R_1,\dots,X_n\le R_n
\right)
\le
\Pr\left(
\sum_{i=1}^n Z_i>T
\right).
\end{align}
\end{theorem}

\begin{proof}

We prove the claim by iteratively replacing each \(X_i\) with the corresponding geometric random variable \(Z_i\).
We now replace the random variables one at a time. First,
\begin{align}
&\Pr\left(\sum_{k=1}^n X_k > T,X_i \le R_i \quad \forall i \le n\right) \notag\\
&=
\sum_{s=0}^{\infty}
\Pr\left(\sum_{k=1}^{n-1} X_k = s,X_i \le R_i \quad \forall i \le n-1\right) \times
\Pr\left(X_n > T-s, X_n \le R_n \mid \sum_{k=1}^{n-1} X_k = s,X_i \le R_i \quad \forall i \le n-1 \right)\notag\\
&\le
\sum_{s=0}^{\infty}
\Pr\left(\sum_{k=1}^{n-1} X_k = s,X_i \le R_i \quad \forall i \le n-1\right) \times
\Pr\left(X_n > T-s \mid \sum_{k=1}^{n-1} X_k = s,X_i \le R_i \quad \forall i \le n-1 \right)\notag
\intertext{Since the domination bound holds for every admissible realization of the past, it continues to hold after averaging over all admissible pasts with total sum \(s\).
Using the bound in Eq.~\eqref{Eqn:conditionalDomination}, and noting that $Z_n$ is independent of all the $X$'s and $Z$'s, }
&\le
\sum_{s=0}^{\infty}
\Pr\left(\sum_{k=1}^{n-1} X_k = s,X_i \le R_i \quad \forall i \le n-1\right) \times
\Pr\left(Z_n > T-s \mid \sum_{k=1}^{n-1} X_k = s,X_i \le R_i \quad \forall i \le n-1 \right)
\\
&=
\Pr\left(\sum_{k=1}^{n-1} X_k+Z_n > T ,X_i \le R_i \quad \forall i \le n-1 \right).
\end{align}

We use similar arguments to replace $X_{n-1}$ with $Z_{n-1}$ in the following way:
\begin{align}
  &\Pr\left(\sum_{k=1}^{n-1} X_k+Z_n > T ,X_i \le R_i \quad \forall i \le n-1 \right)\notag\\
  &=\sum_{s=0}^{\infty}
\Pr\left(\sum_{k=1}^{n-2} X_k + Z_n= s,X_i \le R_i \quad \forall i \le n-2\right)\notag\\
&\times\Pr\bigg(X_{n-1} > T-s, X_{n-1} \le R_{n-1} \mid \sum_{k=1}^{n-2} X_k + Z_n = s,X_i \le R_i \quad \forall i \le n-2 \bigg)
\notag\\
&\le
\sum_{s=0}^{\infty}
\Pr\left(\sum_{k=1}^{n-2} X_k + Z_n= s,X_i \le R_i \quad \forall i \le n-2\right) \times
\Pr\left(X_{n-1} > T-s\mid \sum_{k=1}^{n-2} X_k + Z_n = s,X_i \le R_i \quad \forall i \le n-2 \right).\notag
\end{align}
Once more, we use the bound in Eq.~\eqref{Eqn:conditionalDomination}, and noting that $Z_{n-1}$ is independent of all the $X$'s and other $Z$'s,
\begin{align}
&\le
\sum_{s=0}^{\infty}
\Pr\left(\sum_{k=1}^{n-2} X_k + Z_n = s,X_i \le R_i \quad \forall i \le n-2\right)\times
\Pr\left(Z_{n-1} > T-s \mid \sum_{k=1}^{n-2} X_k +Z_n = s,X_i \le R_i \quad \forall i \le n-2 \right)
\notag\\
&=
\Pr\left(\sum_{k=1}^{n-2} X_k+Z_{n-1}+Z_n > T ,X_i \le R_i \quad \forall i \le n-2 \right).
\end{align}

Repeating the same argument successively for \(X_{n-2},\dots,X_1\),
\begin{align}
\Pr\left(\sum_{k=1}^n X_k > T,X_1 \le R_1,\cdots, X_n \le R_n\right)
\le
\Pr\left(\sum_{k=1}^n Z_k > T\right).
\end{align}
\end{proof}

\subsection{Concentration of Geometric Random variables}
\begin{theorem}\label{Thm:GeometricConcentration}
Let \(Z_1,\ldots,Z_k\) be independent geometric random variables with
success parameters \(p_1,\ldots,p_k \). Define
\begin{align}
T_* =\sum_{i=1}^k \frac{1}{p_i},
\qquad
p^* =\min_{1\le i\le k}p_i.
\end{align}
Then, for every \(0<\epsilon<1\), setting
\begin{align}  \label{Eqn:T_nj bound}
T(\epsilon)
=
2T_* +\frac{2}{p^*}\ln\frac{1}{\epsilon},
\end{align}
gives,
\begin{align}
\Pr\left(
\sum_{i=1}^k Z_i>T(\epsilon)
\right)
\le \epsilon.
\end{align}
\end{theorem}

\begin{proof}
We prove Theorem~\ref{Thm:GeometricConcentration} using the following
concentration bound for sums of independent geometric random variables,
due to Janson. Note that
\begin{align}
T_*=\sum_{i=1}^k \mathbb{E}[Z_i]
\end{align}
is the expected value of their sum.

\begin{lemma} (Theorem 2.1 of \cite{jansontail})
Let $Z_1,\cdots, Z_k$ be independent but not necessarily identically distributed geometric random variables with the success parameters $p_1,\cdots,p_k$. Then, for $T \ge \mu$,
\begin{align}
    \operatorname{Pr}(Z_1+\cdots+Z_k \geq T) \leq \exp{\left(-p^*\mu\left(\frac{T}{\mu}-\ln{\left(\frac{T}{\mu}\right)-1}\right)\right)}, \label{jansonconc}
\end{align}
where $p^*=\min_{1\le i \le k} p_i$ and $\mu= \sum_{i=1}^k \frac{1}{p_i}$.
\end{lemma}

    We find a lower bound on $T \ge T_*$ that guarantees the right-hand side of Eq.~\eqref{jansonconc} is bounded by $\epsilon$, i.e.,
\be
\exp\bigg(-p^*T_*\bigg(\frac{T+1}{T_*}-\ln\bigg(\frac{T+1}{T_*}\bigg)-1\bigg)\bigg)\le \epsilon\ ,
\ee
which means,
\begin{align}\label{d4}
    \frac{T+1}{T_*}-\ln{\left(\frac{T+1}{T_*}\right)}\geq \frac{1}{T_*}\left( T_*+\frac{1}{p^*}\ln\frac{1}{\epsilon}\right).
\end{align}
 Next, we use the following property:
\be \label{x-lnx ineq}
 \forall a,x>1:\  
 x= 2a \ \ \Longrightarrow\ \ x-\ln{x} \geq a.
\ee
Substituting,
\begin{align}
    x=\frac{T+1}{T_*},\quad a=\frac{1}{T_*}\left( T_*+\frac{1}{p^*}\ln\frac{1}{\epsilon}\right),
\end{align}
 into $x= 2a$ and multiplying both sides by $T_*$, we obtain the following chain of inequalities, 
\begin{align} 
    T+1 &= 2T_*\left(\frac{1}{T_*}\left( T_*+\frac{1}{p^*}\ln\frac{1}{\epsilon}\right)\right) = 2T_*+\frac{2}{p^*}\ln\frac{1}{\epsilon} \ ,
\end{align}
Then, Eq.(\ref{x-lnx ineq}) implies that for 
\begin{align}
    T= 2T_*+\frac{2}{p^*}\ln\frac{1}{\epsilon}-1.
\end{align} 
condition in Eq. \eqref{d4} is satisfied.

\end{proof}

\end{document}